\documentclass[11pt]{article}

\usepackage{amsmath,amssymb,amsfonts,amsthm}
\usepackage{bbm}        \usepackage{nicefrac}      

\usepackage{booktabs}

\usepackage{graphicx}
\usepackage{float}
\usepackage[caption=false]{subfig}

\usepackage[margin=1in]{geometry}
\usepackage{changepage}

\usepackage{enumitem}
\usepackage{comment}
\usepackage{thm-restate}

\usepackage[ruled,vlined,linesnumbered]{algorithm2e}

\usepackage{natbib}
\usepackage{csquotes}

\usepackage{tikz}
\usepackage{pgfplots}
\pgfplotsset{compat=1.18}
\usetikzlibrary{patterns,intersections,arrows,decorations.markings}
\usepgfplotslibrary{fillbetween}

\tikzstyle{vecArrow} = [thick, decoration={markings,mark=at position
   1 with {\arrow[semithick]{open triangle 60}}},
   double distance=1.4pt, shorten >= 5.5pt,
   preaction = {decorate},
   postaction = {draw,line width=1.4pt, white,shorten >= 4.5pt}]
\tikzstyle{innerWhite} = [semithick, white,line width=1.4pt, shorten >= 4.5pt]

\theoremstyle{plain}
\newtheorem{theorem}{Theorem}[section]
\newtheorem{lemma}[theorem]{Lemma}
\newtheorem{claim}[theorem]{Claim}
\newtheorem{fact}[theorem]{Fact}
\newtheorem{proposition}[theorem]{Proposition}
\newtheorem{corollary}[theorem]{Corollary}

\theoremstyle{plain}
\newtheorem{definition}{Definition}[section]
\newtheorem{example}[definition]{Example}

\theoremstyle{plain}
\newtheorem{assumption}{Assumption}

\usepackage{hyperref}
\usepackage{cleveref}

\crefname{claim}{claim}{claims}
\Crefname{algocf}{Algorithm}{Algorithms}

\crefname{appendix}{appendix}{appendices}
\Crefname{appendix}{Appendix}{Appendices}

\let\oldappendix\appendix
\renewcommand{\appendix}{\oldappendix
  \crefalias{section}{appendix}}



\newcommand{\xhdr}[1]{\vspace{2mm}\noindent{\bf #1}}

\newcommand{\noaccents}[1]{#1}
\newcommand{\newagentvar}[3][\noaccents]{\expandafter\newcommand\expandafter{\csname #2\endcsname}[1][]{{#1{#3}_{##1}}}\expandafter\newcommand\expandafter{\csname #2s\endcsname}{#1{\boldsymbol{#3}}}\expandafter\newcommand\expandafter{\csname #2OPT\endcsname}[1][]{#1{#3}_{##1}^*}\expandafter\newcommand\expandafter{\csname #2Prime\endcsname}{#1{#3}'}\expandafter\newcommand\expandafter{\csname #2DoublePrime\endcsname}{#1{#3}''}\expandafter\newcommand\expandafter{\csname #2Plus\endcsname}{#1{#3}^{+}}\expandafter\newcommand\expandafter{\csname #2Minus\endcsname}{#1{#3}^{-}}\expandafter\newcommand\expandafter{\csname #2Dag\endcsname}[1][]{#1{#3}_{##1}^{\dagger}}\expandafter\newcommand\expandafter{\csname #2DoubleDag\endcsname}[1][]{#1{#3}_{##1}^{\ddagger}}\expandafter\newcommand\expandafter{\csname #2Zero\endcsname}{#1{#3}^{0}}}

\newcommand{\eps}{\varepsilon}
\newcommand{\smallEps}{\varepsilon}

\newcommand{\OPT}{\textsc{Opt}}

\newcommand{\cost}{c}

\newcommand{\sellerprice}{\tilde{\price}}

\newcommand{\dist}{F}
\newcommand{\CDF}{\dist}
\newcommand{\PDF}{f}
\newcommand{\buyerdist}{F}

\newcommand{\sellerdist}{G}

\newcommand{\buyercdf}{\buyerdist}
\newcommand{\buyerpdf}{f}

\newcommand{\sellerpdf}{g}

\newcommand{\buyerhazardrate}{\phi}
\newcommand{\virtualval}{\psi}

\newcommand{\virtualcost}{\phi}

\newcommand{\numAgents}{n}

\newcommand{\valFunc}{v}
\newcommand{\valFuncPrime}{v'}

\newagentvar{val}{v}
\newagentvar{quant}{q}
\newcommand{\qquant}{t}
\newcommand{\quantM}{\quant[m]}

\newagentvar{alloc}{x}
\newagentvar{accAlloc}{X}
\newcommand{\accAllocUb}{\hat{X}}
\newcommand{\allocUb}{\hat{x}}

\newagentvar{payme}{p}
\newagentvar{price}{p}
\newagentvar{bid}{b}
\newagentvar{surpp}{u}
\newagentvar{distt}{F}

\newcommand{\bidspace}{\mathcal{B}}

\newagentvar{dualV}{\lambda}
\newagentvar{revCurve}{R}

\newcommand{\lagCurve}[1][\dualV]{\revCurve_{#1}}
\newcommand{\lagCurveCH}[1][\dualV]{\bar{\revCurve}_{#1}}
\newcommand{\lagCurvePrime}[1][\dualV]{\revCurve'_{#1}}
\newcommand{\lagCurveDoublePrime}[1][\dualV]{\revCurve''_{#1}}
\newcommand{\lagCurveCHPrime}[1][\dualV]{(\bar{\revCurve}_{#1})'}
\newcommand{\lagCurveCHDoublePrime}[1][\dualV]{(\bar{\revCurve}_{#1})''}

\newagentvar{quantItv}{\mathcal Q}
\newcommand{\leftS}{\ell}
\newcommand{\rightS}{r}
\newcommand{\subLim}{m}

\newcommand{\mechFam}{\mathfrak{M}}

\newcommand{\mechfam}{\mathfrak{M}}
\newcommand{\mech}[1][]{\mathcal{M}_{#1}}

\newcommand{\largeN}{K}
\newcommand{\valRatio}{\alpha}
\newcommand{\revApproxL}{0.8246}
\newcommand{\revApproxU}{0.8532}

\newcommand{\ironedItvs}{\mathcal I}
\newcommand{\residuItv}{\mathcal J}
\newagentvar{evol}{\nu}
\newcommand{\mixprob}{{\evol}}
\newcommand{\interpol}{{\evol}}

\newcommand{\hval}{{\val}_H}

\newcommand{\lval}{{\val}_L}

\newcommand{\buyerutil}{U}
\newcommand{\sellerutil}{\Pi}

\newcommand{\sellerbenchmark}{\Pi^*}
\newcommand{\buyerbenchmark}{U^*}
\newcommand{\buyerutilmid}{U^+}
\newcommand{\sellerbenchmarkDist}{\Pi_{\buyerdist}^*}
\newcommand{\buyerbenchmarkDist}{U_{\buyerdist}^*}

\usepackage{ifthen}

\newcommand{\GFT}[2][]{\textsc{Wel}\ifthenelse{\not\equal{}{#1}}{_{#1}}{}\!\left({\def\givenn{\middle|}#2}\right)}

\newcommand{\buyerexanteutil}[2][]{U\ifthenelse{\not\equal{}{#1}}{_{#1}}{}\!\left({\def\givenn{\middle|}#2}\right)}
\newcommand{\sellerexanteutil}[2][]{\Pi\ifthenelse{\not\equal{}{#1}}{_{#1}}{}\!\left({\def\givenn{\middle|}#2}\right)}

\newcommand{\reserve}{r}
\newcommand{\optreserve}{\reserve^*}

\newcommand{\constantH}{K}

\newcommand{\threshold}{t}
\newcommand{\winnerval}{\val_{\rm win}}

\newcommand{\multiUnitRatio}{\Psi}

\newcommand{\Mye}{\sf Mye}

\newcommand{\BuyerOffer}{\sf Buyer-Offer Mechanism}
\newcommand{\BOM}{\sf BOM}
\newcommand{\SellerOffer}{\sf Seller-Offer Mechanism}
\newcommand{\SOM}{\sf SOM}

\newcommand{\VCGAuction}{\sf VCG Auction}

\newcommand{\OPTSB}{\OPT_{\rm SB}}
\newcommand{\OPTFB}{\OPT_{\rm FB}}

\newcommand{\ksfairness}{KS-fairness}

\newcommand{\KSsolution}{{\sf KS Solution}}
\newcommand{\Nashsolution}{{\sf Nash Solution}}

\newcommand{\CSNashsolution}{{\sf Cross-Side Nash Solution}}

\newcommand{\indirectSPA}{{\sf indirect Vickrey auction}}
\newcommand{\indirectSPAs}{{\sf indirect Vickrey auctions}}

\newcommand{\primed}{^\dagger}
\newcommand{\doubleprimed}{^\ddagger}
\newcommand{\opted}{^*}

\newcommand{\allocLow}{\alloc_l}
\newcommand{\allocHigh}{\alloc_h}

\newcommand{\SPA}{{\sf SPA}}
\newcommand{\imbalanceness}{\tau}

\newcommand{\auxdist}{T}

\newcommand{\prefixratio}{\xi}
\newcommand{\prefixratioTkn}{\prefixratio_{\auxdist}(k,n)}

\newcommand{\varx}{a}
\newcommand{\varyy}{b}
\newcommand{\varA}{A}
\newcommand{\varB}{B}

\newcommand{\revratio}{\alpha}

\newcommand{\VCG}{{\sf VCG}}

\newcommand{\AutoAdjust}[3]{\mathchoice{\left #1 #2 \right #3}{#1 #2 #3}{#1 #2 #3}{#1 #2 #3}}
\newcommand{\inParentheses}[1]{\AutoAdjust{(}{#1}{)}}
\newcommand{\inBrackets}[1]{\AutoAdjust{[}{#1}{]}}

\newcommand{\condition}{\,\mid\,}
\newcommand{\suchthat}{\,:\,}

\DeclareMathOperator{\argmax}{argmax}

\newcommand{\subjectto}{\mathrm{s.t.}\ }

\newcommand{\prob}[2][]{\text{Pr}\ifthenelse{\not\equal{}{#1}}{_{#1}}{}\!\left[{\def\givenn{\middle|}#2}\right]}
\newcommand{\expect}[2][]{\mathbb{E}\ifthenelse{\not\equal{}{#1}}{_{#1}}{}\!\left[{\def\givenn{\middle|}#2}\right]}

\newcommand{\tparen}{\big}
\newcommand{\tprob}[2][]{\text{Pr}\ifthenelse{\not\equal{}{#1}}{_{#1}}{}\tparen[{\def\given{\tparen|}#2}\tparen]}
\newcommand{\texpect}[2][]{\mathbb{E}\ifthenelse{\not\equal{}{#1}}{_{#1}}{}\tparen[{\def\given{\tparen|}#2}\tparen]}

\newcommand{\sprob}[2][]{\text{Pr}\ifthenelse{\not\equal{}{#1}}{_{#1}}{}[#2]}
\newcommand{\sexpect}[2][]{\mathbb{E}\ifthenelse{\not\equal{}{#1}}{_{#1}}{}[#2]}

\newcommand{\dd}{\ \mathrm d}

\newcommand{\reals}{\mathbb{R}}

\newcommand{\supp}{{\sf supp}}

\newcommand{\indicator}[1]{\mathbbm{1}\left\{#1\right\}} 
\title{Pareto-Efficient Multi-Buyer Mechanisms:\\
Characterization, Fairness and Welfare}

\author{Moshe Babaioff\thanks{Hebrew University of Jerusalem. Email: {\tt moshe.babaioff@mail.huji.ac.il}} 
\and Sijin Chen\thanks{Peking University. Email: {\tt csj20040330@stu.pku.edu.cn}}
\and Zhaohua Chen\thanks{Peking University. Email: {\tt chenzhaohua@pku.edu.cn}}
\and Yiding Feng\thanks{Hong Kong University of Science and Technology. Email: {\tt ydfeng@ust.hk}}}
\date{}

\begin{document}

\maketitle
\begin{abstract}
    
A truthful mechanism for a Bayesian single-item auction results with some ex-ante revenue for the seller, and some ex-ante total surplus for the buyers. We study the Pareto frontier of the set of seller-buyers ex-ante utilities, generated by all truthful mechanisms when buyers values are sampled independently and identically (i.i.d.). We first provide a complete structural characterization of the Pareto frontier under natural distributional assumptions. For example, when valuations are drawn i.i.d.\ from a distribution that is both regular and anti-MHR, every Pareto-optimal mechanism is a second-price auction with a reserve no larger than the monopoly reserve.

Building on this, we interpret the problem of picking a mechanism as a two-sided bargaining game, and analyze two canonical Pareto-optimal solutions from cooperative bargaining theory: the \emph{Kalai-Smorodinsky (KS) solution}, and the \emph{Nash solution}. We prove that when values are drawn i.i.d. from a distribution that is both regular and anti-MHR, in large markets both solutions yield near-optimal welfare. In contrast, under worst-case MHR distributions, their performance diverges sharply: the KS solution guarantees one-half of the optimal welfare, while the Nash solution might only achieve an arbitrarily small fraction of it. These results highlight the sensitivity of fairness-efficiency tradeoffs to distributional structure, and affirm the KS solution as the more robust notion of fairness for asymmetric two-sided markets. \end{abstract}

\thispagestyle{empty}
\newpage

\section{Introduction}
\label{sec:intro}

The theory of Bayesian mechanism design has achieved remarkable success in the single-item multi-buyer setting,
when
optimizing a single objective: whether maximizing social welfare \citep{vic-61}, seller revenue \citep{mye-81}, or buyers' 
surplus \citep{HR-08}. These foundational results yield elegant---and often simple---mechanisms, such as the second-price auction, second-price auction with a monopoly reserve, and uniform lottery, which are
optimal under standard distributional assumptions like regularity or monotone hazard rate (MHR). For these objectives, the designer's goal is unambiguous, and the resulting mechanisms are celebrated for their simplicity and robustness.

Yet many modern marketplaces are not zero-sum games between a monopolist and passive buyers. Instead, they are two-sided platforms---such as online advertising exchanges, cloud resource markets, ride-sharing services, or labor platforms---where both the seller and the collective of buyers are strategic stakeholders whose long-term participation hinges on perceived fairness. A mechanism that maximizes the seller's revenue may leave buyers with negligible surplus, discouraging future engagement; conversely, a welfare-maximizing mechanism might allocate the item for free, yielding zero revenue and undermining platform sustainability. In such settings, optimizing a single objective is insufficient. What is needed is a balanced approach that respects the interests of both sides.

This has motivated the study of \emph{bi-objective mechanism design}, where the goal is to characterize the Pareto frontier of achievable ex-ante payoff pairs---namely, (seller revenue, buyers' surplus) pairs. While this frontier has been partially analyzed in the single-item, single-buyer setting \citep{CHP-23,BFM-25}, the multi-buyer regime---where competition among buyers fundamentally reshapes the structure of optimal mechanisms---remains largely uncharted.
We thus ask:

\begin{displayquote}
\textbf{(Question 1)} 
\emph{For the single-item, multi-buyer setting, characterize the structure of the Pareto frontier of achievable ex-ante payoff pairs: Which mechanism are Pareto-optimal? Are they simple? Do they need to randomize?}
\end{displayquote}
\noindent
In a nutshell, we provide a complete structural characterization of this frontier under natural distributional assumptions.

Of course, characterizing the Pareto frontier does not tell us which mechanism to implement. Among the continuum of Pareto-optimal options, we seek one that balances the interests of both sides in a principled way. To this end, we interpret the selection problem as a cooperative bargaining game \citep{nas-50} between the seller and the buyers, and study two canonical solution concepts from bargaining theory---the {\sf Kalai-Smorodinsky (KS) Solution} \citep{KS-75} and the {\Nashsolution} \citep{nas-50}---which formalize distinct notions of fairness. The social welfare highly depends on the chosen Pareto-optimal mechanism.
This leads to our second question:

\begin{displayquote}
\textbf{(Question 2)} 
\emph{In the single-item, multi-buyer setting, what are the welfare guarantees of the {\KSsolution} and the {\Nashsolution}, particularly in large markets?}
\end{displayquote}
\noindent
Our work shows that the welfare performance of {\KSsolution} and {\Nashsolution} depends critically on the underlying distribution, with {\KSsolution} exhibiting greater robustness than {\Nashsolution}.

\subsection{Our Contributions and Techniques}
In this work, we provide comprehensive answers to both questions. Below, we outline our main contributions and techniques. Throughout most of our analysis, we consider a Bayesian setting in which a seller offers a single item to $n$ buyers whose valuations are drawn i.i.d.\ from a common distribution $\buyerdist$. We focus on three standard distributional assumptions: regularity, monotone hazard rate (MHR), and anti-monotone hazard rate (anti-MHR). See \Cref{fig:dist-hierarchy} for an illustration of their relationships, examples and \Cref{def:dist} for formal definitions.

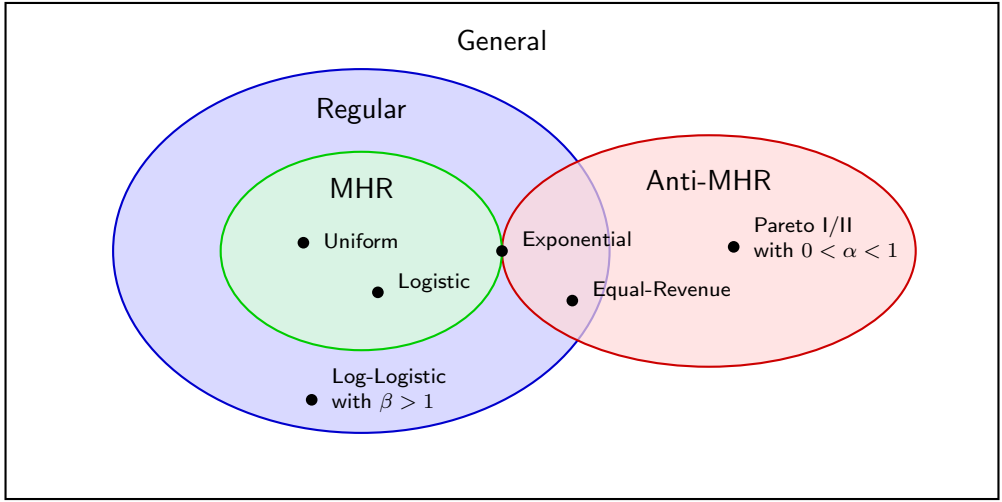
\begin{figure}
    \centering
    \begin{tikzpicture}[scale=1.1]

\draw[thick] (-6,-3) rectangle (6,3);
  \node[below] at (0,2.8) {\textsf{General}};

\fill[blue!15] (-1.7,0) ellipse (3.0 and 2.2);
  \draw[thick, blue!80!black] (-1.7,0) ellipse (3.0 and 2.2);
  \node[above] at (-1.7, 1.42) {\textsf{Regular}};

\fill[red!15, fill opacity=0.7] (2.5,0) ellipse (2.5 and 1.4);
  \draw[thick, red!80!black] (2.5,0) ellipse (2.5 and 1.4);
  \node[above] at (2.5, 0.62) {\textsf{Anti-MHR}};

\fill[green!15, fill opacity=0.7] (-1.7,0) ellipse (1.7 and 1.2);
  \draw[thick, green!80!black] (-1.7,0) ellipse (1.7 and 1.2);
  \node[above] at (-1.7, 0.52) {\textsf{MHR}};

  \fill (0,0) circle (2pt);
  \node[font=\scriptsize, anchor=west] at (0.12,0.12) {\textsf{Exponential}};

  \fill (-2.4,0.1) circle (2pt);
  \node[font=\scriptsize, anchor=west] at (-2.28,0.12) {\textsf{Uniform}};

  \fill (-1.5,-0.5) circle (2pt);
  \node[font=\scriptsize, anchor=west] at (-1.38,-0.38) {\textsf{Logistic}};

  \fill (0.85,-0.6) circle (2pt);
  \node[font=\scriptsize, anchor=west] at (0.97,-0.48) {\textsf{Equal-Revenue}};

  \fill (-2.3,-1.8) circle (2pt);
  \node[font=\scriptsize, anchor=west, align=left] at (-2.18,-1.68) {\textsf{Log-Logistic} \\ \textsf{with $\beta > 1$}};

  \fill (2.8,0.05) circle (2pt);
  \node[font=\scriptsize, anchor=west, align=left] at (2.92,0.17) {\textsf{Pareto I/II} \\ \textsf{with $0 < \alpha < 1$}};
\end{tikzpicture}     \caption{Distribution hierarchy. MHR distributions and anti-MHR distributions intersect at exponential distributions.}
    \label{fig:dist-hierarchy}
\end{figure}

\xhdr{Complete characterization of the Pareto frontier (\Cref{sec:Pareto frontier}).}
We characterize the mechanisms that trace out the entire Pareto frontier of achievable ex-ante payoff pairs. 

In the single-buyer setting, every Pareto-optimal mechanism is a randomization over at most two posted prices (\Cref{prop:Pareto frontier:single buyer general}); when the valuation distribution is regular, it simplifies further to a single posted price (\Cref{prop:Pareto frontier:single buyer regular}).

{In contrast, the multi-buyer setting is significantly more complex due to strategic competition among buyers. Our characterization reveals that every \emph{extreme point} on the Pareto frontier can be implemented by a rather simple mechanism---one that belongs to the family of {\indirectSPAs} \citep{HR-08}, which generalizes the standard second-price auction by 
allowing to restrict the bidding space, and allowing to add a reserve price.
Depending on the properties of the valuation distribution, we fully describe how these implementing mechanisms vary along the Pareto frontier, and establish monotonicity properties of their structure as the buyers' surplus increases.}

\begin{itemize}[leftmargin=*]
    \item \textbf{Regular $\cap$ anti-MHR:} Every point on the Pareto frontier can be implemented by a second-price auction with a reserve no larger than the lowest monopoly reserve. As buyers' surplus increases along the frontier, the reserve decreases monotonically. See \Cref{thm:Pareto frontier:multi buyer:regular and antiMHR distribution}.
    
    \item \textbf{Anti-MHR (not necessarily regular):} Every \emph{extreme point} on the Pareto frontier is implemented by an {\indirectSPA} with a reserve no larger than the monopoly reserve. As buyers' surplus increases, the feasible bid space \emph{expands} and the reserve decreases. See \Cref{thm:Pareto frontier:multi buyer:antiMHR distribution}.
    
    \item \textbf{MHR:} The structure of Pareto-optimal mechanisms changes at a threshold equal to the buyers' surplus achieved by the standard second-price auction: (See \Cref{thm:Pareto frontier:multi buyer:MHR distribution})
    \begin{itemize}
        \item Below this threshold, every point on the Pareto frontier is implemented by a second-price auction with a reserve no larger than the lowest monopoly reserve. As buyers' surplus increases, the reserve decreases monotonically.
        \item Above this threshold, every \emph{extreme point} on the Pareto frontier is implemented by an {\indirectSPA} with \emph{zero reserve} (i.e., the item is always allocated). In this region, the feasible bid space \emph{shrinks} as buyers' surplus increases. \end{itemize}
\end{itemize}
For the anti-MHR and MHR distributions, once the extreme points of the Pareto frontier are characterized, all remaining points are obtained by randomizing between two appropriately chosen {\indirectSPAs}.

In summary, the structure of Pareto-optimal mechanisms---and how they evolve along the frontier---depends critically on the underlying distributional assumptions, reflecting how different hazard rate conditions induce qualitatively distinct competitive dynamics among buyers. Moreover, from a structural perspective, the intersection of regularity and anti-MHR represents a ``sweet spot'': for distributions in this class, the entire Pareto frontier is implemented by simple mechanisms---namely, second-price auctions with a reserve.

\xhdr{Welfare guarantees of {\KSsolution} and {\Nashsolution} (\Cref{sec:ks solution,sec:nash solution}).}
In the second part of the paper, we interpret the selection of a mechanism on the Pareto frontier as a two-sided bargaining problem between the seller and the buyers. Within this framework, we analyze two canonical Pareto-optimal solutions from cooperative bargaining theory: the {\KSsolution} and the {\Nashsolution}.

The {\KSsolution} is the Pareto-optimal mechanism in which both sides achieve the same fraction of their ideal ex-ante payoffs---i.e., the seller attains the same fraction of the maximum possible revenue as the buyers do of the maximum total surplus. (See \Cref{def:KS solution}.) We show that, for any distribution, the {\KSsolution} guarantees at least 50\% of the optimal social welfare. Moreover, this bound is essentially tight: for large markets (i.e., as $n \to \infty$), there exists a worst-case instance with MHR valuation distribution, in which
the {\KSsolution} achieves less than $(50 + \eps)\%$ of the optimum. In contrast, when the valuation distribution lies in the intersection of regularity and anti-MHR, the {\KSsolution} achieves asymptotically optimal welfare as $n \to \infty$.

The {\Nashsolution} is defined as the Pareto-optimal mechanism that maximizes the product of all traders' ex-ante payoffs (seller and all buyers). (See \Cref{def:nash solution}.) For large markets, using the same instance discussed above, we show that the {\Nashsolution} attains less than $\eps$ fraction of the optimal welfare---significantly worse than the {\KSsolution}. However, like the {\KSsolution}, it achieves asymptotically optimal welfare when the distribution is both regular and anti-MHR.

We also consider a variant, the {\CSNashsolution}, which maximizes the product of the seller revenue and the buyers' surplus (rather than individual utilities). (See \Cref{def:cross-side nash solution}.) Its welfare guarantees closely mirror those of the {\KSsolution} across distributional classes.

We summarize all welfare guarantees in \Cref{tab:welfare-guarantees}. Remarkably, echoing our structural findings, the intersection of regularity and anti-MHR again emerges as a ``sweet spot'': for distributions in this class, all three mechanisms---{\KSsolution}, {\Nashsolution}, and {\CSNashsolution}---achieve asymptotically optimal social welfare as $n \to \infty$. 
Moreover, we note that {\KSsolution} is significantly more robust in terms of efficiency: while it guarantees at least 50\% of optimal welfare under arbitrary distributions and degrades gracefully in the worst-case MHR setting, {\Nashsolution} can perform arbitrarily poorly in the same regimes.

Finally, we present two additional results:  
(i) a welfare guarantee for the {\KSsolution} in multi-unit, multi-buyer settings with MHR valuations, and  
(ii) an improved welfare bound for the {\CSNashsolution} in bilateral trade when both traders' distributions are MHR.

\begin{table}[t]
\centering
\caption{Asymptotic welfare guarantees ($n \to \infty$) of Pareto-optimal bargaining solutions.}
\label{tab:welfare-guarantees}
\begin{tabular}{lccc}
\toprule
\textbf{Distribution class} & \textbf{{\KSsolution}} & \textbf{{\Nashsolution}} & \textbf{{\CSNashsolution}} \\
\midrule
General & $\geq \tfrac{1}{2}$ & $o(1)$ & $\geq \tfrac{1}{2}$ \\
MHR (worst-case) & $\leq \tfrac{1}{2} + o(1)$ & $o(1)$ & $\leq \tfrac{1}{2} + o(1)$ \\
Regular $\cap$ anti-MHR & $1 - o(1)$ & $1 - o(1)$ & $1 - o(1)$ \\
\bottomrule
\end{tabular}
\end{table}

\xhdr{Technique: revenue maximization subject to buyers' surplus constraint.} 
From a technical standpoint, we develop a unified framework to characterize the Pareto frontier across all three distributional regimes. The key insight is to reduce the characterization problem to a constrained optimization problem: maximizing seller revenue subject to a lower bound constraint on buyers' surplus. We prove that this reduction is exact for any target buyers' surplus lying between the surplus achieved by the seller-revenue-optimal mechanism (i.e., the Myerson auction) and the maximum achievable buyers' surplus. We believe that this buyer-surplus-constrained revenue maximization framework may be of independent interest.

To solve this constrained problem, we employ a Lagrangian duality approach. Specifically, for any non-negative dual variable $\dualV$ associated with the surplus constraint, we construct an {\indirectSPA} $\mech[\dualV]\opted$ that maximizes the Lagrangian objective:
\begin{align*}
\text{Seller Revenue} + \dualV \cdot \text{Buyers' Surplus}.
\end{align*}
Crucially, we show that as $\dualV$ increases continuously from $0$ to $\infty$, the resulting family of mechanisms $\{\mech[\dualV]\opted\}_{\dualV \geq 0}$ traces out all \emph{extreme points} of the Pareto frontier, with buyers' surplus varying \emph{continuously} from that of the Myerson auction to the buyer-optimal benchmark.

For any buyers' surplus lower bound constraint corresponding to an extreme point on the Pareto frontier, we establish the existence of a dual value $\dualV\opted \geq 0$ such that the mechanism $\mech[\dualV\opted]\opted$ satisfies the Karush–Kuhn–Tucker (KKT) conditions---particularly, complementary slackness---for the original constrained problem. Consequently, $\mech[\dualV\opted]\opted$ solves the constrained optimization and is therefore Pareto-optimal.

Moreover, the \emph{continuity of buyers' surplus} in the parameter $\dualV$ enables us to derive monotonicity properties of Pareto-optimal mechanisms along the frontier. By analyzing how $\mech[\dualV]\opted$ evolves with $\dualV$---in terms of reserves and feasible bidding spaces---we obtain all the structural monotonicity results stated in \Cref{sec:Pareto frontier}.

\subsection{Related Work}

\xhdr{Pareto-efficient mechanism design.} A growing body of work studies mechanism design through the lens of Pareto efficiency, seeking mechanisms that balance competing objectives---typically seller revenue and buyer surplus---without leaving gains on the table. Most closely related is the work of \citet{BFM-25}, who initiate the study of Pareto-optimal mechanisms in Bayesian bilateral trade and introduce the Kalai–Smorodinsky (KS) solution as a fairness-refined selection from the Pareto frontier. They establish tight welfare guarantees for KS-fair mechanisms under regularity and MHR assumptions, and contrast them with Nash welfare maximization. Extending this to multi-buyer settings, our work provides the first full characterization of the Pareto frontier and analyzes the robustness of KS fairness in competitive markets. In parallel, \citet{BOT-25} study dynamic rental games where agents are stagewise individually rational; they show that classical Myersonian tools break down, and develop alternative characterizations of Pareto-optimal allocations under various reward structures (e.g., mix between welfare, revenue, buyers' surplus). Their work highlights how temporal constraints reshape Pareto efficiency, but focuses on a single-agent-per-stage model without competition. Finally, \citet{CHP-23} propose an $(\alpha,\theta)$-CSR pricing scheme that interpolates between profit and buyers' surplus. While their framework yields distribution-free bounds on welfare ratios, it only considers pricing mechanism, does not characterize the Pareto frontier or address multi-buyer strategic competition. In contrast, our paper bridges these threads by offering a structural characterization of the entire Pareto frontier among the entire mechanism space in multi-buyer auctions and evaluating canonical {\KSsolution} and {\Nashsolution} as principled selectors of fair and efficient mechanisms. Beyond two-sided fairness, there is also a line of research focused on fairness \emph{among buyers}---for example, ensuring envy-freeness, proportionality, or equitable surplus allocation in multi-buyer settings (see, e.g., \citealp{DP-17,AJ-18,MFGW-20,RG-20,DKA-21,CEL-22,ADK-24,CK-25,WKR-25,FLMMP-25,BCHX-25}).

\xhdr{Cooperative bargaining with incomplete information.} \label{subsubsec:related-bargaining}
The cooperative bargaining framework was introduced by \citet{nas-50}, who proposed the Nash solution as the unique outcome satisfying four axioms: Pareto optimality, symmetry, scale invariance, and independence of irrelevant alternatives. In the same complete-information setting, \citet{KS-75} replaced the independence axiom with resource monotonicity to obtain the Kalai–Smorodinsky (KS) solution, while \citet{Kal-77,Mye-77} introduced the egalitarian solution by substituting scale invariance with resource monotonicity. We refer the reader to \citet{AH-92} for a comprehensive survey.
A separate strand of literature studies bargaining under private information \citep[e.g.,][]{HS-72,Mye-84,Sam-84,KW-93,ACD-02}, where agents possess asymmetric knowledge about their own types. In this setting, researchers have proposed analogues of the classical axioms and corresponding mechanism-based solutions.
While these works focus on the axiomatization and existence of fair outcomes in abstract bargaining problems, our paper takes a concrete mechanism design perspective: we study the single-item, multi-buyer auction with i.i.d.\ private values, fully characterize the mechanisms that implement the Pareto frontier, and analyze the welfare guarantees of two canonical bargaining solutions---the {\KSsolution} and the {\Nashsolution}---which embody distinct notions of fairness between the seller and the buyers.

\section{Preliminaries}
\label{sec:prelim}

In this work, we study the space of Bayesian single-seller, multi-buyer auctions, focusing on the Pareto frontier and the implications of fairness constraints on the welfare. We next present some background.

\xhdr{Environment.} There is a single seller and $n$ unit-demand buyers, whom we collectively refer to as traders (or agents). The seller holds a single item with zero production cost. Buyers have quasi-linear utilities: given an allocation $\alloc \in [0, 1]$ and a monetary transfer $\price \in \reals_+$ that she pays (i.e., her \emph{payment} to the seller), a buyer with private value $\val$ obtains a utility of $\alloc \cdot \val - \price$. Meanwhile, upon receiving a monetary transfer (buyers' payments) of $\price \in \reals_+$, the seller obtains revenue (and utility) equal to~$\price$.

For each buyer $i \in [n]$, her private value $\val_i$ is drawn independently and identically (i.i.d.) from a publicly-known common valuation distribution $\buyerdist$. We assume that the distribution $\buyerdist$ has a continuous support $\supp(\buyerdist)$, 
so for some $\hval\geq \lval\geq 0$ it holds that either ${\supp(\buyerdist)} = [\lval, \hval]$ or ${\supp(\buyerdist)} = [\lval, \infty)$.    
We assume that $\buyerdist$ has no atoms, except possibly at $\hval$. 
Following the convention in auction design literature, we define the cumulative density function (CDF) of the buyer's valuation distribution as $\buyercdf(t)\triangleq \prob[\val\sim\buyercdf]{\val < t}$. 
We assume that it is differentiable within the interior of the support, with measure 1, and denote the corresponding probability density function by $\buyerpdf$.
Consequently, $\buyercdf$ is left-continuous.  Finally, we assume there exists a price $\price \in \supp(\buyerdist)$ that maximizes the expected revenue $\price \cdot (1 - \buyercdf(\price))$.

\xhdr{Mechanisms.} A benevolent social planner (i.e., a mechanism designer) aims to design mechanisms that facilitate trade between the seller and the buyers. Specifically, a mechanism $\mech = (\allocs, \prices)$ first solicits bids from all buyers and then determines the allocation and payments according to an allocation rule $\allocs = (\alloc_1, \dots, \alloc_n)$ and a payment rule $\prices = (\price_1, \dots, \price_n)$. For each buyer $i \in [n]$, her allocation $\alloc_i \colon \reals_+^n \to [0, 1]$ and payment $\price_i \colon \reals_+^n \to \reals_+$ are functions mapping the bid profile of all $n$ buyers to, respectively, the fraction of the item allocated to her and the monetary transfer she makes. Then, the social planner transfers all money received from buyers directly to the seller.\footnote{In the literature, such mechanisms are known as strong budget balanced (SBB) mechanisms, since the social planner does not keep any money. Such a restriction causes no loss on our goal of maximizing welfare, as we will discuss in \Cref{sec:weak budget balance}.}

Under the revelation principle, a standard argument in the literature~\citep{mye-81}, it is without loss of generality to restrict attention to ex-post feasible, Bayesian incentive compatible, and interim individually rational mechanisms. These properties are formally defined as follows. 
\begin{itemize}
    \item \emph{Ex-post feasibility}: the mechanism never allocates more than its supply (of one unit), i.e., for any $\vals$, $\sum_{i \in [n]} \alloc_i(\vals) \leq 1$.
    \item \emph{Bayesian incentive compatibility (BIC)}: given any buyer $i$ and her value $\val_i$, when all other buyers bid truthfully, buyer $i$ maximizes her interim utility by also bidding truthfully. 
    \item \emph{Interim individual rationality (interim IR)}: given any buyer $i$ and her value $\val_i$, her interim utility is non-negative when bidding truthfully. 
\end{itemize}
Specifically, under BIC, we assume that buyers report their private values truthfully and thus interpret the allocation and payment rules $(\allocs, \prices)$ of a given mechanism as functions of the buyers’ true valuation profile $\vals = (\val_1, \dots, \val_n)$, mapping it to their allocations, and payments, respectively.

Finally, since all buyers are ex-ante symmetric (i.e., their values are drawn i.i.d.\ from the same distribution $\buyerdist$), any mechanism can be symmetrized by randomly permuting the buyers' identities. 
Therefore, without loss of generality, we restrict attention to mechanisms that are symmetric across buyers. 
In all, we use $\mechFam$ to denote the family of all mechanisms satisfying the conditions stated above, i.e., ex-post feasibility, BIC, interim IR, and symmetry. 

\xhdr{Welfare maximization.}
An important consideration for the social planner is the economic efficiency of the mechanism, which we measure by its social welfare -- the combined utility of the seller and buyers.
For a given mechanism $\mech \in \mechFam$, the \emph{ex-ante seller revenue} is $\sellerexanteutil[\buyerdist]{\mech} \triangleq \expect{\sum_{i\in[n]}\price_i(\vals)}$, and for each buyer $i \in [n]$, her \emph{ex-ante utility} $\buyerexanteutil[i, \buyerdist]{\mech} \triangleq \expect{\alloc_i(\vals)\cdot\val_i-\price_i(\vals)}$ is the expected utility she receives under $\mech$, where the expectation is taken over the realization of all buyers' valuations from prior $\buyerdist$ and any randomness in the mechanism (under truthful reporting). 
In addition, we define the \emph{ex-ante buyers' surplus} (also called residual surplus) as the sum of all buyers' ex-ante utilities, i.e., $\buyerexanteutil[\buyerdist]{\mech} = \sum_{i \in [n]} \buyerexanteutil[i, \buyerdist]{\mech}$.
Finally, the \emph{(ex-ante) social welfare} (or gains-from-trade) is defined as the sum of the ex-ante seller revenue and the ex-ante buyers' surplus, i.e., 
\begin{align*}
    \GFT[\buyerdist]{\mech} \triangleq \sellerexanteutil[\buyerdist]{\mech} + \buyerexanteutil[\buyerdist]{\mech}.
\end{align*}
We denote the \emph{optimum social welfare benchmark} $\OPT_{\buyerdist}$ by 
\begin{align*}
    \OPT_{\buyerdist} \triangleq \expect[\vals \sim \buyerdist^{n}]{\max\{\val_1, \dots, \val_n\}}.
\end{align*}
Notably, this benchmark is achievable by the classic Vickrey auction \citep{vic-61}, i.e., the second-price auction. 
We also define the optimal (ideal) seller revenue and the optimal (ideal) buyers' surplus as 
\[\sellerbenchmarkDist = \max_{\mech \in \mechFam} \sellerexanteutil[\dist]{\mech}, \quad \buyerbenchmarkDist = \max_{\mech \in \mechFam} \buyerexanteutil[\dist]{\mech}. \]

\xhdr{Pareto optimality and Pareto frontier.}
In addition to maximizing social welfare, another important aspect is to consider how to distribute it between the seller and all buyers.
In this sense, we define the Pareto-optimal mechanisms in $\mechFam$ and the induced Pareto frontier as follows:
\begin{definition}[Pareto domination, Pareto optimality, and Pareto frontier]
    Fix a value distribution $\dist$. We say a mechanism $\mech \in \mechFam$ is \emph{Pareto-dominated} by $\mech\primed \in \mechFam$, if both $\sellerexanteutil[\dist]{\mech\primed} \geq \sellerexanteutil[\dist]{\mech}$ and $\buyerexanteutil[\dist]{\mech\primed} \geq \buyerexanteutil[\dist]{\mech}$ hold, and at least one of them holds with strict inequality. $\mech \in \mechFam$ is \emph{Pareto-optimal}, if no mechanism $\mech\primed \in \mechFam$ Pareto-dominates $\mech$. In this case, we say that the (seller revenue, buyers' surplus) pair $(\sellerexanteutil[\dist]{\mech}, \buyerexanteutil[\dist]{\mech})$ realized by mechanism $\mech$ is on the \emph{Pareto frontier}.
\end{definition}

As a direct consequence of the definition, we remark that there is no point on the Pareto frontier with lower buyers' surplus than that of the seller-revenue-optimal mechanism, i.e., the Myerson auction. In fact, the mechanism that realizes such a point is Pareto-dominated by the Myerson auction. Similarly, there is also no point on the Pareto frontier with lower seller revenue than that of the buyer-surplus-optimal mechanism. We will characterize the Pareto frontier in \Cref{sec:Pareto frontier}. Further, in \Cref{sec:ks solution,sec:nash solution}, we use that to derive implications on the welfare of two standard fairness solution concepts, i.e., Kalai-Somorodinsky fairness, Nash welfare maximization, that are on that Pareto frontier.

\subsection{Additional Mechanism Design Notations and Concepts}
\label{sec:prelim:mechanism design techniques}

Next, we introduce notations and concepts from the mechanism design literature that we use in our analysis. 

\xhdr{Regularity, MHR, and anti-MHR.} 
We first present several prominent classes of distributions.
\begin{definition}[Regularity, MHR, and anti-MHR]
\label{def:dist}
For a valuation distribution $\dist$:
    \begin{itemize}
        \item $\dist$ is called \emph{regular} if its virtual value function $\virtualval(\val) \triangleq \val - \frac{1 - \CDF(\val)}{\PDF(\val)}$ is monotone non-decreasing in $\val$. 
        \item $\dist$ satisfies the \emph{monotone hazard rate (MHR) condition} if its hazard rate function $\virtualcost(\val) \triangleq \frac{\PDF(\val)}{1 - \CDF(\val)}$ is monotonically non-decreasing in $\val$. 
        \item $\dist$ satisfies the \emph{anti-MHR condition} if its hazard rate function is monotonically non-increasing in $\val$.
    \end{itemize}
\end{definition}

The MHR condition is stronger than the regularity condition. Meanwhile, the only intersection between the anti-MHR and MHR distributions is the exponential distribution, whose hazard rate is constant. However, there are some natural distributions that are both regular and anti-MHR besides exponential distributions, e.g., equal revenue distributions, Pareto type I distributions with shape $\alpha \geq 1$, Pareto type II distributions (a.k.a. Lomax distributions) with shape $\alpha \geq 1$, etc.

\xhdr{Quantile and revenue curve.}
We also establish the quantile terminology, which gives an equivalent view of the buyer value, but is more convenient in usage. 
The quantile $\quant$ of a buyer with value $\val$ drawn from distribution $\dist$ is defined as $\quant = 1 - \CDF(\val)$. The value function $\valFunc(\cdot)$ maps a buyer's quantile to her value, i.e., $\valFunc(\quant) \triangleq \inf \{\val \suchthat \CDF(\val) > 1 - \quant\}$. Since $\CDF(\cdot)$ is left-continuous, so is $\valFunc(\cdot)$. 
We further assume that $\valFunc(\cdot)$ is also differentiable in $(0, 1)$, and use $\valFunc'(\cdot)$ to denote its derivative.
We now define the price-posting revenue curve. 
\begin{definition}[Price-posting revenue curve]
    Given value distribution $\dist$, the price-posting revenue curve is defined as $\revCurve(\quant) \triangleq \quant \cdot \valFunc(\quant)$ for every quantile $\quant\in(0, 1]$ and $\revCurve(0)\triangleq \lim_{\quant\to 0}\revCurve(\quant)$. 
\end{definition}
In fact, at each quantile $\quant$, the revenue curve characterizes the seller's revenue under a single-buyer scenario in which the mechanism posts a price $\valFunc(\quant)$, sets it as a take-or-leave-it price, and transfers all buyer payments to the seller. Along the revenue curve, by our assumption, there is at least one quantile $\quant$ such that $\revCurve(\quant)$ reaches its maximum, but there might be many such quantiles.
In this case, we care about the largest such quantile $\quantM$, i.e., 
$\quantM \triangleq \sup \argmax_{\quant} \revCurve(\quant)$.
Under our assumptions, the supremum is obtained, that is, $\revCurve(\quantM) = \max_{\quant \in [0, 1]} \revCurve(\quant)$, 
\footnote{This is because, 
due to our assumption that there exists a price $\price \in \supp(\dist)$ which maximizes $\price \cdot (1 - \dist(\price))$, and that $\quantM = \sup \argmax_{\quant} \revCurve(\quant)$, we derive that $\quantM \geq \dist(\price)$ and $\valFunc(\quantM) \leq \price$. So the price $\valFunc(\quantM)$ is obtained. Since $\revCurve$ is left-continuous, we can further achieve that $\revCurve(\quantM) = \max_{\quant \in [0, 1]} \revCurve(\quant)$ holds, thus the supremum is also obtained.} 
and the corresponding value $\valFunc(\quantM)$ of quantile $\quantM$ is thus the \emph{lowest monopoly reserve}. 
Correspondingly, we call $\quantM$ the quantile of the lowest monopoly reserve quantile.

The following are equivalent interpretations of the conditions of regularity, of MHR, and of anti-MHR, but in the quantile space:
\begin{proposition}[\citealp{BR-89,har-16}]
    The following holds:
    \begin{itemize}
        \item $\dist$ is regular if and only if the price-posting revenue curve $\revCurve(\quant)$ is concave. 
        \item $\dist$ satisfies the MHR condition if and only if $\quant\cdot \valFunc'(\quant)$ is non-increasing in $\quant$.  
        \item $\dist$ satisfies the anti-MHR condition if and only if $\quant\cdot \valFunc'(\quant)$ is non-decreasing in $\quant$. 
    \end{itemize}
\end{proposition}

\xhdr{Indirect Vickrey auctions.}
Finally, we introduce \indirectSPAs, as originally defined in \cite{HR-08}. These auctions extend the Vickrey auction by restricting the bidding space and setting reserves. 

\begin{definition}[Indirect Vickrey auctions]
\label{def:indirect SPA}
    An \emph{\indirectSPA} is parameterized by a bidding space $\bidspace \subseteq \reals_+$ and a uniform reserve $\reserve \in \reals_+$. Buyers are \emph{required} to bid in $\bidspace$, and the buyer $i^*$ with the highest bid wins the item (with ties broken uniformly at random), provided that it is no smaller than the reserve $\reserve$. The winner pays her \emph{critical bid}, defined as the minimum bid she would need to submit to remain the winner. 
\end{definition}

We say the bidding space $\bidspace\primed$ of an {\indirectSPA} $\mech\primed$ is a \emph{refinement} of the bidding space $\bidspace\doubleprimed$ of another {\indirectSPA} $\mech\doubleprimed$, if $\bidspace\primed$ is a superset of $\bidspace\doubleprimed$. Generally speaking, if $\mech\primed$ refines $\mech\doubleprimed$, it allows more abundant bidding strategies for each buyer than the latter.

We remark that although {\indirectSPAs} are not direct-revelation mechanisms---and thus not BIC by definition---they induce dominant strategy equilibria. Specifically, it is a dominant strategy for each buyer to submit the supremum of the bids in the feasible bid space $\bidspace$ that are no larger than her true valuation. By the revelation principle, there exists an equivalent direct-revelation DSIC mechanism that yields the same interim allocation and payment rules. Therefore, with a slight abuse of notation, we treat {\indirectSPAs} as belonging to the mechanism class $\mechFam$.

As shown in \citet{mye-81}, a seller-revenue-optimal mechanism is an {\indirectSPA} with the lowest monopoly reserve $\reserve = \valFunc(\quantM)$.\footnote{When the buyer distribution is regular, this mechanism uses the unconstrained bidding space and the lowest monopoly reserve. Otherwise, the feasible bid space $\bidspace$ corresponds to the set of values that remain fixed under Myerson's ironing procedure.}  
For brevity, we denote this mechanism by $\Mye$, and we have $\sellerexanteutil[\buyerdist]{\Mye} = \sellerbenchmarkDist$.

\section{Characterizing the Pareto Frontier}
\label{sec:Pareto frontier}
In this section, we present structural characterizations of the Pareto frontier of the set of (seller revenue, buyers' surplus) pairs for mechanisms that are ex-post feasible, BIC, interim IR, and symmetric. 

We begin with a warm-up by analyzing the single-buyer market in \Cref{sec:Pareto frontier:single buyer}, then extend our results to the substantially more challenging setting of the multi-buyer market in \Cref{sec:Pareto frontier:multiple buyer}. Subsequently, in \Cref{sec:Pareto frontier:reduction,sec:Pareto frontier:reduction proof example}, we provide a unified framework that can yield all of our structural characterizations.

\subsection{Warm-Up: Single-Buyer Market}
\label{sec:Pareto frontier:single buyer}

We start with the single-buyer market as a warm-up. For several critical points on the Pareto frontier, we know that the seller-optimal, buyer-optimal, and welfare-optimal mechanisms all take the form of posting a single (take-it-or-leave-it) price \citep{mye-81}. In fact, it is not hard to show that any point on the Pareto frontier can be realized by a mechanism that randomizes over at most two posted prices.

\begin{restatable}{proposition}{PFSingleGeneral}
\label{prop:Pareto frontier:single buyer general}
    In a single-item seller, single-buyer market, for any valuation distribution $\dist$, any point on the Pareto frontier can be realized by a mechanism that fixes a probability $\mixprob \in [0, 1]$ ex ante, and then posts a price $\price\primed$ to the buyer with probability $\mixprob$ and posts a price $\price\doubleprimed$ to the buyer with probability $1 - \mixprob$. Both these prices are no larger than the lowest monopoly reserve to the buyer. If the buyer accepts, she receives the item and pays the corresponding price to the seller. 
\end{restatable}

To obtain \Cref{prop:Pareto frontier:single buyer general}, we observe that in the single buyer case, any mechanism in $\mechFam$ is equivalent to a randomization over some price-posting mechanisms. Therefore, consider the set of (seller revenue, buyers' surplus) pairs that can be reached by some mechanism in $\mechFam$. We then know that all extreme points of this 2-dimensional set are achieved by some price-posting mechanisms. Thus, all points in this set can be realized by a randomization of at most two price-posting mechanisms. 
We further have the following fact. 
\begin{fact} \label{fct:Pareto frontier:buyers' surplus increase}
For a single buyer, her surplus increases as the posted price decreases.
\end{fact}
Therefore, all price-posting mechanisms with a price larger than the lowest monopoly reserve are Pareto-dominated by the price-posting mechanism with the lowest monopoly reserve, i.e., the revenue-optimal mechanism. Subsequently, extreme points on the Pareto frontier are those with prices no larger than the lowest monopoly reserve. This proves \Cref{prop:Pareto frontier:single buyer general}.

In particular, when the value distribution is regular, we can show that price-posting mechanisms suffice to construct the entire Pareto frontier. We defer the proof of the following result to \Cref{sec:proofs in Pareto frontier}.

\begin{restatable}{proposition}{PFSingleRegular}
\label{prop:Pareto frontier:single buyer regular}
    In a single-item seller, single-buyer market, for any \emph{regular} valuation distribution $\dist$, any point on the Pareto frontier can be realized by any point on the Pareto frontier can be realized by a single price-posting mechanism.  
\end{restatable}

A natural subsequent question is whether, for general distributions, the Pareto frontier of randomized mechanisms strictly dominates the Pareto frontier of deterministic mechanisms.  
In \Cref{thm:revenue approximation for single buyer}, we show that it is indeed the case.
In fact, under some valuation distribution $\dist$ and some realizable buyers' surplus $\buyerutil$, the optimal seller revenue of single price-posting mechanisms can be arbitrarily small compared to that of a randomization over two price-posting mechanisms. 
When restricting to the family of quasi-regular value distributions~\citep{FJ-24}, the gap between these two Pareto frontiers still exists, but for any buyers' surplus $\buyerutil$, the ratio of seller revenues is bounded by a constant. We defer a formal presentation and an analysis to \Cref{sec:revenue approximation with single buyer}.

\subsection{Multi-Buyer Market}
\label{sec:Pareto frontier:multiple buyer}

We now consider the setting with multiple ex-ante symmetric buyers whose values are drawn i.i.d.\ from the same value distribution $\dist$. In contrast to the relatively simple form of Pareto-optimal mechanisms in the single-buyer market, Pareto-optimal mechanisms in the multi-buyer setting might be significantly more complex due to competition among buyers.
In this section, we additionally make the following assumption.\footnote{For multi-buyer markets, our characterizations of Pareto-optimal mechanisms rely on additional distributional assumptions. This is driven by both technical and conceptual considerations. As we illustrate in \Cref{sec:Pareto frontier:reduction}, our analysis hinges on a duality-based (Lagrangian relaxation) approach, and these distributional assumptions yield a cleaner construction of the optimal dual variable (Lagrangian multiplier). From a results perspective, they also enable a structurally simple and interpretable characterization of the frontier. We conjecture that without such assumptions, Pareto-optimal mechanisms may exhibit significantly more complex forms---an intriguing direction we leave for future work.} 
\begin{assumption}
    The value function $\valFunc(\cdot)$ of the valuation distribution $\dist$ is twice differentiable.
\end{assumption}

We first consider the case in which the value distribution $\dist$ is both regular and anti-MHR. Intuitively, this represents the most tractable setting: because $\dist$ is regular (resp.\ anti-MHR), the seller-optimal (resp.\ buyer-optimal) mechanism takes the simple form of a second-price auction with a monopoly reserve (resp.\ a standard second-price auction) \citep{mye-81,HR-08}. As shown below, Pareto-optimal mechanisms in this case can be captured by a natural interpolation between these two mechanisms.

\begin{restatable}{theorem}{PFRegularAntiMHR}
    \label{thm:Pareto frontier:multi buyer:regular and antiMHR distribution}
    Consider a single-item market with $n \geq 1$ buyers whose valuations are drawn i.i.d.\ from a regular and anti-MHR distribution $\buyerdist$. Then the following hold:
\begin{enumerate}[label=(\roman*)]
        \item Any point on the Pareto frontier can be realized by implementing a second-price auction with a reserve $\reserve$ no larger than the lowest monopoly reserve. 
        \item Suppose that $\mech\primed$ and $\mech\doubleprimed$ are two Pareto-optimal second-price auctions with reserves on the Pareto frontier, with buyers' surpluses $\buyerutil\primed$ and $\buyerutil\doubleprimed$ respectively. If $\buyerutil\primed < \buyerutil\doubleprimed$, then the reserve of $\mech\primed$ is no smaller than that of $\mech\doubleprimed$. 
    \end{enumerate}
\end{restatable}

As we have mentioned in \Cref{sec:prelim} and \Cref{fig:dist-hierarchy}, some natural distributions that are both regular and anti-MHR include equal revenue distributions, Pareto Type I distributions with $\alpha \geq 1$ and Pareto Type II distributions with $\alpha \geq 1$.

We next consider the case where the value distribution $\dist$ is anti-MHR but not necessarily regular. In this setting, we show that all \emph{extreme points} on the Pareto frontier, i.e., those points that are not a convex combination of any two other points on the Pareto frontier, can be realized by an {\indirectSPA}.

\begin{restatable}{theorem}{PFAntiMHR}
    \label{thm:Pareto frontier:multi buyer:antiMHR distribution}
Consider a single-item market with $n \geq 1$ buyers whose valuations are drawn i.i.d.\ from an anti-MHR distribution $\buyerdist$. Then the following hold:
    \begin{enumerate}[label=(\roman*)]
        \item Any \emph{extreme point} on the Pareto frontier can be realized by implementing an {\indirectSPA}, with a reserve $\reserve$ no larger than the lowest monopoly reserve.
        \item Suppose that $\mech\primed$ and $\mech\doubleprimed$ are two Pareto-optimal {\indirectSPAs} as above, with buyers' surpluses $\buyerutil\primed$ and $\buyerutil\doubleprimed$ respectively. If $\buyerutil\primed < \buyerutil\doubleprimed$, then the bidding space of $\mech\doubleprimed$ is a refinement of that of $\mech\primed$, and the reserve of $\mech\doubleprimed$ is no larger than that of $\mech\primed$. 
    \end{enumerate}
\end{restatable}

As a direct corollary, we have the following that characterizes any point on the Pareto frontier:
\begin{corollary} 
Consider a single-item market with $n \geq 1$ buyers whose valuations are drawn i.i.d.\ from an anti-MHR distribution $\buyerdist$.
    Any point on the Pareto frontier can be realized by a mechanism that fixes a probability $\mixprob \in [0, 1]$ ex-ante, then with probability $\mixprob$, realizes an {\indirectSPA} $\mech\primed$ with bidding space $\bidspace\primed$ and reserve $\reserve\primed$, and realizes another {\indirectSPA} $\mech\doubleprimed$ with bidding space $\bidspace\doubleprimed$ and reserve $\reserve\doubleprimed$ with probability $1 - \mixprob$. Both $\reserve\primed$ and $\reserve\doubleprimed$ are no larger than the lowest monopoly reserve. Additionally, $\bidspace\doubleprimed$ is a refinement of $\bidspace\primed$, and $\reserve\doubleprimed$ is no larger than $\reserve\primed$.
\end{corollary}

We may need to restrict the feasible bidding space. In fact, as shown in \cite{HR-08}, the optimal mechanism for buyers' surplus in this case is giving the item uniformly at random to one buyer for free, which is equivalent to setting the feasible bid space as a singleton $\{\lval\}$. Randomizing over two {\indirectSPAs} may also be necessary. 
The restriction on the feasible bid space (i) arises primarily from the need to iron the virtual valuation function to optimize seller revenue \citep[cf.][]{mye-81}. Meanwhile, the randomization over two {\indirectSPAs} (ii) is necessary to balance the seller's ex-ante utility and the buyers’ ex-ante surplus in order to achieve a continuous Pareto frontier, mirroring the structure observed in the single-buyer market (\Cref{prop:Pareto frontier:single buyer general}).

Furthermore, \Cref{thm:Pareto frontier:multi buyer:antiMHR distribution} reveals a monotonicity structure of Pareto-optimal mechanisms along the Pareto frontier across their design parameters and two sides' ex-ante utility/surplus. Specifically, suppose that $\mech\primed$ realizes an extreme point on the Pareto frontier and $\mech\doubleprimed$ realizes another, then the following conditions are equivalent:
\begin{align*}
    \bidspace\primed \subseteq \bidspace\doubleprimed
    \Leftrightarrow
    \reserve\primed \geq \reserve\doubleprimed
    \Leftrightarrow
    \sellerexanteutil[\dist]{\mech\primed} \geq \sellerexanteutil[\dist]{\mech\doubleprimed}
    \Leftrightarrow
    \buyerexanteutil[\dist]{\mech\primed} \leq \buyerexanteutil[\dist]{\mech\doubleprimed}.
\end{align*}

Finally, we characterize the Pareto frontier when the value distribution is MHR in \Cref{thm:Pareto frontier:multi buyer:MHR distribution} below. 
Recall that for MHR distributions, the seller-optimal mechanism is the second-price auction with the lowest monopoly reserve. On the other hand, as shown in \citet{HR-08}, under MHR distributions, competition among buyers only reduces their surplus; consequently, the buyer-optimal mechanism allocates the item to a buyer picked uniformly at random, at zero price. 
As the following theorem shows, Pareto-optimal mechanisms can be viewed as an interpolation between these two extremes: 
it either (i) implements a second-price auction with a reserve price strictly below the monopoly reserve, but without restricting the bidding space, or (ii) employs a randomization of two {\indirectSPAs} with no reserve to ensure that the item is always allocated, but restricts the bidding space to dampen competition among buyers. 
Interestingly, there is no need for simultaneously restricting the bidding space and setting a reserve. 

\begin{restatable}{theorem}{PFMHR}
    \label{thm:Pareto frontier:multi buyer:MHR distribution}
Consider a single-item market with $n \geq 1$ buyers whose valuations are drawn i.i.d.\ from an MHR distribution $\buyerdist$.
    Let $\buyerutilmid$ be buyers' surplus under the standard second-price auction. Then the following hold:
    \begin{enumerate}[label=(\roman*)]
        \item If $\buyerexanteutil[\dist]{\Mye} \leq \buyerutil \leq \buyerutilmid$, any point on the Pareto frontier with buyers' surplus $\buyerutil$ can be implemented by a second-price auction with a reserve $\reserve$ no larger than the lowest 
        monopoly reserve. 
        \item Suppose that $\mech\primed$ and $\mech\doubleprimed$ are two Pareto-optimal second-price auctions with reserves on the Pareto frontier as above, with buyers' surpluses $\buyerutil\primed$ and $\buyerutil\doubleprimed$ respectively. If $\buyerutil\primed < \buyerutil\doubleprimed$, then the reserve of $\mech\primed$ is no smaller than that of $\mech\doubleprimed$.
        \item If $\buyerutilmid \leq \buyerutil \leq \buyerbenchmarkDist$, and the point on the Pareto frontier with buyers' surplus $\buyerutil$ is an extreme point, it can be realized by an {\indirectSPA} with zero reserve. 
        \item Suppose that $\mech\primed$ and $\mech\doubleprimed$ are two Pareto-optimal {\indirectSPAs} as above, with buyers' surpluses $\buyerutil\primed$ and $\buyerutil\doubleprimed$ respectively. If $\buyerutil\primed < \buyerutil\doubleprimed$, then the bidding space of $\mech\primed$ is a refinement of that of $\mech\doubleprimed$.
    \end{enumerate}
\end{restatable}

Similar to the anti-MHR case, we have the following corollary:
\begin{corollary}
Consider a single-item market with $n \geq 1$ buyers whose valuations are drawn i.i.d.\ from an MHR distribution $\buyerdist$.
    Any point on the Pareto frontier with buyers' surplus no smaller than $\buyerutilmid$ can be realized by a mechanism that fixes a probability $\mixprob \in [0, 1]$ ex-ante and runs as follows: With probability $\mixprob$, it realizes an {\indirectSPA} $\mech\primed$ with bidding space $\bidspace\primed$ and zero reserve, and realizes another {\indirectSPA} $\mech\doubleprimed$  with bidding space $\bidspace\doubleprimed$ and zero reserve, with probability $1 - \mixprob$. Meanwhile, $\bidspace\primed$ is a refinement of $\bidspace\doubleprimed$. 
\end{corollary}

\Cref{thm:Pareto frontier:multi buyer:MHR distribution} reveals a phase transition in Pareto-optimal mechanisms along the Pareto frontier. When buyers' surplus is below a threshold $\buyerutilmid$ (which is buyers' surplus under the second-price auction), these mechanisms are second-price with reserve, resembling the case when the valuation distribution is regular and anti-MHR, and the reserve decreases as buyers' surplus increases. However, when buyers' surplus is above $\buyerutilmid$, Pareto-optimal mechanisms are randomization over at most two {\indirectSPAs} with zero reserve, one being a refinement of the other, resembling the case of anti-MHR valuation distributions. Nevertheless, the monotonicity direction is reversed: In the anti-MHR case, to implement a high buyers' surplus, the bidding space expands, while it shrinks in the MHR case.

\subsection{A Unified Framework for Characterizing Pareto Frontier}
\label{sec:Pareto frontier:reduction}

In this section, we present a unified analysis framework for characterizing the Pareto frontier, which can be applied to prove all the above results, both for the single-buyer market and the multi-buyer market. Our main idea is to transfer this problem to a constrained optimization problem that maximizes seller revenue subject to a lower bound constraint on buyers' surplus. We then solve this problem via the Lagrangian dual method. This problem could be of independent interest itself.

We start with the following result. 
\begin{lemma} \label{lem:Pareto frontier:transfer to constrained optimization}
    Any mechanism $\mech\primed \in \mechFam$ is Pareto-optimal, if and only if $\mech\primed$ is an optimal solution to the following optimization problem:     
    \begin{gather} 
        \max_{\mech \in \mechFam} \sellerexanteutil[\dist]{\mech}, 
        \quad 
        \subjectto \buyerexanteutil[\dist]{\mech} \geq \buyerutil\primed. 
        \label{eq:constrained optimization}
    \end{gather}
    In the above, $\buyerutil\primed = \buyerexanteutil[\dist]{\mech\primed}$.
\end{lemma}

\begin{proof}[Proof of \Cref{lem:Pareto frontier:transfer to constrained optimization}]
    We prove the two sides in order. 
    
    \xhdr{``Only if'' side.} Suppose that $\mech\primed$ is not an optimal solution to \eqref{eq:constrained optimization}. Then there exists a mechanism $\mech$ such that $\sellerexanteutil[\dist]{\mech} > \sellerexanteutil[\dist]{\mech\primed}$ and $\buyerexanteutil[\dist]{\mech} \geq \buyerexanteutil[\dist]{\mech\primed}$, contradicting that $(\sellerexanteutil[\dist]{\mech\primed}, \buyerexanteutil[\dist]{\mech\primed})$ is on the Pareto frontier. 
    
    \xhdr{``If'' side.} Suppose $\mech\primed$ is an optimal solution to \eqref{eq:constrained optimization}, then there is no mechanism such that $\sellerexanteutil[\dist]{\mech} > \sellerexanteutil[\dist]{\mech\primed}$ and $\buyerexanteutil[\dist]{\mech} \geq \buyerexanteutil[\dist]{\mech\primed}$ both hold. 
    We now show that it is also impossible that $\sellerexanteutil[\dist]{\mech} = \sellerexanteutil[\dist]{\mech\primed}$ and $\buyerexanteutil[\dist]{\mech} > \buyerexanteutil[\dist]{\mech\primed}$ both hold. Otherwise, consider the mechanism that randomizes between $\Mye$ and $\mech$ such that the buyers' surplus is $\buyerexanteutil[\dist]{\mech\primed}$. Since $\Mye$ maximizes seller revenue, we know that seller revenue under this randomized mechanism is at least $\sellerexanteutil[\dist]{\mech}$. By the optimality of $\mech\primed$, we know that $\sellerexanteutil[\dist]{\Mye} = \sellerexanteutil[\dist]{\mech\primed} = \sellerexanteutil[\dist]{\mech}$. Notice that $\Mye$ is the mechanism that maximizes buyers' surplus among all seller-utility-maximizing mechanisms; this implies a contradiction.
\end{proof}

Under \Cref{lem:Pareto frontier:transfer to constrained optimization}, to characterize the full Pareto frontier, we only need to consider the solutions to the following program when $\buyerutil \in [\buyerexanteutil[\dist]{\Mye}, \buyerbenchmarkDist]$:
\begin{gather*} 
    \max_{\mech \in \mechFam} \sellerexanteutil[\dist]{\mech}, 
    \quad 
    \subjectto \buyerexanteutil[\dist]{\mech} \geq \buyerutil. 
\end{gather*}
This is to say, characterizing Pareto-optimal solutions is equivalent to studying mechanisms that maximize seller revenue under a buyers' surplus constraint. 
Now, for the above constrained optimization problem, we consider the dual problem with Lagrangian variable $\dualV$ on the surplus constraint. 
\begin{align*}
\min_{\dualV \geq 0} \max_{\mech \in \mechFam} \sellerexanteutil[\dist]{\mech} + \dualV \cdot \buyerexanteutil[\dist]{\mech} - \dualV \cdot \buyerutil. 
\end{align*}
Since for any given $\dualV \geq 0$, the term $- \dualV \cdot \buyerutil$ is a constant that does not depend on the mechanism, we aim to find a mechanism that maximizes the following expression:
\begin{align*}
    \sellerexanteutil[\dist]{\mech} + \dualV \cdot \buyerexanteutil[\dist]{\mech} &= \sellerexanteutil[\dist]{\mech} + \dualV \cdot \inParentheses{\GFT[\dist]{\mech} - \buyerexanteutil[\dist]{\mech}} \\
    &= \dualV\cdot \GFT[\dist]{\mech} + (1 - \dualV)\cdot \sellerexanteutil[\dist]{\mech}, 
\end{align*}
which is a mix of welfare and seller revenue. 
En route to solving the original constrained optimization problem, our main idea is to take the following two steps:
\begin{enumerate}
    \item [\textbf{Step 1.}] For any Lagrangian variable $\dualV \geq 0$, we characterize the optimal mechanism $\mech[\dualV]\opted$ to maximize the mixed welfare-seller revenue objective $\dualV\cdot \GFT[\dist]{\mech} + (1 - \dualV)\cdot \sellerexanteutil[\dist]{\mech}$. We call this the primal problem.  
    \item [\textbf{Step 2.}] Prove strong duality of the problem, by establishing the existence of a dual variable $\dualV(\buyerutil) \geq 0$ and corresponding optimal primal mechanism $\mech[\dualV(\buyerutil)]\opted$ that satisfy KKT conditions. This immediately induces dual optimality.
\end{enumerate}

\subsection{Proof of \texorpdfstring{\Cref{thm:Pareto frontier:multi buyer:regular and antiMHR distribution}}{Theorem 3.4}}
\label{sec:Pareto frontier:reduction proof example}

To demonstrate how the above framework can be applied, we prove \Cref{thm:Pareto frontier:multi buyer:regular and antiMHR distribution} in this section. Specifically, we show that when the valuation distribution is both regular and anti-MHR, every point on the Pareto frontier is realized by a second-price auction with a reserve no larger than the lowest monopoly reserve. 

To analyze this multi-buyer setting, we consider the interim space. We use $\alloc(\quant)$ to denote any buyer's interim allocation probability when her value is $\valFunc(\quant)$. We further use $\accAlloc(\quant)$ to denote the accumulated allocation, i.e., $\accAlloc(\quant) \triangleq \int_0^{\quant} \alloc(\qquant) \dd \qquant$. By Border's condition~\citep{bor-91}, an interim allocation function $\alloc(\quant)$ on the quantile space can be realized ex post by a BIC and interim IR mechanism if and only if
\begin{gather*}
    \accAlloc(\quant) \leq \frac{1 - (1 - \quant)^{\numAgents}}{\numAgents}, \quad \forall \quant \in [0, 1]; \qquad
    \alloc(\quant) \leq \frac{1}{\numAgents}, \quad \forall \quant \in [0, 1]; \qquad
    \alloc(\quant) \text{ is decreasing. }
\end{gather*}
Specifically, in the single buyer case, the first constraint becomes $\accAlloc(\quant) \leq \quant$, which is directly implied by the second constraint $\alloc(\quant) \leq 1$ and non-restrictive. But this is not the case for the multi-buyer case. 

Back to our problem, according to the revenue equivalence theorem, the expected revenue under a mechanism with interim allocation $\alloc(\quant)$ is $\numAgents\cdot \expect[\quant]{\alloc(\quant)\cdot \revCurvePrime(\quant)}$, which is derived by seeing $\alloc(\quant)$ as a randomization of price-posting mechanisms~\citep{mye-81}. 
Thus, the constrained program \eqref{eq:constrained optimization} is formalized as follows:
\begin{gather}
    \begin{gathered}
        \max_{\alloc} \numAgents\cdot \expect[\quant]{\alloc(\quant)\cdot \revCurvePrime(\quant)}, \\
        \subjectto \numAgents\cdot \expect[\quant]{\alloc(\quant)\cdot (\valFunc(\quant) - \revCurvePrime(\quant))} \geq \buyerutil, \\
        \accAlloc(\quant) \leq \frac{1 - (1 - \quant)^{\numAgents}}{\numAgents}, \quad \forall \quant \in [0, 1]; \qquad
        \alloc(\quant) \leq \frac{1}{\numAgents}, \quad \forall \quant \in [0, 1]; \qquad
        \alloc(\quant) \text{ is decreasing. }
    \end{gathered} \label{eq:multi-buyer constrained optimization}
\end{gather}

For brevity, we let $\accAllocUb(\quant) \triangleq \nicefrac{1 - (1 - \quant)^{\numAgents}}{\numAgents}$ and $\allocUb(\quant) \triangleq \nicefrac{1}{n}$. 
As mentioned in our unified framework, we Lagrangify the surplus constraint to derive the following dual program: 
\[\min_{\dualV \geq 0} \max_{\accAlloc(\quant) \leq \accAllocUb(\quant), \; \alloc(\quant) \leq \allocUb(\quant), \; \alloc(\quant) \text{ is decreasing}} \numAgents\cdot \expect[\quant]{\alloc(\quant)\cdot (\dualV\cdot \valFunc(\quant) + (1 - \dualV)\cdot \revCurvePrime(\quant))} - \dualV\cdot \buyerutil.\]
In particular, we let 
\[\lagCurve(\quant) \triangleq \dualV\cdot \int_{0}^{\quant} \valFunc(\qquant) \dd \qquant + (1 - \dualV) \cdot \revCurve(\quant). \] 
We then notice that $\dualV\cdot \valFunc(\quant) + (1 - \dualV)\cdot \revCurvePrime(\quant) = \lagCurvePrime(\quant)$. So the dual program is the following:
\[\min_{\dualV \geq 0} \max_{\accAlloc(\quant) \leq \accAllocUb(\quant), \; \alloc(\quant) \leq \allocUb(\quant), \; \alloc(\quant) \text{ is decreasing}} \numAgents\cdot \expect[\quant]{\alloc(\quant)\cdot \lagCurvePrime(\quant)} - \dualV\cdot \buyerutil.\]

Our roadmap is as follows.
First, for each $\dualV \geq 0$, we derive the allocation $\alloc(\quant)$ that maximizes the primal objective. 
Further, we establish the existence of a $\dualVOPT \geq 0$ such that the complementary slackness condition is met. 
They in together solve the problem. 

Here, one special property brought by the anti-MHR condition is that we can restrict ourselves to $\dualV \in [0, 1]$. In other words, we can prove that it is without loss to suppose $\dualVOPT \leq 1$ in the following proposition. 

\begin{proposition} \label{prop:Pareto frontier:dualVOPT bounded by 1 for anti-MHR}
    For anti-MHR value distribution $\dist$, there exists a $\dualV \in [0, 1]$ that minimizes
    \[\max_{\accAlloc(\quant) \leq \accAllocUb(\quant), \; \alloc(\quant) \leq \allocUb(\quant), \; \alloc(\quant) \text{ is decreasing}} \numAgents\cdot \expect[\quant]{\alloc(\quant)\cdot \lagCurvePrime(\quant)} - \dualV\cdot \buyerutil\]
    for all $\dualV \geq 0$. 
\end{proposition}

Since the value distribution is regular, we know that $\revCurve(\quant)$ is concave. Meanwhile, as $\valFunc(\qquant)$ is non-negative and non-increasing, $\int_0^{\quant} \valFunc(\qquant) \dd \qquant$ is also a concave function of $\quant$. Therefore, for $\dualV \in [0, 1]$, $\lagCurve(\quant)$ is always concave. We now let
\[\leftS(\dualV) \triangleq \inf \argmax_{\quant} \lagCurve(\quant), \quad \rightS(\dualV) \triangleq \sup \argmax_{\quant} \lagCurve(\quant). \]
Thus, a family of optimal primal solutions $\allocOPT[\dualV](\quant)$ takes the following form: 
Fix any $\quant[\dualV] \in [\leftS(\dualV), \rightS(\dualV)]$, when $\quant \leq \quant[\dualV]$, $\allocOPT[\dualV](\quant) = \accAllocUb'(\quant)$; and when $\quant > \quant[\dualV]$, $\allocOPT[\dualV](\quant) = 0$. 
This is because when $\quant \in [\leftS(\dualV), \rightS(\dualV)]$, the value of $\alloc(\quant)$ does not affect the objective; when $\quant > \rightS(\dualV)$, $\lagCurvePrime(\quant) \leq 0$ and the optimal choice of $\alloc(\quant)$ is $0$. At last, when $\quant < \leftS(\dualV)$, we know by integration by parts that
\[\int_{0}^{\leftS(\dualV)} \alloc(\quant)\cdot \lagCurvePrime(\quant) \dd \quant = - \int_{0}^{\leftS(\dualV)} \accAlloc(\quant)\cdot \lagCurveDoublePrime(\quant) \dd \quant. \]
Since $\lagCurveDoublePrime(\quant) \leq 0$ when $\quant \leq \leftS(\dualV)$, the optimal choice of $\alloc(\quant)$ is to have $\accAlloc(\quant) = \accAllocUb(\quant)$, i.e., $\allocOPT[\dualV](\quant) = \accAllocUb'(\quant)$. 

We also have the following result that establishes the monotonicity of $\leftS(\dualV)$ and $\rightS(\dualV)$ with $\dualV$: 
\begin{proposition}
    For any $0 \leq \dualV < \dualV' \leq 1$, $\leftS(\dualV') \geq \leftS(\dualV)$ and $\rightS(\dualV') \geq \rightS(\dualV)$ hold. 
\end{proposition}

Now notice that when $\dualV = 0$, $\lagCurve(\quant) = \revCurve(\quant)$, and thus $\rightS(\dualV)$ is the quantile of the lowest monopoly reserve $\quantM$. When $\dualV = 1$, $\lagCurve(\quant) = \int_0^{\quant} \valFunc(\qquant) \dd \qquant$ is an increasing function, thus $\rightS(\dualV) = 1$.
In all, we can construct a process $\{\quant[\dualV]\}_{0 \leq \dualV \leq 1}$ of quantiles and a process of $\{\allocOPT[\dualV]\}_{0 \leq \dualV \leq 1}$ of primal solutions that satisfies the following:
\begin{enumerate}
    \item When $\dualV$ increases from $0$ to $1$, $\quant[\dualV]$ increases \emph{continuously} from $\quant[0] = \quantM$ to $\quant[1] = 1$.
    \item $\allocOPT[\dualV](\quant) = \accAllocUb'(\quant)$ when $\quant \leq \quant[\dualV]$, and $\allocOPT[\dualV](\quant) = 0$ when $\quant > \quant[\dualV]$; 
    \item $\allocOPT[\dualV]$ is an optimal solution to the primal program $\max \expect[\quant]{\alloc(\quant) \cdot \lagCurvePrime(\quant)}$; 
\end{enumerate}

Under the process $\{\allocOPT[\dualV]\}_{0 \leq \dualV \leq 1}$, we observe that when $\dualV = 0$, the corresponding primal-optimal mechanism under $\allocOPT[0]$ is just $\Mye$, which obtains a buyers' surplus of $\buyerexanteutil[\dist]{\Mye}$; when $\dualV = 1$, the corresponding primal-optimal mechanism is in fact the second-price auction, which reaches the optimal buyers' surplus $\buyerbenchmarkDist$, as given in the following proposition:  
\begin{proposition}[Corollary 2.12, \citealp{HR-08}] \label{prop:multi buyer:anti-MHR SPA gives umax}
    When the value distribution is anti-MHR, $\buyerbenchmarkDist$ is achieved by the second-price auction.
\end{proposition}

Meanwhile, with a little bias of notation, we know the formula of buyers' surplus $\buyerexanteutil[\dist]{\alloc}$ under the allocation function $\alloc$ can be written as follows. 
\begin{align*}
    \buyerexanteutil[\dist]{\alloc} = \numAgents\cdot \expect[\quant]{\alloc(\quant)\cdot (\valFunc(\quant) - \revCurvePrime(\quant))}
    = \numAgents\cdot \expect[\quant]{\alloc(\quant)\cdot (-\quant \valFuncPrime(\quant))}. 
\end{align*}
Notice that $-\quant \valFuncPrime(\quant) \geq 0$. Thus, under the process $\{\allocOPT[\dualV]\}_{0 \leq \dualV \leq 1}$, since $\allocOPT[\dualV](\quant)$ monotonically and continuously increases with $\dualV$, so does $\buyerexanteutil[\dist]{\allocOPT[\dualV]}$. Therefore, by the intermediate value theorem, for any $\buyerutil \in [\buyerexanteutil[\dist]{\Mye}, \buyerbenchmarkDist]$, there exists a $\dualV(\buyerutil)$ such that $\buyerexanteutil[\dist]{\allocOPT[\dualV(\buyerutil)]} = \buyerutil$ holds. 
This implies that $(\dualV(\buyerutil), \allocOPT[\dualV(\buyerutil)])$ satisfies the KKT condition of the primal program \eqref{eq:multi-buyer constrained optimization}. 

Therefore, for any $\buyerutil \in [\buyerexanteutil[\dist]{\Mye}, \buyerbenchmarkDist]$, $\allocOPT[\dualV(\buyerutil)]$ is an optimal solution to \eqref{eq:multi-buyer constrained optimization}. 
By \Cref{lem:Pareto frontier:transfer to constrained optimization}, we conclude that for any $(\sellerutil, \buyerutil)$ that is on the Pareto frontier, it can be realized by the mechanism with allocation function $\allocOPT[\dualV(\buyerutil)]$. 
This mechanism itself is a second-price auction with a reserve of quantile $\quant[\dualV(\buyerutil)] \geq \quantM$, and thus, the reserve is no larger than the lowest monopoly reserve. Meanwhile, as $\buyerutil$ increases, the reserve of the corresponding reserve decreases.
This finishes the proof of \Cref{thm:Pareto frontier:multi buyer:regular and antiMHR distribution}. 

Despite the succinctness of combining the dual-based approach and the analysis of the mixed welfare-seller revenue curve $\lagCurve$, in the scenario of \Cref{thm:Pareto frontier:multi buyer:regular and antiMHR distribution}, such a combination would be much more difficult under a looser setting for the value distribution. On one hand, without regularity, although we can still restrict ourselves to $\dualV \in [0, 1]$ as given in \Cref{prop:Pareto frontier:dualVOPT bounded by 1 for anti-MHR}, the mixed curve can be non-concave for some $\dualV \leq 1$. On the other hand, without the anti-MHR condition, even if the distribution is MHR, we must still consider the cases that the dual variable $\dualV \geq 1$, when the mixed curve can be non-concave. 
In all, to meet more general requirements, we need a more delicate analysis of how the mixed curve evolves, and constructing a corresponding continuous and monotone evolution process of primal optimal solutions becomes much more complicated. 
We defer such an analysis and a unified proof of \Cref{prop:Pareto frontier:single buyer general,prop:Pareto frontier:single buyer regular}, and \Cref{thm:Pareto frontier:multi buyer:regular and antiMHR distribution,thm:Pareto frontier:multi buyer:MHR distribution,thm:Pareto frontier:multi buyer:antiMHR distribution} to \Cref{sec:proofs in Pareto frontier}.  
\section{The Kalai-Smorodinsky Solution}
\label{sec:ks solution}

Having characterized the Pareto frontier in \Cref{sec:Pareto frontier}, an important question remains: which point on this frontier should be implemented? A natural desideratum is to achieve fairness between the two sides---the seller and the buyers---while maintaining strong efficiency guarantees. 

Guided by this principle, we interpret the selection of a Pareto-optimal mechanism as a cooperative bargaining problem \citep{nas-50} between the seller and the buyers.\footnote{We refer the reader to \citet[Appendix~A]{BFM-25} for further discussion on interpreting mechanism design as a cooperative bargaining problem.} In this section, we study a canonical solution concept from bargaining theory: the {\KSsolution} \citep{KS-75}. (Another canonical solution---the {\Nashsolution}---is analyzed in \Cref{sec:nash solution}.)

\begin{definition}[{\KSsolution}]
    \label{def:KS solution}
    Fix a valuation distribution $\buyerdist$. A Pareto-optimal mechanism $\mech$ is a {{\sf Kalai-Smorodinsky (KS) Solution}} if both the ex-ante seller revenue and ex-ante buyers' surplus achieve the same fraction of each side's own ideal benchmark, i.e., 
    \begin{align}
    \label{eq:KS fairness}
    \tag{\texttt{KS fairness}}
        \frac{\sellerexanteutil[\buyerdist]{\mech}}{\sellerbenchmarkDist} = \frac{\buyerexanteutil[\buyerdist]{\mech}}{\buyerbenchmarkDist}.
    \end{align}
    Here, the benchmarks $\sellerbenchmarkDist$ and $\buyerbenchmarkDist$ denote the optimal achievable ex-ante seller revenue and ex-ante buyers' surplus, respectively, over all mechanisms in $\mechFam$.
\end{definition}

We make the following remarks regarding the {\KSsolution}.  
First, due to the convexity of the mechanism space and the linearity of both ex-ante seller revenue and buyers' surplus, the set of achievable payoff pairs (seller revenue, buyers' surplus) is convex, and the Pareto frontier is continuous. Consequently, a {\KSsolution} always exists \citep{KS-75}.  

Second, while multiple mechanisms may implement the {\KSsolution}, they all correspond to the same payoff pair on the Pareto frontier and thus yield identical social welfare. Therefore, our structural characterization of Pareto-optimal mechanisms in \Cref{sec:Pareto frontier} can be used to construct an explicit implementation of the {\KSsolution}.  

Third, \citet{BFM-25} study the single-seller, single-buyer bilateral-trade setting and refer to constraint~\eqref{eq:KS fairness} as \emph{KS fairness}. By definition, the {\KSsolution} maximizes social welfare among all KS-fair mechanisms.\footnote{Indeed, to establish positive welfare guarantees for the {\KSsolution}, we sometimes construct a KS-fair mechanism and lower bound its welfare; by definition, this serves as a lower bound on the welfare of the {\KSsolution}.}

We now present the main result of this section: the welfare approximation guarantee for the {\KSsolution} under various distributional assumptions on the buyers' valuation distribution $\buyerdist$.

\begin{restatable}{theorem}{thmKSsolution}
\label{thm:KS solution welfare approximation}
Consider a single-item market with $n \geq 1$ buyers whose values are drawn i.i.d.\ from a distribution $\buyerdist$. Then the following hold:
\begin{enumerate}[label=(\roman*)]
    \item For any valuation distribution $\buyerdist$, the social welfare of the {\KSsolution} is at least $50\%$ of the optimal social welfare.
    
    \item If $\buyerdist$ is both regular and anti-MHR, then the {\KSsolution} achieves at least $\frac{n}{n+1}$ of the optimal social welfare.
    
    \item For every $\eps > 0$ and any sufficiently large $n$, there exists an instance (\Cref{example:multi-buyer:mhr hard instance}) of $n$ buyers with MHR valuation distribution $\buyerdist$ such that the social welfare of the {\KSsolution} is at most $(50 + \eps)\%$ of the optimal social welfare.
\end{enumerate}
\end{restatable}

\begin{example}
\label{example:multi-buyer:mhr hard instance}
    Consider a single-item market with $n$ buyers. Let $\constantH\triangleq \frac{1}{2}\sqrt{\ln n}$. Buyers' values are drawn i.i.d.\ from distribution $\buyerdist$ (which is MHR) with support $\supp(\buyerdist) = [0, \constantH]$ and cumulative density $\buyercdf(\val) = 1 - e^{-\val^2}$ for $\val \in [0, \constantH]$ and $\buyercdf(\val) = 1$ for $\val\in(\constantH, \infty)$. (Namely, there is an atom at $\constantH$ with probability mass of $\frac{1}{{n}^{1/4}}$.)
\end{example}

We now provide several remarks and intuition behind \Cref{thm:KS solution welfare approximation}. All the proofs for this section can be found in \Cref{apx:ks solution analysis}.

First, the theorem demonstrates the robustness of the {\KSsolution} in terms of market efficiency: for all market sizes and all valuation distributions, it guarantees at least half of the optimal social welfare. The proof of this universal bound follows an approach similar to that of \citet{BFM-25} in the single-seller, single-buyer bilateral-trade setting.

As shown in \Cref{sec:Pareto frontier}, when the distribution is both regular and anti-MHR, every point on the Pareto frontier admits a simple implementation---a second-price auction with a reserve no larger than the monopoly reserve. While this structural simplicity does not imply that all Pareto-optimal mechanisms achieve high welfare (e.g., the Myerson auction may yield negligible welfare under certain regular and anti-MHR distributions in large markets\footnote{For example, consider the triangular distribution $\buyerdist$ with $\buyerdist(\val) = 1-1/(\val+1)$ with $\supp(\buyerdist) = [0,\infty)$.}),
it is sufficient to ensure that the {\KSsolution} attains a strong welfare guarantee---namely, $\frac{n}{n+1}$ of the optimum.

In contrast to much of the algorithmic mechanism design literature---where the MHR assumption typically enables improved approximation guarantees\footnote{In the special case of our setting when there is only a single buyer,
\citet{BFM-25} show that the {\KSsolution} achieves at least 91\% of optimal welfare under MHR.}---our negative result reveals that, in large multi-buyer markets, MHR distributions can be as adversarial as general ones. Intuitively, this stems from the tension between competition and information: under hard MHR instances like \Cref{example:multi-buyer:mhr hard instance}, price competition (e.g., in a second-price auction) drastically reduces buyers' surplus. Yet, without such competition, the mechanism lacks the discriminatory power to distinguish high-value from low-value buyers, limiting its ability to extract efficient allocations.

The fragility of the {\KSsolution} under MHR distributions in the single-unit setting raises a natural question: can increased supply mitigate this inefficiency? In many real-world platforms---such as cloud markets or online advertising---the seller often allocates multiple units, which modulates the balance of bargaining power between the seller and buyers.

We analyze this setting with $m$ identical units and $n$ unit-demand buyers whose valuations are drawn i.i.d.\ from an MHR distribution. We show in \Cref{thm:KS solution welfare approximation multi-unit} below that the social welfare of the {\KSsolution} is at least $\multiUnitRatio(m/n)$ fraction of the optimal social welfare, where $\multiUnitRatio(\imbalanceness)$ is defined in \Cref{thm:KS solution welfare approximation multi-unit}. (Also see \Cref{fig:mhr:multi-unit}.)
This provides a nontrivial lower bound that depends only on the supply-to-demand ratio $\imbalanceness = m/n$.

Intuitively, bargaining power shifts with $\imbalanceness$: when $\imbalanceness$ is small (scarce supply), the seller dominates; when $\imbalanceness$ is close to 1 (abundant supply), buyers gain leverage. In either extreme, the need to satisfy KS fairness---equalizing relative gains---can distort allocations away from the welfare optimum. Consequently, the highest welfare is expected to occur at intermediate levels of supply, yielding a non-monotonic (first increasing, then decreasing) relationship between welfare and $\imbalanceness$. 
Notably, our lower bound $\multiUnitRatio(\imbalanceness)$ satisfies $\multiUnitRatio(\imbalanceness) > 1/2$ for all $\imbalanceness \in (0,1]$. This implies that whenever the number of items grows proportionally with the number of buyers---i.e., $m = \lceil \imbalanceness \cdot n \rceil$ for some constant $\imbalanceness > 0$---the asymptotic welfare guarantee strictly improves upon the single-unit worst case.
While this bound is not known to be tight, its non-monotonic shape is consistent with the predicted interplay between market structure, bargaining power, and efficiency.

\begin{restatable}{theorem}{thmKSsolutionMultiUnit}
\label{thm:KS solution welfare approximation multi-unit}
    Consider a market with $m$ identical units of a good and $n\geq m$ unit-demand buyers whose valuations are drawn i.i.d.\ from an MHR distribution $\buyerdist$.
    The social welfare of the {\KSsolution} is at least a $\multiUnitRatio(m/n)$-fraction of the optimal social welfare, where the auxiliary function $\multiUnitRatio : (0,1] \to [0,1]$ is defined as
    \begin{align*}
        \multiUnitRatio(\imbalanceness) \triangleq 
        \min\left\{
            \frac{2 - \ln(\imbalanceness)}{\,2 - 2\ln(\imbalanceness)\,},
            \;
            \frac{e-1}{e}+\frac{1}{e+e^2\cdot \imbalanceness}
        \right\}.
    \end{align*}
\end{restatable}

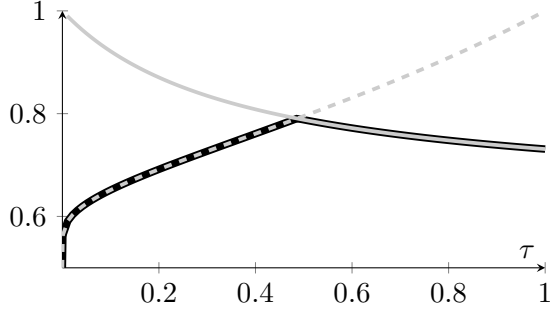
\begin{figure}
    \centering
    \begin{tikzpicture}
\begin{axis}[
    width=8cm,height=5cm,
    axis lines=middle, xlabel={$\imbalanceness$}, xmin=0.0001, xmax=1, ymin=0.5, ymax=1, domain=0.01:1, samples=100, ]
    
    \pgfmathsetmacro{\fzeroone}{(2 - ln(0.001)) / (2 - 2 * ln(0.001))}

\addplot[black, line width=1.0mm, domain=0.001:1] {min((2 - ln(x)) / (2 - 2 * ln(x)), (exp(1)-1) / exp(1) + 1/(exp(1) + exp(1) * exp(1) * x))};

    \addplot[black, line width=1.0mm] coordinates {
    (0, 0.5)
    (0.001, \fzeroone)
    };

\addplot[black!20, line width=0.5mm, domain=0.001:1, dashed] { (2 - ln(x)) / (2 - 2 * ln(x)) };

    \addplot[black!20, line width=0.5mm, dashed] coordinates {
    (0, 0.5)
    (0.001, \fzeroone)
    };

\addplot[black!20, line width=0.5mm] { (exp(1)-1) / exp(1) + 1/(exp(1) + exp(1) * exp(1) * x)};

\end{axis}
\end{tikzpicture}     \caption{The welfare approximation guarantee of the {\KSsolution} in \Cref{thm:KS solution welfare approximation multi-unit}. The gray dashed (solid) curve is $\frac{2 - \ln(\imbalanceness)}{\,2 - 2\ln(\imbalanceness)\,}$ ($\frac{e-1}{e}+\frac{1}{e+e^2\cdot \imbalanceness}$) and the black solid curve is $\multiUnitRatio(\imbalanceness)$.}
    \label{fig:mhr:multi-unit}
\end{figure}

\section{The Nash Solution}
\label{sec:nash solution}

In this section, we consider another canonical solution concept from bargaining theory: the {\Nashsolution} \citep{nas-50}.

\begin{definition}[{\Nashsolution}]
\label{def:nash solution}
    Fix a valuation distribution $\buyerdist$. A Pareto-optimal mechanism $\mech$ is a {\Nashsolution} if it maximizes the product of all traders (seller and buyers)'s ex-ante utilities, i.e.,
    \begin{align}
    \label{eq:nash fairness}
    \tag{\texttt{Nash social welfare}}
        \mech \in \argmax_{\mech\primed\in\mechfam}~
        \sellerexanteutil[\buyerdist]{\mech\primed}
        \cdot 
        \prod\nolimits_{i\in[n]}\buyerexanteutil[i,\buyerdist]{\mech\primed}.
    \end{align}
\end{definition}

We remark that the objective maximized by the {\Nashsolution} in its definition is known as the \emph{Nash social welfare} (NSW). Similar to the discussion of the {\KSsolution} in \Cref{sec:ks solution}, the {\Nashsolution} always exists. Moreover, although multiple mechanisms may implement it, they all correspond to the same payoff pair on the Pareto frontier and thus yield identical social welfare.

In contrast to the {\KSsolution}, whose fairness condition is explicit---and arguably intuitive---requiring both sides to achieve the same fraction of their respective ideal benchmarks, the fairness criterion underlying the {\Nashsolution} is more implicit, as it arises from maximizing the product of utilities. However, a crucial distinction emerges in multi-buyer markets: the standard {\Nashsolution} maximizes the product of the seller’s utility and the utility of \emph{each individual buyer}, thereby treating buyers as separate agents. This formulation becomes increasingly sensitive to competition as the number of buyers grows, ultimately leading to poor welfare guarantees. (Nonetheless, when the valuation distribution $\buyerdist$ is both regular and anti-MHR, the {\Nashsolution} performs well: like the {\KSsolution}, it achieves asymptotically optimal social welfare as the market size grows.) All the missing proofs in this section can be found in \Cref{apx:nash solution analysis}.

\begin{restatable}{theorem}{thmNashsolutionMHR}
\label{thm:Nash solution welfare approximation mhr}
For every $\eps > 0$ and any sufficiently large $n$, there exists an instance (\Cref{example:multi-buyer:mhr hard instance}) of $n$ buyers with MHR valuation distribution $\buyerdist$ such that the social welfare of the {\Nashsolution} is at most $\eps$ fraction of the optimal social welfare.
\end{restatable}

\begin{restatable}{theorem}{thmNashsolutionRegAntiMHR}
\label{thm:Nash solution welfare approximation reg anti-mhr}
    Consider a single-item market with $n \geq 1$ buyers whose values are drawn i.i.d.\ from a regular and anti-MHR distribution $\buyerdist$. The {\Nashsolution} achieves at least $\frac{n-1}{n}$ of the optimal social welfare.
\end{restatable}

To obtain a notion of Nash fairness that better aligns with the aggregate nature of the bargaining problem between seller and buyers---as in the {\KSsolution}, which operates on total buyers' surplus---we introduce the {\CSNashsolution}, which maximizes the product of the seller’s revenue and the \emph{total} buyers' surplus. This cross-side formulation restores symmetry between the two sides of the market and yields significantly more robust performance (in contrast to the negative result presented in \Cref{thm:Nash solution welfare approximation mhr}).

\begin{definition}[{\CSNashsolution}]
\label{def:cross-side nash solution}
    Fix a valuation distribution $\buyerdist$. A Pareto-optimal mechanism $\mech$ is a {\CSNashsolution} if it maximizes the product of the seller revenue and buyers' surplus, i.e.,
    \begin{align}
    \label{eq:cross-side nash fairness}
    \tag{\texttt{cross-side Nash social welfare}}
        \mech \in \argmax_{\mech\primed\in\mechfam}~
        \sellerexanteutil[\buyerdist]{\mech\primed}
        \cdot 
        \buyerexanteutil[\buyerdist]{\mech\primed}.
    \end{align}
\end{definition}

In the special case of a single buyer, the {\Nashsolution} and {\CSNashsolution} coincide. The tight welfare approximation guarantee of 50\% for this setting is established in \citet{BFM-25}. In \Cref{apx:nash solution analysis}, we further provide an improved welfare guarantee for MHR distributions (\Cref{thm:Nash solution welfare approximation single-buyer mhr}).
For general market sizes, the welfare guarantees of the {\CSNashsolution} closely mirror those of the {\KSsolution} across distributional classes.

\begin{restatable}{theorem}{thmCSNashsolution}
\label{thm:cross-side Nash solution welfare approximation}
Consider a single-item market with $n \geq 1$ buyers whose values are drawn i.i.d.\ from a distribution $\buyerdist$. Then the following hold:
\begin{enumerate}[label=(\roman*)]
    \item For any 
any valuation distribution $\buyerdist$, the social welfare of the {\CSNashsolution} is at least $50\%$ of the optimal social welfare.
    
    \item If $\buyerdist$ is both regular and anti-MHR, then the {\CSNashsolution} achieves at least $\frac{n-1}{n}$ of the optimal social welfare.
    
    \item For every $\eps > 0$ and any sufficiently large $n$, there exists an instance (\Cref{example:multi-buyer:mhr hard instance}) of $n$ buyers with MHR valuation distribution $\buyerdist$ such that the social welfare of the {\CSNashsolution} is at most $(50 + \eps)\%$ of the optimal social welfare.
\end{enumerate}
\end{restatable}

\bibliographystyle{apalike}
\bibliography{refs.bib}

\appendix

\section{Strong Budget Balance vs. Weak Budget Balance}
\label{sec:weak budget balance}
In this section, we introduce the concept of weak budget balance, by extending the definition of a mechanism. In general, a mechanism $\mech = (\allocs, \prices, \sellerprice)$ includes an allocation rule $\allocs = (\alloc_1, \dots, \alloc_n)$, a buyer payment rule $\prices = (\price_1, \dots, \price_n)$, and a seller payment rule $\sellerprice$. As previously defined, for each buyer $i \in [n]$, her allocation $\alloc_i \colon \reals_+^n \to [0, 1]$ and payment $\price_i \colon \reals_+^n \to \reals_+$ map the bids of all $n$ buyers to, respectively, the fraction of the item allocated to her and the monetary transfer she makes. 
Further, we involve an additional seller payment rule $\sellerprice \colon \reals_+^n \to \reals_+$, which maps the bid profile to the monetary transfer received by the seller. 

As standard in the literature~\citep{mye-81}, under this extended definition, it is without loss of generality to restrict attention to ex-post feasible, Bayesian incentive compatible (BIC), interim individually rational (IR), ex-ante (weak) budget balanced (WBB) mechanisms. Specifically, the last property is defined as follows. 
\begin{itemize}
    \item \emph{Ex-ante weak budget balance}: The total ex-ante payment collected from all buyers is at least the ex-ante payment sent to the seller, i.e., $\sexpect{\sum_{i \in [n]}\price_i(\vals)} \geq \sexpect{\sellerprice(\vals)}$. When $\sexpect{\sum_{i \in [n]}\price_i(\vals)} = \sexpect{\sellerprice(\vals)}$, the mechanism is said to satisfy \emph{ex-ante strong budget balance (SBB)}.
\end{itemize}

One observes that in the main body, we focus on SBB mechanisms, since all buyers' payments are directly transferred to the seller. We now show that this is without loss in characterizing the Pareto frontier, in the sense that every Pareto-optimal mechanism is SBB. 

\begin{lemma} \label{lem:weak budget balance:SBB}
    In the family of mechanisms that are ex-post feasible, BIC, interim IR, and \emph{WBB}, all Pareto-optimal mechanisms are \emph{SBB}.
\end{lemma}

\begin{proof}[Proof of \Cref{lem:weak budget balance:SBB}]
    Given any WBB mechanism that is not SBB, one can always transfer any money left to the mechanism to the seller in the ex-ante sense. This results in an SBB mechanism that Pareto-dominates the original mechanism, since seller revenue strictly increases while buyers' surplus does not change. This finishes the proof.
\end{proof}

As a consequence, since all our results in the main body focus on Pareto-optimal mechanisms, they can also be extended to the family of mechanisms that satisfy WBB.

\section{A Unified Proof of Results in \texorpdfstring{\Cref{sec:Pareto frontier}}{Section 3}}
\label{sec:proofs in Pareto frontier}
In this section, we provide a unified proof for all results in \Cref{sec:Pareto frontier}. 

\subsection{Preparations: Properties of the Mixed Welfare-Revenue Curve}

We start by studying the mixed welfare-revenue curve $\lagCurve$:  
\[\lagCurve(\quant) = \dualV\cdot \int_0^{\quant} \valFunc(\qquant) \dd \qquant + (1 - \dualV)\cdot \quant \valFunc(\quant).\]
We now provide some structural properties for $\lagCurve$ and its concave hull, which we denote by $\lagCurveCH$. 
We let 
\[\revCurveOPT(\dualV) \triangleq \max_{\quant \in [0, 1]} \lagCurve(\quant) = \max_{\quant \in [0, 1]} \lagCurveCH(\quant)\]
be the maximum of $\lagCurve(\quant)$ and $\lagCurveCH(\quant)$. Further, let 
\[\quantItvOPT(\dualV) \triangleq \argmax_{\quant \in [0, 1]} \lagCurveCH(\quant)\]
be the set of quantiles $\quant$ that maximize $\lagCurveCH(\quant)$. Note that since $\lagCurveCH$ is a concave curve, $\quantItvOPT(\dualV)$ is a singleton or a closed interval. 
We denote its infimum and supremum by $\leftS(\dualV)$ and $\rightS(\dualV)$ for $\dualV \geq 0$, respectively, satisfying $\rightS(\dualV) \geq \leftS(\dualV)$.

The following proposition characterizes the behavior of $\revCurveOPT(\dualV)$ and $\quantItvOPT(\dualV)$ as functions of $\dualV$.

\begin{proposition}\label{prop:proofs in Pareto frontier:dualV and quantItvOPT relation}
    We have the following:
    \begin{enumerate} 
        \item $\revCurveOPT(\dualV)$ is continuous on $[0, +\infty)$.
        \item For any $\dualV > \dualVPrime$, $\leftS(\dualV) \geq \leftS(\dualVPrime)$ and $\rightS(\dualV) \geq \rightS(\dualVPrime)$. 
        \item As a set-valued function of $\dualV$, $\quantItvOPT(\dualV)$ is upper hemicontinuous on $[0, 1]$.
        \item For any $\quant \in [\leftS(0), 1]$, there is some $\dualV \geq 0$ such that $\quant \in \quantItvOPT(\dualV)$. 
        \item If for some $\dualVPrime < \dualV$, $\leftS(\dualV) < \rightS(\dualVPrime)$, then $\rightS(\dualVPrime) = \rightS(\dualV) = 1$. 
    \end{enumerate}
\end{proposition}

\begin{proof}[Proof of \Cref{prop:proofs in Pareto frontier:dualV and quantItvOPT relation}]
    For the first statement, notice that for any $\quant$, $\dualV$, and $\dualVPrime$, 
    \begin{align*}
        |\lagCurve(\quant) - \lagCurve[\dualVPrime](\quant)| &= |\dualV - \dualVPrime|\cdot (\int_0^{\quant} \valFunc(\qquant) \dd \qquant - \quant \valFunc(\quant)) \\
        &= |\dualV - \dualVPrime|\cdot (\int_0^{\quant} (\valFunc(\qquant) - \valFunc(\quant)) \dd \qquant \\
        &\leq |\dualV - \dualVPrime|\cdot \int_0^{1} \valFunc(\qquant) \dd \qquant .
    \end{align*} 
    Therefore, 
    \begin{align*}
        \revCurveOPT(\dualV) - \revCurveOPT(\dualVPrime) &= \max_{\quant} \lagCurve(\quant) - \max_{\quant} \lagCurve[\dualVPrime](\quant) \\
        &\leq \max_{\quant} \lagCurve(\quant) - \max_{\quant} \inParentheses{\lagCurve(\quant) - |\dualV - \dualVPrime|\cdot \int_0^{1} \valFunc(\qquant) \dd \qquant} \\
        &\leq |\dualV - \dualVPrime|\cdot \int_0^{1} \valFunc(\qquant) \dd \qquant. 
    \end{align*}
    Similarly, $\revCurveOPT(\dualVPrime) - \revCurveOPT(\dualV) \leq |\dualV - \dualVPrime|\cdot \int_0^{1} \valFunc(\qquant) \dd \qquant$ also holds. 
    This concludes that $\revCurveOPT(\dualV)$ is Lipschitz continuous on $[0, +\infty)$. 

    For the second statement, it suffices to prove that for any $\dualV > \dualVPrime > 0$ and $1 > \quant > \quantPrime > 0$ , if $\lagCurve[\dualVPrime](\quant) \geq \lagCurve[\dualVPrime](\quantPrime)$, then $\lagCurve[\dualV](\quant) \geq \lagCurve[\dualV](\quantPrime)$ also holds. 
    Notice that 
    \[\lagCurve[\dualV](\quant) - \lagCurve[\dualVPrime](\quant) = (\dualV - \dualVPrime)\inParentheses{\int_0^{\quant} \valFunc(\qquant) \dd \qquant - \quant \valFunc(\quant)}, \]
    and that $\int_0^{\quant} \valFunc(\qquant) \dd \qquant - \quant \valFunc(\quant) = \int_0^{\quant} (\valFunc(\qquant) - \valFunc(\quant)) \dd \qquant$ is a non-decreasing function of $\quant$, we derive that
    \begin{align*}
        &\mathrel{\phantom{\Longleftrightarrow}} \lagCurve[\dualV](\quant) - \lagCurve[\dualVPrime](\quant) \geq \lagCurve[\dualV](\quantPrime) - \lagCurve[\dualVPrime](\quantPrime) \\
        &\Longleftrightarrow \lagCurve[\dualV](\quant) - \lagCurve[\dualV](\quantPrime) \geq \lagCurve[\dualVPrime](\quant) - \lagCurve[\dualVPrime](\quantPrime) \geq 0.
    \end{align*}

    For the third statement, let $\{\dualV[\subLim]\} \to \dualV[0]$ be an increasing sequence. Thus, by the third statement, $\{\rightS(\dualV[\subLim])\}$ is also a non-decreasing sequence, and converges to some $\rightS$ by the monotone convergence theorem. We know from the above that: 
    \[
    \lim_{\subLim \to \infty}\lagCurve[{\dualV[0]}](\rightS(\dualV[\subLim])) = 
    \lim_{\subLim \to \infty}\lagCurve[{\dualV[\subLim]}](\rightS(\dualV[\subLim])) = \lim_{\subLim \to \infty}\revCurveOPT(\dualV[\subLim]) = \revCurveOPT(\dualV[0]).
    \]
    Now observe that $\int_0^{\quant} \valFunc(\qquant) \dd \qquant$ is a continuous function, and since $\valFunc(\quant)$ is a non-increasing left-continuous function, $\quant \valFunc(\quant)$ is also left-continous, and if $\{\rightS(\dualV[\subLim])\} \to \rightS$, $\lim_{\subLim \to \infty} \rightS(\dualV[\subLim])\valFunc(\rightS(\dualV[\subLim])) \leq \rightS\valFunc(\rightS)$.
    Therefore, 
    \[
    \lagCurve[{\dualV[0]}](\rightS) \geq 
    \lim_{\subLim \to \infty}\lagCurve[{\dualV[0]}](\rightS(\dualV[\subLim])) = 
    \revCurveOPT(\dualV[0]),
    \]
    which implies $\rightS \in \quantItvOPT(\dualV[0])$. Similarly, $\{\leftS(\dualV[\subLim])\}$ converges to some $\leftS \in \quantItvOPT(\dualV[0])$. Thus, for any $\{\quant[\subLim]\} \to \quant[0]$ that satisfy $\quant[\subLim] \in \quantItvOPT(\dualV[\subLim])$, $\leftS \leq \quant[0] \leq \rightS$, indicating that $\quant[0] \in \quantItvOPT(\dualV[0])$. This proves the upper hemicontinuity of $\quantItvOPT(\dualV)$.

    For the fourth statement, we know that $\rightS(1) = 1$. Since $\quantItvOPT(\dualV)$ is upper hemicontinuous, we know that for any $\quant \in [\leftS(0), \rightS(1) = 1]$, the sets $\{\dualV \in [0, 1]\suchthat \leftS(\dualV) \leq \quant\}$ and $\{\dualV \in [0, 1]\suchthat \rightS(\dualV) \geq \quant\}$ are non-empty and closed. 
Now let 
    \begin{align*}
        \dualVPlus \triangleq \sup \{\dualV \in [0, 1]\suchthat \leftS(\dualV) \leq \quant\}, \quad \dualVMinus \triangleq \inf \{\dualV \in [0, 1]\suchthat \rightS(\dualV) \geq \quant\}.
    \end{align*}
    We show $\dualVPlus \geq \dualVMinus$. Or else, for any $\dualV \in (\dualVPlus, \dualVMinus)$, $\leftS(\dualV) > \quant > \rightS(\dualV)$, contradicting that $\leftS(\dualV) \leq \rightS(\dualV)$. 
    Therefore, for any $\dualVMinus \leq \dualV \leq \dualVPlus$, by the second statement, $\leftS(\dualV) \leq \leftS(\dualVPlus) \leq \quant \leq \rightS(\dualVMinus) \leq \rightS(\dualV)$, or that $\quant \in \quantItvOPT(\dualV)$. 
    
    At last, for the last statement, we know that $\lagCurve[\dualVPrime](\leftS(\dualV)) \leq \lagCurve[\dualVPrime](\rightS(\dualVPrime))$ and $\lagCurve(\leftS(\dualV)) \geq \lagCurve(\rightS(\dualVPrime))$ hold simultaneously. Since $\leftS(\dualV) < \rightS(\dualVPrime)$, by a similar analysis of the second statement, both inequalities hold with equality. Thus, 
    \begin{align*}
        0 &= [\lagCurve(\rightS(\dualVPrime)) - \lagCurve(\leftS(\dualV))] - [\lagCurve[\dualVPrime](\rightS(\dualVPrime)) - \lagCurve[\dualVPrime](\leftS(\dualV))] \\
        &= (\dualV - \dualVPrime) \cdot \inParentheses{\int_0^{\rightS(\dualVPrime)} (\valFunc(\qquant) - \valFunc(\rightS(\dualVPrime)) \dd \qquant) - \int_0^{\leftS(\dualV)} (\valFunc(\qquant) - \valFunc(\leftS(\dualV)) \dd \qquant)} \\
        &= (\dualV - \dualVPrime) \cdot \inParentheses{\int_0^{\leftS(\dualV)} (\valFunc(\leftS(\dualV)) - \valFunc(\rightS(\dualVPrime)) \dd \qquant) + \int_{\leftS(\dualV)}^{\rightS(\dualVPrime)} (\valFunc(\qquant) - \valFunc(\rightS(\dualVPrime)) \dd \qquant)}. 
    \end{align*}
    Consequently, we derive that $\valFunc(\leftS(\dualV)) = \valFunc(\rightS(\dualVPrime))$. As a result, 
    \begin{align*}
        0 = \lagCurve(\rightS(\dualVPrime)) - \lagCurve(\leftS(\dualV)) = (\rightS(\dualVPrime) - \leftS(\dualV))\cdot \valFunc(\leftS(\dualV)), 
    \end{align*}
    leading to $\valFunc(\leftS(\dualV)) = 0$. As such, $\lagCurve[\dualVPrime](\rightS(\dualVPrime)) = \lagCurve[\dualVPrime](1)$ and $\lagCurve(\leftS(\dualV)) = \lagCurve(1)$ hold.
    Therefore, $\rightS(\dualVPrime) = \rightS(\dualV) = 1$.
\end{proof}

Furthermore, the following result establishes the trend of the difference between $\lagCurve$ and $\lagCurveCH$, when $\dualV$ increases from 0.

\begin{proposition} \label{prop:proofs in Pareto frontier:ironed intervals shrink expand with dualV}
    We have the following: 
    \begin{itemize}
        \item For $0 \leq \dualVPrime < \dualV < 1$, if $\lagCurve(\quant) < \lagCurveCH(\quant)$, then $\lagCurve[\dualVPrime](\quant) < \lagCurveCH[\dualVPrime](\quant)$ holds. 
        \item For $1 < \dualV < \dualVPrime$, $\lagCurve(\quant) < \lagCurveCH(\quant)$, then $\lagCurve[\dualVPrime](\quant) < \lagCurveCH[\dualVPrime](\quant)$ holds. 
    \end{itemize}
\end{proposition}

\begin{proof}[Proof of \Cref{prop:proofs in Pareto frontier:ironed intervals shrink expand with dualV}]
    For the first statement, we show that if there exists $\quantPrime$ and $\quantDoublePrime$ with $\quant = \interpol\cdot \quantPrime + (1 - \interpol)\cdot \quantDoublePrime$ for some $0 < \interpol < 1$ and $\lagCurve(\quant) < \interpol\cdot \lagCurve(\quantPrime) + (1 - \interpol)\cdot \lagCurve(\quantDoublePrime)$ holds, then $\lagCurve[\dualVPrime](\quant) < \interpol\cdot \lagCurve[\dualVPrime](\quantPrime) + (1 - \interpol)\cdot \lagCurve[\dualVPrime](\quantDoublePrime)$ also holds. To see this, one can regard $\lagCurve$ as a convex combination of $\lagCurve[1]$ and $\lagCurve[\dualVPrime]$. Now we know that $\lagCurve[1](\quant) \geq \interpol\cdot \lagCurve[1](\quantPrime) + (1 - \interpol)\cdot \lagCurve[1](\quantDoublePrime)$ establishes because $\lagCurve[1]$ is concave. This immediately leads to $\lagCurve[\dualVPrime](\quant) < \interpol\cdot \lagCurve[\dualVPrime](\quantPrime) + (1 - \interpol)\cdot \lagCurve[\dualVPrime](\quantDoublePrime)$. The second statement is proved anologously. 
\end{proof}

Intuitively, \Cref{prop:proofs in Pareto frontier:ironed intervals shrink expand with dualV} says that when $\dualV$ increases from $0$, the ironed intervals of the mixed welfare-revenue curve $\lagCurve$ ``shrink'' when $\dualV \in [0, 1]$ and ``expand'' when $\dualV \in [1, +\infty)$. In specific, one observes that when $\dualV = 1$, $\lagCurve[1]$ itself is concave. 
To better characterize what ``shrink'' and ``expand'' mean, we let $\ironedItvs(\dualV)$ denote the set containing all quantiles $\quant \leq \rightS(\dualV)$ such that $\lagCurveCH(\quant) > \lagCurve(\quant)$ holds. Then for any $\dualV$, $\ironedItvs(\dualV)$ is a union of \emph{open intervals}. We have the following. 
\begin{corollary} \label{coro:proofs in Pareto frontier:ironedItvs trend}
    For any $\dualVPrime < \dualV < 1$ or $1 < \dualV < \dualVPrime$, $\ironedItvs(\dualV) \subseteq \ironedItvs(\dualVPrime)$. 
\end{corollary}

We further let $\residuItv(\dualV) \triangleq \{\quant \suchthat \rightS(\dualV) < \quant < 1\}$. As a corollary of \Cref{prop:proofs in Pareto frontier:dualV and quantItvOPT relation}, we have
\begin{corollary} \label{coro:proofs in Pareto frontier:residuItv trend}
    For any $\dualVPrime < \dualV$, $\residuItv(\dualV) \subseteq \residuItv(\dualVPrime)$. 
\end{corollary}

\subsection{First Step: A General Lagrangian Analysis}

To analyze the setting, we consider the interim space. With a little bias of notation, we use $\alloc(\quant)$ to denote any buyer's interim allocation probability when its value is $\valFunc(\quant)$. We further use $\accAlloc(\quant)$ to denote the accumulated allocation, i.e., $\accAlloc(\quant) \triangleq \int_0^{\quant} \alloc(\qquant) \dd \qquant$. By Border's condition~\citep{bor-91}, an interim allocation function $\alloc(\quant)$ on the quantile space can be ex post realized by a BIC and interim IR mechanism if and only if
\begin{gather*}
    \accAlloc(\quant) \leq \frac{1 - (1 - \quant)^{\numAgents}}{\numAgents}, \quad \forall \quant \in [0, 1]; \qquad
    \alloc(\quant) \leq \frac{1}{\numAgents}, \quad \forall \quant \in [0, 1]; \qquad
    \alloc(\quant) \text{ is decreasing. }
\end{gather*}
Specifically, in the single buyer case, the first constraint becomes $\accAlloc(\quant) \leq \quant$, which is directly implied by the second constraint $\alloc(\quant) \leq 1$ and non-restrictive. But this is not the case for the multi-buyer case. 

Back to our problem, according to the revenue equivalence theorem, the expected revenue under a mechanism with interim allocation $\alloc(\quant)$ is $\expect[\quant]{\alloc(\quant)\cdot \revCurvePrime(\quant)}$, which is derived by seeing $\alloc(\quant)$ as a randomization of price-posting mechanisms~\citep{mye-81}.
Meanwhile, we give buyers' surplus under $\alloc$ as follows. 
\begin{proposition} \label{prop:proofs in Pareto frontier:surplus expression}
    Under a BIC and interim IR mechanism with interim allocation $\alloc$ and its accumulation $\accAlloc$, a buyer's ex ante utility is 
    \begin{align*}
        &\mathrel{\phantom{=}} \expect[\quant]{\alloc(\quant)\cdot (\valFunc(\quant) - \revCurvePrime(\quant))} \\
        &= \expect[\quant]{\alloc(\quant)\cdot (-\quant \valFuncPrime(\quant))} \\
        &= \expect[\quant]{\accAlloc(\quant)\cdot (\quant \valFuncPrime(\quant))'} - \accAlloc(1)\cdot \valFuncPrime(1) \\
        &= \expect[\quant]{(\accAlloc(1) - \accAlloc(\quant))\cdot (-\quant \valFuncPrime(\quant))'} + \accAlloc(1)\cdot \inBrackets{\expect[\quant]{(\quant \valFuncPrime(\quant))'} - \valFuncPrime(1)}. 
    \end{align*}
\end{proposition}

Therefore, our problem can be formalized as follows:

\begin{gather*}
    \max_{\alloc} \numAgents\cdot \expect[\quant]{\alloc(\quant)\cdot \revCurvePrime(\quant)}, \\
    \subjectto \numAgents\cdot \expect[\quant]{\alloc(\quant)\cdot (\valFunc(\quant) - \revCurvePrime(\quant))} \geq \buyerutil, \\
    \accAlloc(\quant) \leq \frac{1 - (1 - \quant)^{\numAgents}}{\numAgents}, \quad \forall \quant \in [0, 1]; \qquad
    \alloc(\quant) \leq \frac{1}{\numAgents}, \quad \forall \quant \in [0, 1]; \qquad
    \alloc(\quant) \text{ is decreasing. }
\end{gather*}

For brevity, we let $\accAllocUb(\quant) \triangleq \nicefrac{1 - (1 - \quant)^{\numAgents}}{\numAgents}$ and $\allocUb(\quant) \triangleq \nicefrac{1}{n}$. 
We Lagrangify the surplus constraint to derive the following dual program: 
\[\min_{\dualV \geq 0} \max_{\accAlloc(\quant) \leq \accAllocUb(\quant), \; \alloc(\quant) \leq \allocUb(\quant), \; \alloc(\quant) \text{ is decreasing}} \numAgents\cdot \expect[\quant]{\alloc(\quant)\cdot (\dualV\cdot \valFunc(\quant) + (1 - \dualV)\cdot \revCurvePrime(\quant))} - \dualV\cdot \buyerutil.\]
In particular, notice that $\dualV\cdot \valFunc(\quant) + (1 - \dualV)\cdot \revCurvePrime(\quant) = \lagCurvePrime(\quant)$. 

We do the following steps. 
First, for each $\dualV \geq 0$, we derive the allocation $\alloc(\quant)$ that maximizes the primal objective. 
Further, we establish the existence of a $\dualVOPT \geq 0$ such that the complementary slackness condition is met. They together solve the problem. 
 
In general, under the accumulated allocation constraint, the optimal primal solution is characterized by the following~\citep{mye-81}.
\begin{proposition} \label{prop:proofs in Pareto frontier:optimal primal allocation}
    Both the optimal solution and the optimal objective of 
    \[\max_{\accAlloc(\quant) \leq \accAllocUb(\quant), \; \alloc(\quant) \leq \allocUb(\quant), \; \alloc(\quant) \text{ is decreasing}} \expect[\quant]{\alloc(\quant)\cdot \lagCurvePrime(\quant)}\]
    coincides, respectively, with the optimal solution and the optimal objective of 
    \begin{gather*}
        \max_{\accAlloc(\quant) \leq \accAllocUb(\quant), \; \alloc(\quant) \leq \allocUb(\quant), \; \alloc(\quant) \text{ is decreasing}} \expect[\quant]{\alloc(\quant) \cdot \lagCurveCHPrime(\quant)} \\
        = \max_{\accAlloc(\quant) \leq \accAllocUb(\quant), \; \alloc(\quant) \leq \allocUb(\quant), \; \alloc(\quant) \text{ is decreasing}} \expect[\quant]{-\accAlloc(\quant) \cdot \lagCurveCHDoublePrime(\quant)} + \accAlloc(1)\cdot \lagCurveCHPrime(1), \\
        \subjectto \allocPrime(\quant) = 0 \text{, if }\lagCurveCH(\quant) > \lagCurve(\quant). 
    \end{gather*}
\end{proposition}

\subsection{Core: A Continuous and Monotone Evolving Process of the Primal Optimal Accumulated Allocation} \label{sec:proofs in Pareto frontier:evolving process}

In the light of \Cref{coro:proofs in Pareto frontier:ironedItvs trend,coro:proofs in Pareto frontier:residuItv trend}, we now define a process of the primal optimal accumulated allocation $\accAllocOPT[\dualV, \evol](\quant)$ according to \Cref{prop:proofs in Pareto frontier:ironed intervals shrink expand with dualV} when $\dualV$ increases from 0. 
Here, $\evol \in [0, 1]$ is an auxiliary parameter with full support $[0, 1]$ for only a countable many $\dualV$s and a unique value $0$ for other $\dualV$s. 
Our hope for this process is that when the pair $(\dualV, \evol)$ increases from $(0, 0)$ in a lexicographical order, buyers' surplus under the mechanism following this process increases continuously and monotonically from $\buyerexanteutil[\dist]{\Mye}$ to $\buyerbenchmarkDist$.
\footnote{Since $\evol$ has support $[0, 1]$ for only countably many $\dualV$, continuity on the increasing process of $(\dualV, \evol)$ can be defined by establishing a continuous and monotone bijection between all pairs $(\dualV, \evol)$ and $\reals_+$.} 
Therefore, for any $\buyerutil \in [\buyerexanteutil[\dist]{\Mye}, \buyerbenchmarkDist]$, there is a $(\dualVOPT, \evolOPT)$ such that buyers' surplus under $\accAllocOPT[\dualVOPT, \evolOPT](\quant)$ is $\buyerutil$, and the complementary slackness condition is met. 

With the aim, an important property we need is that for any $\quant$, $\accAllocOPT[\dualV, \evol](\quant)$ is continuous in the increasing process of $(\dualV, \evol)$. The sub-processes for $\dualV \leq 1$ and $\dualV > 1$ are highly symmetric, and we start with $\dualV \leq 1$. 
We define the \emph{canonical} primal optimal solution $\accAllocOPT[\dualV](\quant)$ for the dual variable $\dualV$ as the following:
\begin{itemize}
    \item For $\quant \in [0, \rightS(\dualV)]$:
    \begin{itemize}
        \item If $\quant \notin \ironedItvs(\dualV)$, $\accAllocOPT[\dualV](\quant) = \accAllocUb(\quant)$; 
        \item If $\quant \in \ironedItvs(\dualV)$, suppose it lies in the open interval $(\quantDag, \quantDoubleDag) \subseteq \ironedItvs(\dualV)$, and $\quant = \mu\cdot \quantDag + (1 - \mu)\cdot \quantDoubleDag$ for some $0 < \mu < 1$, then $\accAllocOPT[\dualV](\quant) = \mu\cdot \accAllocUb(\quantDag) + (1 - \mu)\cdot \accAllocUb(\quantDoubleDag)$. 
    \end{itemize}
    \item For $\quant \in (\rightS(\dualV), 1]$, $\accAllocOPT[\dualV](\quant) = \accAllocOPT[\dualV](\rightS(\dualV))$. 
\end{itemize}

By applying \Cref{prop:proofs in Pareto frontier:optimal primal allocation}, one can verify that the above accumulated allocation rule is indeed an optimal primal solution. 
Our intuition for constructing the desired continuous process is to transfer from a canonical primal optimal solution to another canonical one as $\dualV$ increases. However, as $\dualV$ increases, the canonical primal optimal solution itself may not be continuous for some $\quant$. 
To see this, for any $0 < \dualV \leq 1$, according to \Cref{coro:proofs in Pareto frontier:ironedItvs trend,coro:proofs in Pareto frontier:residuItv trend}, we can define 
\begin{align*}
    \ironedItvs(\dualVMinus) \triangleq \lim_{\dualVPrime \to \dualVMinus} \ironedItvs(\dualVPrime); \quad \residuItv(\dualVMinus) \triangleq \lim_{\dualVPrime \to \dualVMinus} \residuItv(\dualVPrime). 
\end{align*}
We therefore have $\ironedItvs(\dualV) \subseteq \ironedItvs(\dualVMinus)$ and $\residuItv(\dualV) \subseteq \residuItv(\dualVMinus)$. 

Now, an equivalent condition for the canonical primal optimal allocation itself to be continuous at $\dualV$ for every $\quant$ is that both $\ironedItvs(\dualV) = \ironedItvs(\dualVMinus)$ and $\residuItv(\dualV) = \residuItv(\dualVMinus)$ hold.
When these conditions hold, the support of $\evol$ is the singleton $\{0\}$ with $\dualV$ and $\accAllocOPT[\dualV, 0](\quant)$ is just the canonical primal optimal allocation at $\dualV$. 
In fact, one observes that this happens almost everywhere except a countably set of $\dualV$s. 
And for any other $\dualV$, we need an auxiliary process at $\dualV$, parameterized by $\evol \in [0, 1]$, to continuously transfer from the canonical primal optimal solution at $\dualVMinus$ to that of $\dualV$, \footnote{The canonical primal optimal solution at $\dualVMinus$ means $\lim_{\dualV' \to \dualVMinus} \accAllocOPT[\dualV']$.} meanwhile keeping the solution always primal optimal at $\dualV$. 
On this side, let $\accAllocOPT[\dualV, 0](\quant)$ be the canonical primal optimal solution at $\dualVMinus$, and $\accAllocOPT[\dualV, 1](\quant)$ be that at $\dualV$. Our construction is in fact simple: For any $\evol \in (0, 1)$ and any $\quant$, \[\accAllocOPT[\dualV, \evol](\quant) \triangleq \evol\cdot \accAllocOPT[\dualV, 0](\quant) + (1 - \evol)\cdot \accAllocOPT[\dualV, 1](\quant).\]

By definition, for any $\quant$, continuity of $\accAllocOPT[\dualV, \evol](\quant)$ in $\evol$ is natural. It suffices to prove the following lemma.
\begin{lemma} \label{lem:proofs in Pareto frontier:primal optimal allocation evolvement}
    $\accAllocOPT[\dualV, \evol]$ is a primal optimal solution at $\dualV$ for any $\evol \in [0, 1]$.
\end{lemma} 

\begin{proof}[Proof of \Cref{lem:proofs in Pareto frontier:primal optimal allocation evolvement}]
    To prove \Cref{lem:proofs in Pareto frontier:primal optimal allocation evolvement}, first notice that all of
    \[\accAllocOPT[\dualV, \evol](\quant) \leq \accAllocUb(\quant), \quad \allocOPT[\dualV, \evol](\quant) \leq \allocUb(\quant), \quad \allocOPT[\dualV, \evol](\quant) \text{ is decreasing}\]
    hold. 
    Meanwhile, the accumulated allocation $\accAllocOPT[\dualV, \evol](\quant)$ of any $\quant$ outside $\ironedItvs(\dualVMinus)$ and $\residuItv(\dualVMinus)$ keeps $\accAllocUb(\quant)$ and is optimal. 
    we now consider the optimality for $\quant$ in $\ironedItvs(\dualVMinus)$ and $\residuItv(\dualVMinus)$, respectively. 
    
    Starting with $\ironedItvs$, since both $\ironedItvs(\dualV)$ and $\ironedItvs(\dualVMinus)$ are union of open intervals, we claim the following:
    \begin{claim}
        Let $(\quantDag, \quantDoubleDag)$ be any open interval in $\ironedItvs(\dualVMinus)$. For any $\quantDag \leq \quant \leq \quantDoubleDag$ but $\quant \notin \ironedItvs(\dualV)$, $(\quant, \lagCurve(\quant))$ is on the line between $(\quantDag, \lagCurve(\quantDag))$ and $(\quantDoubleDag, \lagCurve(\quantDoubleDag))$. 
    \end{claim}
    To prove the claim, notice that by \Cref{prop:proofs in Pareto frontier:ironed intervals shrink expand with dualV}, such $\lagCurve[\dualVPrime](\quant) < \lagCurveCH[\dualVPrime](\quant)$ holds for any $\dualVPrime < \dualV$. By an argument of continuity, we know that $(\quant, \lagCurve(\quant))$ is not above the line between $(\quantDag, \lagCurve(\quantDag))$ and $(\quantDoubleDag, \lagCurve(\quantDoubleDag))$. 
    Since $\quant$ is also not ironed at $\lagCurve$, i.e., not below the line, it has to be on the line. 
    
    As a result, we know that $\lagCurveCH$ is a straight line between $\quantDag$ and $\quantDoubleDag$. Thus, $\lagCurveCHDoublePrime(\quant) = 0$ in $(\quantDag, \quantDoubleDag)$, and the value of $\accAllocOPT$ does not affect the final objective, according to \Cref{prop:proofs in Pareto frontier:optimal primal allocation}. 
    Further, for any $\quant \in \ironedItvs(\dualV)$, $(\allocOPT[\dualV, 0])'(\quant) = 0$ and $(\allocOPT[\dualV, 1])'(\quant) = 0$ both hold because of primal optimality, and thus $(\allocOPT[\dualV, \evol])'(\quant) = 0$ also establishes for any $\evol \in [0, 1]$. 
    The above states the primal optimality of $\allocOPT[\dualV, \evol]$ in $\ironedItvs(\dualVMinus)$. 
    
    For quantiles on $\residuItv$, we have a similar claim:
    \begin{claim}
        If $\rightS(\dualVMinus) < \rightS(\dualV)$, then $\lagCurve(\rightS(\dualVMinus)) = \lagCurve(\rightS(\dualV))$. 
    \end{claim}
    This claim is also proved by an argument of continuity. This implies that $\lagCurveCHDoublePrime(\quant) = 0$ in $(\rightS(\dualVMinus), \rightS(\dualV))$, and therefore $\allocOPT[\dualV, \evol](\quant) = 0$ is also primal optimal on $\residuItv(\dualVMinus)$. We finish the proof of \Cref{lem:proofs in Pareto frontier:primal optimal allocation evolvement}.
\end{proof}

Meanwhile, it is not hard to observe that $\accAllocOPT[\dualV, \evol](\quant)$ is monotonically increasing in $(\dualV, \evol)$ for any $\quant$. 
Analogously, we can define such a process for $\dualV > 1$. In summary, the process goes as follows:
\begin{enumerate}
    \item For $\dualV \leq 1$, define $\ironedItvs(\dualVMinus) \triangleq \lim_{\dualVPrime \to \dualVMinus} \ironedItvs(\dualVPrime)$; for $\dualV > 1$, define $\ironedItvs(\dualVPlus) \triangleq \lim_{\dualVPrime \to \dualVPlus} \ironedItvs(\dualVPrime)$. Define $\residuItv(\dualVMinus) \triangleq \lim_{\dualVPrime \to \dualVMinus} \residuItv(\dualVPrime)$.
    \item If $\ironedItvs(\dualVMinus) = \ironedItvs(\dualV)$ and $\residuItv(\dualVMinus) = \residuItv(\dualV)$ for $\dualV \leq 1$ or $\ironedItvs(\dualVPlus) = \ironedItvs(\dualV)$ and $\residuItv(\dualVMinus) = \residuItv(\dualV)$ for $\dualV > 1$, then the support of $\evol$ is the singleton $\{0\}$ with $\dualV$ and $\accAllocOPT[\dualV, 0](\quant)$ is the canonical primal optimal allocation at $\dualV$. 
    \item Otherwise, for $\dualV \leq 1$, let $\accAllocOPT[\dualV, 0](\quant)$ be the canonical primal optimal allocation at $\dualVMinus$ and $\accAllocOPT[\dualV, 1](\quant)$ be that at $\dualV$; for $\dualV > 1$, let $\accAllocOPT[\dualV, 0](\quant)$ be the canonical primal optimal allocation at $\dualV$ and $\accAllocOPT[\dualV, 1](\quant)$ be that at $\dualVPlus$. Then for $\evol \in [0, 1]$ and any $\quant$, 
    \[\accAllocOPT[\dualV, \evol](\quant) \triangleq (1 - \evol)\cdot \accAllocOPT[\dualV, 0](\quant) + \evol\cdot \accAllocOPT[\dualV, 1](\quant).\]
\end{enumerate}

We conclude by presenting our main result for this subsection. 
\begin{proposition} \label{prop:proofs in Pareto frontier:properties for the evolving process}
    Under the process $\{\accAllocOPT[\dualV, \evol]\}$, the following hold:
    \begin{enumerate}
        \item For any $(\dualV, \evol)$, $\accAllocOPT[\dualV, \evol]$ is a primal optimal solution at $\dualV$. 
        \item Each buyer's surplus under $\accAllocOPT[0, 0]$ is $\buyerexanteutil[\dist]{\Mye}$. 
        \item For any $\quant$, $\accAllocOPT[\dualV, \evol](\quant)$ is continuous in $(\dualV, \evol)$. 
        \item For any $\quant$, $\accAllocOPT[\dualV, \evol](\quant)$ is monotonically increasing in $\dualV$ when $\dualV \in [0, 1]$. 
        \item For any $\quant$, $\accAllocOPT[\dualV, \evol](\quant)$ is monotonically decreasing in $\dualV$ when $\dualV \geq 1$.
    \end{enumerate}
\end{proposition}

\begin{proof}[Proof of \Cref{prop:proofs in Pareto frontier:properties for the evolving process}]
    We are only left with the last statement, which establishes by noticing that $\rightS(\dualV) = 1$ for any $\dualV \geq 1$ and thus $\residuItv(\dualV) = \emptyset$. 
\end{proof}

\subsection{Single-Buyer Market}

In the single-buyer market, we prove the following two results. 

\PFSingleGeneral*

\PFSingleRegular*

To start with, we notice that in the single-buyer market, the dual program becomes the following:
\[\min_{\dualV \geq 0} \max_{\alloc(\quant) \leq 1; \alloc(\quant) \text{ is decreasing}} \numAgents\cdot \expect[\quant]{\alloc(\quant)\cdot \lagCurvePrime(\quant)} - \dualV\cdot \buyerutil.\]

We now show that the optimal dual variable $\dualV$ is upper bounded by $1$.

\begin{proposition}\label{prop:proofs in Pareto frontier:single buyer dualVOPT bounded by 1}
    In a single-item seller, single-buyer market, there exists a $\dualV \in [0, 1]$ that minimizes 
    \[\min_{\dualV \geq 0} \max_{\alloc(\quant) \leq 1, \alloc(\quant) \text{ is decreasing}} \numAgents\cdot \expect[\quant]{\alloc(\quant)\cdot \lagCurvePrime(\quant)} - \dualV\cdot \buyerutil.\] 
\end{proposition}

\begin{proof}[Proof of \Cref{prop:proofs in Pareto frontier:single buyer dualVOPT bounded by 1}]
    By \Cref{prop:proofs in Pareto frontier:optimal primal allocation}, Given any $\dualV \geq 0$, \[\max_{\alloc(\quant) \leq 1, \alloc(\quant) \text{ is decreasing}} \expect[\quant]{\alloc(\quant)\cdot \lagCurvePrime(\quant)} = \max_{\alloc(\quant) \leq 1, \alloc(\quant) \text{ is decreasing}} \expect[\quant]{\alloc(\quant)\cdot \lagCurveCHPrime(\quant)}. \] 
    For the latter formula, one observes that the maximum of $\expect[\quant]{\alloc(\quant)\cdot \lagCurveCHPrime(\quant)}$ under the constraints that $\alloc(\quant) \leq 1$ and $\alloc(\quant)$ is decreasing is just the maximum of $\lagCurveCH$, which coincides with the maximum of $\lagCurve$. 
    
    Notice that when $\dualV = 1$, the maximum quantile of $\lagCurve$ is $\quant = 1$, since $\lagCurve[1](\quant) = \int_0^{\quant} \valFunc(\qquant) \dd \qquant$ is an increasing function. 
    Thus, by \Cref{prop:proofs in Pareto frontier:dualV and quantItvOPT relation}, for any $\dualV > 1$, the maximum quantile of $\lagCurve$ is also $\quant = 1$, and the corresponding maximum objective is $\dualV (\int_0^1 \valFunc(\qquant) \dd \qquant - \buyerutil)$. 
    Since $\int_0^1 \valFunc(\qquant) \dd \qquant = \buyerbenchmarkDist \geq \buyerutil$, we conclude that $\dualV \in [0, 1]$ suffices to minimize the program. 
\end{proof}

We are now ready to prove our results. 
We notice that each \emph{extreme point} mechanism of $\{\accAllocOPT[\dualV, \evol]\}_{0 \leq \dualV \leq 1}$ satisfies $\evol = 0, 1$, and in the single-buyer market, each such mechanism is a price-posting mechanism, with the posted-price no larger than the lowest monopoly reserve. This finishes the proof of \Cref{prop:Pareto frontier:single buyer general}. 

As a corollary, we have the following. 
\begin{corollary} \label{coro:proofs in Pareto frontier:single buyer full characterization}
    In a single-item seller, single-buyer market, for any $0 \leq \dualV \leq 1$, let $\buyerutil\primed$ be the buyer surplus under the mechanism of posting price $\valFunc(\leftS(\dualV))$ and $\buyerutil\doubleprimed$ be that under the mechanism of posting price $\valFunc(\rightS(\dualV))$. If $\buyerutil\primed \leq \buyerutil \leq \buyerutil\doubleprimed$, then the optimal seller revenue with a buyer surplus of at least $\buyerutil$ is given by the mechanism of fixing a probability $\mixprob \in [0, 1]$ ex ante, posting price $\valFunc(\leftS(\dualV))$ with probability $\mixprob$ and posting price $\valFunc(\rightS(\dualV))$ with probability $1 - \mixprob$. $\mixprob$ satisfies that: 
    \[\buyerutil = \mixprob\cdot \buyerutil\primed + (1 - \mixprob)\cdot \buyerutil\doubleprimed. \]
\end{corollary}

We now consider regular value distributions.
Notice that although the primal optimal mechanism we derive is a randomization over two price-posting mechanisms, there is a single price-posting mechanism (with a price also no larger than the lowest monopoly reserve) that also achieves buyers' surplus of $\buyerutil$. 
Let the price be $\valFunc(\quantDag)$. 
By \Cref{prop:proofs in Pareto frontier:dualV and quantItvOPT relation}, there is a $\dualVDag$ such that $\quantDag \in \quantItvOPT(\dualVDag)$, or that $\quantDag \in \argmax_{\quant} \lagCurveCH(\quant)$. 
Yet, for regular distributions, $\lagCurveCH = \lagCurve$, implying that the mechanism of posting price $\valFunc(\quantDag)$ is also a primal optimal solution for dual variable $\dualVDag$. 
This proves \Cref{prop:Pareto frontier:single buyer regular}.

\subsection{Anti-MHR Distributions}

When the buyers' value distribution is anti-MHR, we prove the following two results.

\PFRegularAntiMHR*

\PFAntiMHR*

We first introduce the following result given by \citet{HR-08}. 
\begin{proposition}[Corollary 2.12, \citealp{HR-08}] \label{prop:proofs in Pareto frontier:anti-MHR SPA gives umax}
    For homogeneous buyers with an anti-MHR value distribution, $\buyerbenchmarkDist$ is achieved with a second-price auction.
\end{proposition}

Although in general, $\dualVOPT \leq 1$ does not hold for the multi-buyer case, it holds for anti-MHR distributions. A key observation is that when $\dualV \geq 1$, the mixed welfare-revenue curve $\lagCurve$ remains concave, so the optimal primal allocation does not change.
\begin{proposition} \label{prop:proofs in Pareto frontier:dualVOPT bounded by 1 for anti-MHR}
    In a single-item seller, multi-buyer market, when the value distribution is anti-MHR, there exists a $\dualV \in [0, 1]$ that minimizes
    \[\max_{\accAlloc(\quant) \leq \accAllocUb(\quant), \; \alloc(\quant) \leq \allocUb(\quant), \; \alloc(\quant) \text{ is decreasing}} \numAgents\cdot \expect[\quant]{\alloc(\quant)\cdot \lagCurvePrime(\quant)} - \dualV\cdot \buyerutil\]
    for all $\dualV \geq 0$. 
\end{proposition}

\begin{proof}[Proof of \Cref{prop:proofs in Pareto frontier:dualVOPT bounded by 1 for anti-MHR}]
    We first show that for anti-MHR distributions, $\lagCurve$ is concave and increasing for all $\dualV \geq 1$. To see that, 
    \begin{align*}
        \lagCurvePrime(\quant) &= \dualV \cdot \valFunc(\quant) + (1 - \dualV) (\valFunc(\quant) + \quant \valFuncPrime(\quant)) \\
        &= \valFunc(\quant) + (\dualV - 1)(-\quant \valFuncPrime(\quant)). 
    \end{align*}
    When $\dualV \geq 1$, since both $\valFunc(\quant)$ and $-\quant \valFuncPrime(\quant)$ are non-increasing and non-negative, $\lagCurvePrime(\quant)$ is non-increasing and non-negative. Thus $\lagCurve$ is concave and increasing, and coincides with $\lagCurveCH$. 
    As a consequence of \Cref{prop:proofs in Pareto frontier:optimal primal allocation}, the optimal accumulated allocation is $\accAllocOPT(\quant) = \accAllocUb(\quant)$ for all $\quant \in [0, 1]$, corresponding to a second-price auction.
    Thus, when $\dualV \geq 1$, the optimal primal objective is 
    \begin{align*}
        &\mathrel{\phantom{=}} \numAgents\cdot \expect[\quant]{\allocOPT(\quant)\cdot \lagCurvePrime(\quant)} - \dualV\cdot \buyerutil \\
        &= \numAgents\cdot \expect[\quant]{\allocOPT(\quant) \valFunc(\quant)} - \buyerutil + (\dualV - 1) \inParentheses{\numAgents\cdot \expect[\quant]{\allocOPT(\quant)\cdot (\valFunc(\quant) - \revCurvePrime(\quant))} - \buyerutil} \\
        &= \numAgents\cdot \expect[\quant]{\allocOPT(\quant) \valFunc(\quant)} - \buyerutil + (\dualV - 1) (\buyerbenchmarkDist - \buyerutil). 
    \end{align*}

    The last equation is because of \Cref{prop:proofs in Pareto frontier:anti-MHR SPA gives umax}. Now since $\buyerutil \leq \buyerbenchmarkDist$, the final formula increases with $\dualV$ when $\dualV \geq 1$. Therefore, the dual optimum is reached when $\dualV \leq 1$. 
\end{proof}

Further, again, by \Cref{prop:proofs in Pareto frontier:surplus expression}, buyers' surplus under allocation $\alloc(\quant)$ is presented as:
\begin{align*}
    \buyerutil = \numAgents\cdot \inParentheses{\expect[\quant]{\accAlloc(\quant)\cdot (\quant \valFuncPrime(\quant))'} + \accAlloc(1)\cdot \valFuncPrime(1)}.
\end{align*}

We are now ready to prove our theorems. 
Observe that when $\dualV \leq 1$, the evolving process $\{\accAllocOPT[\dualV, \evol]\}$ satisfies that $\accAllocOPT[\dualV, \evol](\quant)$ is continuous and increasing for any $\quant$. 
Meanwhile, buyers' surplus is $\buyerexanteutil[\dist]{\Mye}$ under $\accAllocOPT[0, 0]$ and $\buyerbenchmarkDist$ under $\accAllocOPT[1, 1]$ according to \Cref{prop:proofs in Pareto frontier:anti-MHR SPA gives umax}, which gives a second-price auction. 
We therefore obtain the following result.
\begin{proposition} \label{prop:proofs in Pareto frontier:MHR surplus for evovling process}
    For anti-MHR value distributions, under the process $\{\accAllocOPT[\dualV, \evol]\}_{0 \leq \dualV \leq 1}$, buyers' surplus continuously and monotonically increases from $\buyerexanteutil[\dist]{\Mye}$ to $\buyerbenchmarkDist$. 
\end{proposition}

Therefore, for any $\buyerexanteutil[\dist]{\Mye} \leq \buyerutil \leq \buyerbenchmarkDist$, there is a pair $(\dualVOPT, \evolOPT)$ such that buyers' surplus under $\accAllocOPT[\dualVOPT, \evolOPT]$ is $\buyerutil$. 
Now, for anti-MHR distributions that are regular, this is a randomization of at most two SPA mechanisms with different reserves that are no larger than the lowest monopoly reserve. This further indicates that there is a SPA mechanism with a reserve (also no larger than the lowest monopoly reserve) that achieves buyers' surplus of $\buyerutil$. 
Suppose the reserve is $\valFunc(\quantDag)$. 
Now, according to \Cref{prop:proofs in Pareto frontier:dualV and quantItvOPT relation}, there is a $\dualVDag$ such that $\quantDag \in \quantItvOPT(\dualVDag)$, or that $\quantDag \in \argmax_{\quant} \lagCurveCH(\quant)$. 
However, for regular distributions, $\lagCurveCH = \lagCurve$, implying that the SPA mechanism with reserve $\valFunc(\quantDag)$ is a primal optimal solution for dual variable $\dualVDag$. 
This proves \Cref{thm:Pareto frontier:multi buyer:regular and antiMHR distribution}. 

For general anti-MHR distributions, one observes that each \emph{extreme point} mechanism of $\{\accAllocOPT[\dualV, \evol]\}_{0 \leq \dualV \leq 1}$ satisfies $\evol = 0, 1$. Meanwhile, we observe that these mechanisms are {\indirectSPAs}, and as $\dualV$ increases, their bidding space refines and their reserve decreases. Specifically, when $\dualV = 0$, the reserve is the lowest monopoly reserve. This proves \Cref{thm:Pareto frontier:multi buyer:antiMHR distribution}. 

\subsection{MHR Distributions}

We now prove the results for MHR distributions. 

\PFMHR*

Recall \Cref{prop:proofs in Pareto frontier:surplus expression} that for each allocation rule $\alloc(\quant)$, each buyer's surplus can be expressed as follows:
\begin{align*}
    &\mathrel{\phantom{=}} \expect[\quant]{\alloc(\quant)\cdot (\valFunc(\quant) - \revCurvePrime(\quant))} \\
    &= \expect[\quant]{\alloc(\quant)\cdot (-\quant \valFuncPrime(\quant))} \\
    &= \expect[\quant]{\accAlloc(\quant)\cdot (\quant \valFuncPrime(\quant))'} - \accAlloc(1)\cdot \valFuncPrime(1) \\
    &= \expect[\quant]{(\accAlloc(1) - \accAlloc(\quant))\cdot (-\quant \valFuncPrime(\quant))'} + \accAlloc(1)\cdot \inBrackets{\expect[\quant]{(\quant \valFuncPrime(\quant))'} - \valFuncPrime(1)}. 
\end{align*}

We now consider the evolving process $\{\accAllocOPT[\dualV, \evol]\}$ introduced in \Cref{sec:proofs in Pareto frontier:evolving process}. When $\dualV = 0$, the surplus given by $\accAllocOPT[0, 0]$ is $\buyerexanteutil[\dist]{\Mye}$. 
When $(\dualV, \evol)$ increases from $(0, 0)$ to $(1, 0)$, notice that the mixed welfare-revenue curve $\lagCurve$ is always concave since the value distribution is MHR. So the canonical primal optimal allocation is a second-price auction with a decreasing reserve, therefore $\allocOPT[\dualV, \evol](\quant)$ is also increasing for any $\quant$, so is buyers' surplus. 
At $(\dualV, \evol) = (1, 0)$, the mechanism is a second-price auction.
When $(\dualV, \evol)$ increases from $(1, 0)$, since $\rightS(\dualV) = \sup \argmax_{\quant} \lagCurve(\quant) = 1$, $\accAllocOPT[\dualV, \evol](1)$ is fixed at $\nicefrac{1}{n}$. Meanwhile, by \Cref{prop:proofs in Pareto frontier:properties for the evolving process}, $\accAllocOPT[\dualV, \evol](\quant)$ is decreasing for each quantile $\quant$, and buyers' surplus is also increasing because $(\quant \valFuncPrime(\quant))' \leq 0$.

Meanwhile, when $\dualV \to +\infty$, the primal optimal allocation maximizes
\[\expect[\quant]{\alloc(\quant) \cdot \lagCurvePrime(\quant)} = \expect[\quant]{\alloc(\quant) \cdot \inBrackets{\revCurvePrime(\quant) + \dualV\cdot \inParentheses{\valFunc(\quant) - \revCurvePrime(\quant)}}}, \]
which certainly maximizes $\expect[\quant]{\alloc(\quant) \cdot \inParentheses{\valFunc(\quant) - \revCurvePrime(\quant)}}$, i.e., each buyer's surplus, and therefore induces buyers' surplus of $\buyerbenchmarkDist$. 
In all, we have the following result. 

\begin{proposition} \label{prop:proofs in Pareto frontier:anti-MHR surplus for evovling process}
    For MHR value distributions, under the process $\{\accAllocOPT[\dualV, \evol]\}$, buyers' surplus continuously and monotonically increases from $\buyerexanteutil[\dist]{\Mye}$ to $\buyerbenchmarkDist$. 
\end{proposition}

As a consequence of \Cref{prop:proofs in Pareto frontier:anti-MHR surplus for evovling process}, for any $\buyerexanteutil[\dist]{\Mye} \leq \buyerutil \leq \buyerbenchmarkDist$, there is a pair $(\dualVOPT, \evolOPT)$ such that buyers' surplus under allocation $\accAllocOPT[\dualVOPT, \evolOPT]$ is exactly $\buyerutil$. 
In specific, for $\dualV \geq 1$, each extreme point mechanism in $\{\accAllocOPT[\dualV, \evol]\}_{\dualV \geq 1}$ satisfies $\evol = 0, 1$, and is an {\indirectSPA} with zero reserve. Meanwhile, the bidding space anti-refines/shrinks as $\dualV$ increases. 
For $\dualV \leq 1$, $\accAllocOPT[\dualVOPT, \evolOPT]$ is given by a randomization of at most two SPA mechanisms with different reserves. 
Similar to the argument in \Cref{thm:Pareto frontier:multi buyer:regular and antiMHR distribution}, it has the same buyers' surplus as an SPA mechanism with a reserve no larger than the lowest monopoly reserve. 
When $\dualV = 1$, $\accAllocOPT[1, 0]$ is the second-price auction.
In all, we finish the proof of \Cref{thm:Pareto frontier:multi buyer:MHR distribution}.

\section{Welfare Analysis of the {\KSsolution}}
\label{apx:ks solution analysis}
\subsection{Proof of \texorpdfstring{\Cref{thm:KS solution welfare approximation}}{Theorem 4.1}}

\thmKSsolution*

Our analysis relies on the following lemmas regarding the revenue and buyers' surplus guarantee of the second-price auction.

\begin{lemma}[\citealp{BK-96}]
\label{lem:SPA:revenue approx regular}
    For any market size $n \geq 1$ and any regular valuation distribution $\buyerdist$, the revenue of the second-price auction (\SPA) is at least $\frac{n-1}{n}$ of the optimal revenue, i.e., $\sellerexanteutil[\buyerdist]{\SPA} \geq \frac{n-1}{n}\cdot \sellerbenchmarkDist$.
\end{lemma}

\begin{lemma}[\citealp{HR-08}]
\label{lem:SPA:buyer surplus antiMHR}
    For any market size $n \geq 1$ and any anti-MHR valuation distribution $\buyerdist$, the buyers' surplus of the second-price auction (\SPA) is optimal, i.e., $\buyerexanteutil[\buyerdist]{\SPA} = \buyerbenchmarkDist$.
\end{lemma}

\begin{lemma}[{\citealp{mye-81,HR-08}}]
\label{prop:revenue equivalence}
\label{prop:buyer surplus equivalence}
Fix a buyer with a value drawn from a valuation distribution $\buyerdist$. 
For every mechanism $\mech \in\mechfam$, 
the buyer's ex ante payment and utility satisfy 
$\expect[\val\sim\buyerdist]{\price(\val)} = \expect[\val\sim\buyerdist]{\virtualval(\val)\cdot \alloc(\val)}$
and
$\buyerexanteutil[\dist]{\mech} = \expect[\val\sim\buyerdist]{(1 / \buyerhazardrate(\val))\cdot \alloc(\val)}$, 
where $\alloc(\val)$ is the interim allocation of the buyer with value $\val$, $\virtualval(\val)$ is her virtual value, and $\buyerhazardrate(\val)$ is her hazard rate. 
\end{lemma}

\begin{proof}[Proof of \Cref{thm:KS solution welfare approximation}]
We now prove all the statements sequentially.

\xhdr{General distributions.} Given any valuation distribution $\buyerdist$, we show that the social welfare of the {\KSsolution} is at least $\frac{1}{2}$ of the optimal social welfare. Let $\mech\primed$ denote the Myerson auction that maximizes the seller revenue, and $\mech\doubleprimed$ denote the buyer-side optimal mechanism that maximizes the buyers' surplus. Define auxiliary notations $\alpha,\beta,\mixprob$ such as 
\begin{align*}
    \alpha \triangleq \frac{\buyerexanteutil[\buyerdist]{\mech\primed}}{\buyerbenchmarkDist}
    \;\;
    \mbox{and}
    \;\;
    \beta \triangleq \frac{\sellerexanteutil[\buyerdist]{\mech\doubleprimed}}{\sellerbenchmarkDist}
    \;\;
    \mbox{and}
    \;\;
    \mixprob \triangleq \frac{1-\beta}{2 - \alpha - \beta}. 
\end{align*}
By definition, $\alpha,\beta,\mixprob\in[0, 1]$. Consider mechanism $\mech$ that implements Myerson auction $\mech\primed$ and buyer-side optimal $\mech\doubleprimed$ with ex ante probabilities $\mixprob$ and $1-\mixprob$, respectively. We note that mechanism $\mech$ satisfies the {\ksfairness}. To see this, note that 
\begin{align*}
    \frac{\sellerexanteutil[\buyerdist]{\mech}}{\sellerbenchmarkDist}
    &=
    \mixprob\cdot 
    \frac{\sellerexanteutil[\buyerdist]{\mech\primed}}{\sellerbenchmarkDist}
    +
    (1 - \mixprob)\cdot 
    \frac{\sellerexanteutil[\buyerdist]{\mech\doubleprimed}}{\sellerbenchmarkDist}
    \overset{(a)}{=}
    \mixprob\cdot 
    1
    +
    (1 - \mixprob)\cdot 
    \beta
    =
    \frac{1-\alpha\beta}{2 - \alpha-\beta};
    \\
    \frac{\buyerexanteutil[\buyerdist]{\mech}}{\buyerbenchmarkDist}
    &=
    \mixprob\cdot 
    \frac{\buyerexanteutil[\buyerdist]{\mech\primed}}{\buyerbenchmarkDist}
    +
    (1 - \mixprob)\cdot 
    \frac{\buyerexanteutil[\buyerdist]{\mech\doubleprimed}}{\buyerbenchmarkDist}
    \overset{(b)}{=}
    \mixprob\cdot 
    \alpha
    +
    (1 - \mixprob)\cdot 
    1
    =
    \frac{1-\alpha\beta}{2 - \alpha-\beta},
\end{align*}
where equality~(a) holds since Myerson auction $\mech\primed$ maximizes the revenue and equality~(b) holds since the mechanism $\mech\doubleprimed$ maximizes the buyers' surplus. We next analyze the social welfare approximation of the KS-fair mechanism $\mech$, which lower bounds the social welfare approximation of the {\KSsolution}. Note that 
\begin{align*}
    \frac{\GFT[\buyerdist]{\mech}}{\OPT_\buyerdist}
    \overset{(a)}{\geq} 
    \frac{\buyerexanteutil[\buyerdist]{\mech} + \sellerexanteutil[\buyerdist]{\mech}}{\buyerbenchmarkDist+\sellerbenchmarkDist}
    =
    \frac{1-\alpha\beta}{2 - \alpha-\beta}
    \overset{(b)}{\geq}
    \frac{1}{2},
\end{align*}
where inequality~(a) holds since $\OPT_\buyerdist\leq \buyerbenchmarkDist+\sellerbenchmarkDist$, and inequality~(b) holds since $\alpha\in[0, 1]$ and $\beta \in[0, 1]$ and term $(1 - \alpha\beta)/(2-\alpha-\beta)$ is minimized at $\alpha = 0$ and $\beta = 0$.

\xhdr{Regular and anti-MHR distributions.} The analysis is similar to the one used for general distributions. Suppose the valuation distribution $\buyerdist$ is both regular and anti-MHR. We show that the social welfare of the {\KSsolution} is at least $\frac{n}{n + 1}$ of the optimal social welfare. Let $\mech\primed$ denote the Myerson auction that maximizes the seller revenue, and $\mech\doubleprimed$ denote the second-price auction. Define auxiliary notations $\alpha,\beta,\mixprob$ such as 
\begin{align*}
    \alpha \triangleq \frac{\buyerexanteutil[\buyerdist]{\mech\primed}}{\buyerbenchmarkDist}
    \;\;
    \mbox{and}
    \;\;
    \beta \triangleq \frac{\sellerexanteutil[\buyerdist]{\mech\doubleprimed}}{\sellerbenchmarkDist}
    \;\;
    \mbox{and}
    \;\;
    \mixprob \triangleq \frac{1-\beta}{2 - \alpha - \beta}.
\end{align*}
By definition, $\alpha,\beta,\mixprob\in[0, 1]$. Consider mechanism $\mech$ that implements Myerson auction $\mech\primed$ and second-price auction $\mech\doubleprimed$ with ex ante probabilities $\mixprob$ and $1-\mixprob$, respectively. We note that mechanism $\mech$ satisfies the {\ksfairness}. To see this, note that 
\begin{align*}
    \frac{\sellerexanteutil[\buyerdist]{\mech}}{\sellerbenchmarkDist}
    &=
    \mixprob\cdot 
    \frac{\sellerexanteutil[\buyerdist]{\mech\primed}}{\sellerbenchmarkDist}
    +
    (1 - \mixprob)\cdot 
    \frac{\sellerexanteutil[\buyerdist]{\mech\doubleprimed}}{\sellerbenchmarkDist}
    \overset{(a)}{=}
    \mixprob\cdot 
    1
    +
    (1 - \mixprob)\cdot 
    \beta
    =
    \frac{1-\alpha\beta}{2 - \alpha-\beta};
    \\
    \frac{\buyerexanteutil[\buyerdist]{\mech}}{\buyerbenchmarkDist}
    &=
    \mixprob\cdot 
    \frac{\buyerexanteutil[\buyerdist]{\mech\primed}}{\buyerbenchmarkDist}
    +
    (1 - \mixprob)\cdot 
    \frac{\buyerexanteutil[\buyerdist]{\mech\doubleprimed}}{\buyerbenchmarkDist}
    \overset{(b)}{=}
    \mixprob\cdot 
    \alpha
    +
    (1 - \mixprob)\cdot 
    1
    =
    \frac{1-\alpha\beta}{2 - \alpha-\beta},
\end{align*}
where equality~(a) holds since Myerson auction $\mech\primed$ maximizes the revenue and equality~(b) holds since the second-price auction $\mech\doubleprimed$ maximizes the buyers' surplus when distribution $\buyerdist$ is anti-MHR (\Cref{lem:SPA:buyer surplus antiMHR}). We next analyze the social welfare approximation of the KS-fair mechanism $\mech$, which lower bounds the social welfare approximation of the {\KSsolution}. Note that 
\begin{align*}
    \frac{\GFT[\buyerdist]{\mech}}{\OPT_\buyerdist}
    \overset{(a)}{\geq} 
    \frac{\buyerexanteutil[\buyerdist]{\mech} + \sellerexanteutil[\buyerdist]{\mech}}{\buyerbenchmarkDist+\sellerbenchmarkDist}
    =
    \frac{1-\alpha\beta}{2 - \alpha-\beta}
    \overset{(b)}{\geq}
    \frac{1}{2-\frac{n-1}{n}}
    =
    \frac{n}{n+1},
\end{align*}
where inequality~(a) holds since $\OPT_\buyerdist\leq \buyerbenchmarkDist+\sellerbenchmarkDist$, and inequality~(b) holds since $\alpha\in[0, 1]$ and $\beta \in[(n-1)/n, 1]$ (invoking \Cref{lem:SPA:revenue approx regular} for regular $\buyerdist$) and term $(1 - \alpha\beta)/(2-\alpha-\beta)$ is minimized at $\alpha = 0$ and $\beta = (n - 1)/n$.

\xhdr{Hard instance with a MHR distribution.}
We analyze the welfare guarantee of the {\KSsolution} in \Cref{example:multi-buyer:mhr hard instance}. The optimal social welfare  can be computed as 
    \begin{align*}
        \OPT_\buyerdist &= \expect{\val_{(1:n)}} = 
        \expect{\val_{(1:n)}|\val_{(1:n)} < \constantH}\cdot \prob{\val_{(1:n)}<\constantH}
        +
        \constantH \cdot \prob{\val_{(1:n)} = \constantH}
        \\
        &=
        \expect{\val_{(1:n)}|\val_{(1:n)} < \constantH} \cdot 
        \left(1-\frac{1}{{n}^{1/4}}\right)^n
        +
        \constantH\cdot 
        \left(1 - 
        \left(1-\frac{1}{{n}^{1/4}}\right)^n\right)
        =
        (1 - o(1))\cdot \constantH.
    \end{align*}
Next we analyze the social welfare of the {\KSsolution}, denoted by $\mech$. Let random variable $\winnerval$ as the value of the winner (i.e., the buyer receiving the item) in the {\KSsolution} if there exists one such winner. If no such winner exists, we set $\winnerval = 0$.
Define auxiliary value threshold $\threshold \triangleq \ln\ln n \leq \constantH$ 
and the following two probabilities:
\begin{align*}
    \allocLow\triangleq \prob{\text{winner exists and $\winnerval \leq \threshold$}}
    \;\;
    \mbox{and}
    \;\;
    \allocHigh\triangleq \prob{\text{winner exists and $\winnerval > \threshold$}}.
\end{align*}
Invoking \Cref{prop:revenue equivalence}, the expected revenue can be expressed as 
\begin{align*}
    \sellerexanteutil[\buyerdist]{\mech} &= \expect{\virtualval(\winnerval)\cdot \indicator{\text{winner exists}}} 
    \\
    &=
    \allocLow\cdot \expect{\virtualval(\winnerval)\condition \text{winner exists and $\winnerval \leq \threshold$}}
    +
    \allocHigh\cdot \expect{\virtualval(\winnerval)\condition \text{winner exists and $\winnerval > \threshold$}}
    \\
    &
    \overset{(a)}{\leq} 
    \allocLow\cdot \virtualval(\threshold)
    +
    \allocHigh\cdot \virtualval(\constantH)
    \overset{(b)}{=}
    \allocLow\cdot \left(\ln\ln n - \frac{1}{2\ln\ln n}\right)
    +
    \allocHigh\cdot \frac{1}{2}\sqrt{\ln n}
    =
    \left(\allocLow\cdot o(1) + \allocHigh\right) \cdot \constantH,
\end{align*}
where inequality~(a) holds since virtual value function $\virtualval(\cdot)$ is increasing, and equality~(b) holds since virtual values $\virtualval(\threshold) = \threshold - \frac{1}{2\threshold} = \ln\ln n - \frac{1}{2\ln\ln n}$ and $\virtualval(\constantH) = \constantH = \frac{1}{2}\sqrt{\ln n}$ by construction.

Similarly, invoking \Cref{prop:revenue equivalence},  buyers' surplus can be expressed as 
    \begin{align*}
    &\buyerexanteutil[\buyerdist]{\mech} 
    \\
    ={}& \expect{1/\buyerhazardrate(\winnerval)} 
    \\
     ={}&
    \allocLow\cdot \expect{1/\buyerhazardrate(\winnerval)\condition \text{winner exists and $\winnerval \leq \threshold$}}
    +
    \allocHigh\cdot \expect{1/\buyerhazardrate(\winnerval)\condition \text{winner exists and $\winnerval > \threshold$}}
    \\
    \overset{(a)}{\leq}{} &  
    \allocLow\cdot \expect[\val\sim\buyerdist]{\buyerhazardrate(\val)\condition\val \leq \threshold}
    +
    \allocHigh\cdot \buyerhazardrate(\threshold)
    \\
    \overset{(b)}{=}{}&
    \allocLow\cdot (1-o(1))\frac{\sqrt{\pi}}{2}
    +
    \allocHigh\cdot \frac{1}{2\ln\ln n}
    \\
    \overset{(c)}{=} {}&
    \left(\allocLow\cdot (1 + o(1)) + \allocHigh\cdot o(1)\right) \cdot \expect[\val\sim \buyerdist]{\val},
\end{align*}
where inequality~(a) holds since buyers' hazard rate function $\buyerhazardrate$ is decreasing, and the buyer's interim allocation rule is weakly increasing, equality~(b) holds since hazard rate function $\buyerhazardrate(\val) = 2\val$, and inequality~(c) holds since $\expect[\val\sim\buyerdist]{\val} = (1 - o(1))\cdot \frac{\sqrt{\pi}}{2}$ by construction.

Putting both the pieces together, the social welfare of the {\KSsolution} $\mech$ can be upper bounded as follows:
\begin{align*}
    \GFT[\buyerdist]{\mech} 
    &\overset{(a)}{=} 
    \sellerexanteutil[\buyerdist]{\mech} + \buyerexanteutil[\buyerdist]{\mech} 
    \\
    &\overset{(b)}{=}
    \min\left\{
    \frac{\sellerexanteutil[\buyerdist]{\mech}}{\sellerbenchmarkDist},
    \frac{\buyerexanteutil[\buyerdist]{\mech}}{\buyerbenchmarkDist}
    \right\}
    \cdot 
    \left(\sellerbenchmarkDist + \buyerbenchmarkDist\right)
    \\
    &\overset{}{\leq}
    \min\left\{
    \frac{\left(\allocLow\cdot o(1) + \allocHigh\right) \cdot \constantH}{(1-o(1))\cdot \constantH},
    \frac{\left(\allocLow\cdot (1 + o(1)) + \allocHigh\cdot o(1)\right) \cdot \expect[\val\sim \buyerdist]{\val}}{\expect[\val\sim \buyerdist]{\val}}
    \right\}
    \cdot 
    (1-o(1))\cdot \constantH
    \\
    &=
    \min\left\{
    (1 + o(1))(\allocLow\cdot o(1) + \allocHigh),
    \left(\allocLow\cdot (1 + o(1)) + \allocHigh\cdot o(1)\right)
    \right\}
    \cdot 
    (1-o(1))\cdot \constantH
    \\
    &\overset{(c)}{\leq}
    \left(\frac{1}{2} + o(1)\right)\cdot \constantH = 
    \left(\frac{1}{2} + o(1)\right)\cdot \OPT_{\buyerdist},
\end{align*}
where equality~(a) holds by the welfare definition, equality~(b) holds due to \ref{eq:KS fairness} in the definition of the {\KSsolution}, inequality~(c) holds due to the seller revenue and buyers' surplus upper bounds derived above, and inequality~(d) holds since $\allocLow+\allocHigh \leq 1$ by construction. This finishes the analysis of \Cref{example:multi-buyer:mhr hard instance} and concludes the proof of \Cref{thm:KS solution welfare approximation}.
\end{proof}

\subsection{Proof of \texorpdfstring{\Cref{thm:KS solution welfare approximation multi-unit}}{Theorem 4.2}}

\thmKSsolutionMultiUnit*

Our analysis relies on the following lemmas.

\begin{lemma}[\citealp{BK-96}]
\label{lem:regular:multi-unit BK}
    Consider a market with $m$ identical units of a good and $n\geq m$ unit-demand buyers whose valuations are drawn i.i.d.\ from a regular distribution $\buyerdist$.
    The seller revenue in the {\VCGAuction} is at least an $\frac{n - m}{n}$-approximation to the Bayesian optimal revenue $\sellerbenchmark$. 
\end{lemma}

\begin{lemma}[\citealp{HR-08}]
\label{lem:mhr:multi-unit buyer offer structure}
    Consider a market with $m$ identical units of a good and $n\geq m$ unit-demand buyers whose valuations are drawn i.i.d.\ from a MHR distribution $\buyerdist$. The {\BuyerOffer} (that maximizes the buyers' surplus) allocates all $m$ units to $m$ buyers uniformly at random with zero payment.
\end{lemma}

\begin{lemma}
\label{lem:mhr:multi-unit buyer offer}
    Consider a market with $m$ identical units of a good and $n\geq m$ unit-demand buyers whose valuations are drawn i.i.d.\ from a MHR distribution $\buyerdist$. The social welfare of the {\BuyerOffer} that maximizes the buyers' surplus is at least an $\frac{1}{1 + H_n - H_m}$-approximation to the optimal social welfare $\OPT_\buyerdist$. 
\end{lemma}

\begin{proof}
Invoking \Cref{lem:mhr:multi-unit buyer offer structure}, the {\BuyerOffer} ({\BOM}) allocates all $m$ units uniformly at random to $m$ buyers at zero prices. Its social welfare is therefore
\begin{align*}
\GFT[\buyerdist]{\BOM} = m \cdot \expect[\buyerdist]{X},
\end{align*}
where $X \sim \buyerdist$.  
The optimal social welfare equals
\begin{align*}
\OPT_\buyerdist = \sum\nolimits_{i=1}^{m} \expect[\buyerdist]{X_{(i:n)}},
\end{align*}
where $X_{(i:n)}$ denotes the $i$-th largest order statistic from $n$ i.i.d.\ samples.

Thus it suffices to show that for every MHR distribution $\buyerdist$ and every $1 \le k \le n$,
\begin{equation}
\label{eq:order-stat-bound}
k \cdot \expect[\buyerdist]{X}
\;\ge\;
\frac{1}{1 + H_n - H_k}
\sum\nolimits_{i\in[k]} \expect[\buyerdist]{X_{(i:n)}} .
\end{equation}
Equivalently,
\begin{align*}
\frac{\sum\nolimits_{i\in[k]} \expect[\buyerdist]{X_{(i:n)}}}
     {k \cdot \expect[\buyerdist]{X}}
\;\le\;
1 + H_n - H_k .
\end{align*}
We prove~\eqref{eq:order-stat-bound} in two steps.

\xhdr{Step 1: Reduction to the exponential distribution.}
Let $\prefixratioTkn \triangleq 1 + H_n - H_k$.  
Recall the (quantile-)value function $\valFunc:[0,1]\to\reals_+$ by
\begin{align*}
\valFunc(\quant) = F^{-1}(1-q),
\end{align*}
and let $g_{i,n}(\quant)$ denote the density of the $i$-th largest order statistic of $n$ i.i.d.\ $U[0,1]$ samples:
\begin{align*}
g_{i,n}(\quant)
= \frac{n!}{(i-1)!(n-i)!} (1-q)^{i-1} q^{n-i}.
\end{align*}
Using the quantile representation of expectations,
\begin{align*}
\expect[\buyerdist]{X_{(i:n)}} = \int_0^1 g_{i,n}(\quant)\, \valFunc(\quant)\, \dd q,
\qquad
\expect[\buyerdist]{X} = \int_0^1 \valFunc(\quant)\,\dd \quant.
\end{align*}
Hence,
\begin{align*}
&\sum\nolimits_{i\in[k]} \expect[\buyerdist]{X_{(i:n)}}
- \prefixratioTkn \cdot k \cdot \expect[\buyerdist]{X}  =
\int_0^1 w(\quant)\, \valFunc(\quant)\,\dd \quant,
\end{align*}
where
\begin{align*}
w(\quant) \triangleq \sum\nolimits_{i\in[k]} g_{i,n}(\quant) - \prefixratioTkn \cdot k.
\end{align*}

\begin{claim} The function $w(\quant)$ is weakly decreasing on $[0,1]$ and has at most one zero.
\end{claim}
\begin{proof}
Each $g_{i,n}(\quant)$ is a unimodal polynomial density whose derivative satisfies
\begin{align*}
\frac{\dd}{\dd \quant} g_{i,n}(\quant)
= g_{i,n}(\quant)\left(\frac{n-i}{\quant} - \frac{i-1}{1-\quant}\right),
\end{align*}
which is negative whenever $\quant$ exceeds its mode $\frac{n-i}{n-1}$. Since $i\le k$, all modes lie weakly to the right of $\frac{n-k}{n-1}$. Therefore the sum $\sum\nolimits_{i\in[k]} g_{i,n}(\quant)$ is decreasing on $[0,1]$, and subtracting a constant preserves monotonicity. A monotone function can cross zero at most once. 
\end{proof}

Let $\quant^*$ denote the (possibly nonexistent) zero of $w(\quant)$. Define $\val^* = \valFunc(\quant^*)$.  
Consider the exponential distribution $\auxdist'$ with CDF
\begin{align*}
F_{\auxdist'}(v)
=
1 - \exp\!\left(\ln(\quant^*)\, \frac{v}{\val^*}\right),
\end{align*}
whose quantile-value function satisfies
\begin{align*}
\val_{\auxdist'}(\quant) = \val^* \frac{\ln \quant}{\ln \quant^*}.
\end{align*}
Because $\buyerdist$ is MHR, its quantile-value function $\valFunc(\quant)$ is convex in $\ln \quant$. Hence $\valFunc$ lies below the chord through $(\quant^*,\val^*)$, implying
\begin{align*}
\val_{\auxdist'}(\quant) \ge \valFunc(\quant) \quad \text{for } q \le \quant^*,
\qquad
\val_{\auxdist'}(\quant) \le \valFunc(\quant) \quad \text{for } q \ge \quant^*.
\end{align*}
Since $w(\quant)$ is positive for $q<\quant^*$ and negative for $q>\quant^*$,
\begin{align*}
\int_0^1 w(\quant)\, \valFunc(\quant)\,\dd \quant
\;\le\;
\int_0^1 w(\quant)\, \val_{\auxdist'}(\quant)\,\dd \quant.
\end{align*}
Finally, expectations under exponential distributions scale linearly with the mean. Therefore,
\begin{align*}
\int_0^1 w(\quant)\, \val_{\auxdist'}(\quant)\,\dd \quant
&=
\left(
\sum\nolimits_{i\in[k]} \expect[\auxdist]{X_{(i:n)}}
- \prefixratioTkn \cdot k \cdot \expect[\auxdist]{X}
\right)
\cdot
\frac{\expect[\auxdist']{X}}{\expect[\auxdist]{X}} 
= 0,
\end{align*}
where $\auxdist$ denotes the unit-mean exponential distribution and the equality follows by the definition of $\prefixratioTkn$.
This completes the reduction.

\xhdr{Step 2: Evaluation under the exponential distribution.}
Let $X \sim \mathrm{Exp}(1)$. The expected $i$-th largest order statistic satisfies
\begin{align*}
\expect{X_{(i:n)}} = H_n - H_{i-1}.
\end{align*}
Therefore,
\begin{align*}
\sum\nolimits_{i\in[k]} \expect{X_{(i:n)}}
=
k(H_n - H_k) + k,
\end{align*}
and since $\expect{X}=1$,
\begin{align*}
\frac{\sum\nolimits_{i\in[k]} \expect{X_{(i:n)}}}{k}
=
1 + H_n - H_k
=
\prefixratioTkn.
\end{align*}
Combining the two steps proves~\eqref{eq:order-stat-bound} and completes the proof.
\end{proof}

\begin{lemma}
\label{lem:mhr:multi-unit seller offer}
Consider a market with $m$ identical units of a good and $n\geq m$ unit-demand buyers whose valuations are drawn i.i.d.\ from a MHR distribution $\buyerdist$. The social welfare of the Myerson auction that maximizes the Bayesian revenue is at least an $\frac{1}{e-1}$-approximation to the optimal social welfare $\OPT_\buyerdist$.
\end{lemma}
\begin{proof}
Consider $n\ge m$ unit-demand buyers with i.i.d.\ values drawn from an MHR distribution $\buyerdist$.
Let $\OPT_\buyerdist$ denote the optimal social welfare, i.e.,
\begin{align*}
\OPT_\buyerdist = \sum\nolimits_{i\in[m]} \expect[\buyerdist]{\val_{(i:n)}} .
\end{align*}
where $\val_{(i:n)}$ is the $i$-th largest order statistics of $n$ samples from distribution $\buyerdist$.
The Bayesian revenue-optimal mechanism is the Myerson auction, which for MHR distribution is the
second-price auction with (unique) monopoly reserve $\optreserve$.
We denote this mechanism by $\Mye$.
We prove that
\begin{align*}
\GFT[\buyerdist]{\Mye}
\;\ge\;
\frac{1}{e-1}\,\OPT_\buyerdist .
\end{align*}
The proof proceeds in two steps.

\xhdr{Step 1: Reduction to a truncated exponential distribution.}
Let $\optreserve$ be the monopoly reserve of MHR distribution $\buyerdist$.
Define the (quantile-)value function $\valFunc(\quant)=\buyercdf^{-1}(1-\quant)$.
For MHR distributions, $\valFunc(\quant)$ is convex in $\ln \quant$.
Moreover, the Myerson allocation rule is monotone and allocates units only to buyers with
values at least $\optreserve$.

Similar to the proof of \Cref{lem:mhr:multi-unit buyer offer}, both the welfare of $\Mye$
and the optimal welfare can be written as linear functionals of $\valFunc$:
\begin{align*}
\GFT[\buyerdist]{\Mye}
= \int_0^1 a(\quant)\,\valFunc(\quant)\,\dd\quant,
\qquad
\OPT_\buyerdist
= \int_0^1 b(\quant)\,\valFunc(\quant)\,\dd\quant,
\end{align*}
where $a(\quant)$ and $b(\quant)$ are deterministic weight functions induced by the allocation rules.
The function $b(\quant)$ corresponds to allocating to the top $m$ order statistics, while
$a(\quant)$ coincides with $b(\quant)$ for $\quant\le \quantM$ and vanishes for $\quant>\quantM$, where
$\quantM = 1-\buyerdist(\optreserve)$ is the monopoly quantile.
Note that the difference
\begin{align*}
\int_0^1 \bigl( (e-1)a(\quant) - b(\quant) \bigr)\valFunc(\quant)\,\dd\quant
\end{align*}
has a single-crossing structure: the coefficient $(e-1)a(\quant)-b(\quant)$ is non-negative for
$\quant\le \quantM$ and non-positive for $\quant>\quantM$.
By the same extremal argument as in \Cref{lem:mhr:multi-unit buyer offer}, among all MHR distributions with monopoly reserve $\optreserve$,
the ratio $\GFT[\buyerdist]{\Mye}/\OPT_\buyerdist$ is minimized by the exponential distribution
$\buyerdist(\val)=1-e^{-\val}$, truncated below at $\optreserve=1$.
Hence, it suffices to prove the lemma for the truncated exponential distribution
\begin{align*}
\buyerdist(\val)=
\begin{cases}
1-e^{-\val}, & 0\le \val\le 1,\\
1, & \val>1.
\end{cases}
\end{align*}

\xhdr{Step 2: Reduction to the case $n=m$ under the truncated exponential.}
Under the worst-case distribution $\buyerdist$ (i.e., truncated exponential distribution with truncation at monopoly reserve $\optreserve= 1$) established in Step 1, each buyer's value equals $1$ with probability $e^{-1}$ and lies in $[0,1)$ otherwise. By construction, the welfare approximation ratio can be expressed as 
\begin{align*}
    \frac{\GFT[\buyerdist]{\Mye}}{\OPT_\buyerdist}
    &=
    \frac{
    \sum\nolimits_{i\in[m]}\expect[\buyerdist]{\val_{(i:n)}\cdot
    \indicator{\val_{(i:n)} \geq \optreserve}}}{
    \sum\nolimits_{i\in[m]}\expect[\buyerdist]{\val_{(i:n)}
    }}
    \\
    &\geq 
    \frac{
    \sum\nolimits_{i\in[m]}\expect[\buyerdist]{\val_{(i:m)}\cdot
    \indicator{\val_{(i:m)} \geq \optreserve}}}{
    \sum\nolimits_{i\in[m]}\expect[\buyerdist]{\val_{(i:m)}
    }}
    =
    \frac{\expect[\buyerdist]{\val\cdot
    \indicator{\val \geq \optreserve}}}{
   \expect[\buyerdist]{\val}
   }
   =
   \frac{1/e}{1-1/e} = \frac{1}{e-1}
\end{align*}
This establishes the desired bound and completes the proof.
\end{proof}

\begin{proof}[Proof of \Cref{thm:KS solution welfare approximation multi-unit}]
    Let $\VCG$ denote the {\VCGAuction}. Our analysis considers two cases, based on the relative performance of $\mech$ with respect to the seller's ideal benchmark $\sellerbenchmarkDist$ and buyers' ideal benchmark $\buyerbenchmarkDist$.

    \xhdr{Case 1: $\displaystyle {\sellerexanteutil[\buyerdist]{\VCG}}/{\sellerbenchmarkDist} \geq {\buyerexanteutil[\buyerdist]{\VCG}}/{\buyerbenchmarkDist}$.}
    In this case, 
    consider an ex ante randomization between {\VCGAuction} (\VCG) and the {\BuyerOffer} (\BOM) that maximizes the buyers' surplus, and let $\mixprob$ denote the probability of implementing the {\VCGAuction}.

    Since the valuation distribution is MHR, {\BuyerOffer} allocates the $m$ units uniformly at random to buyers for free, yielding zero seller revenue. To satisfy \ref{eq:KS fairness} (by case assumption and mixture construction, this can be done), the {\VCGAuction} must be used with probability at least $1/2$. By \Cref{lem:mhr:multi-unit buyer offer}, the social welfare of this mixture is lower bounded by
    \begin{align*}
        &\mixprob \cdot \GFT[\buyerdist]{\VCG} + (1 - \mixprob) \cdot  \GFT[\buyerdist]{\BOM}
        \\
        \overset{(a)}{\geq} {} &\mixprob \cdot \OPT_\buyerdist + (1 - \mixprob) \cdot \frac{1}{1 + H_n - H_m}\cdot \OPT_\buyerdist
        \\
        \overset{(b)}{\geq}{}& \left(\frac{1}{2} \cdot 1 + \frac{1}{2} \cdot \frac{1}{1 + H_n - H_m}\right)\cdot \OPT_\buyerdist
        \\
        \overset{(c)}{\geq}{} &  \frac{2 - \ln\!\left(\frac{m}{n}\right)}{2 - 2\ln\!\left(\frac{m}{n}\right)}\cdot 
        \OPT_\buyerdist,
    \end{align*}
    where inequality~(a) holds due to \Cref{lem:mhr:multi-unit buyer offer} and the social welfare optimality of the {\VCGAuction}, inequality~(b) holds due to $\mixprob \geq 1/2$ argued above, and inequality~(c) holds due to the bound $H_n - H_m \leq \ln(n/m)$.

    \xhdr{Case 2: $\displaystyle {\sellerexanteutil[\buyerdist]{\VCG}}/{\sellerbenchmarkDist} < {\buyerexanteutil[\buyerdist]{\VCG}}/{\buyerbenchmarkDist}$.}
    In this case, consider an ex ante randomization between the {\VCGAuction} ({\VCG}) and the Myerson auction ({\Mye}) that maximizes the seller revenue, with $\mixprob$ again denoting the probability of using $\mech$. By \Cref{lem:regular:multi-unit BK}, we have
    \begin{align*}
        \frac{\sellerexanteutil[\buyerdist]{\mech}}{\sellerbenchmark} \geq \frac{n - m}{n}.
    \end{align*}
    Moreover, since the distribution is MHR, the Myerson auction achieves at least a $1/e$ fraction of the optimal social welfare \citep{HR-09}, implying that its buyers' surplus satisfies
    \begin{align*}
        \frac{\buyerexanteutil[\buyerdist]{\Mye}}{\buyerbenchmarkDist} \leq 1 - \frac{1}{e}.
    \end{align*}
    To satisfy \ref{eq:KS fairness} (by case assumption and mixture construction, this can be done), the mixture must equalize the normalized seller revenue and buyers' surplus:
    \begin{align*}
        \mixprob \cdot \frac{\sellerexanteutil[\buyerdist]{\VCG}}{\sellerbenchmarkDist} + (1 - \mixprob) \cdot \frac{\sellerexanteutil[\buyerdist]{\Mye}}{\sellerbenchmarkDist}
        \;=\;
        \mixprob \cdot \frac{\buyerexanteutil[\buyerdist]{\VCG}}{\buyerbenchmarkDist} + (1 - \mixprob) \cdot \frac{\buyerexanteutil[\buyerdist]{\Mye}}{\buyerbenchmarkDist}.
    \end{align*}
    Using the bounds above and simplifying yields
    \begin{align*}
        \mixprob \geq \frac{1}{1 + e \cdot \frac{m}{n}}.
    \end{align*}
    Thus, the social welfare of the mixture is then lower bounded by
    \begin{align*}
        &\mixprob \cdot \GFT[\buyerdist]{\VCG} + (1 - \mixprob) \cdot  \GFT[\buyerdist]{\Mye}
        \\
        \overset{(a)}{\geq} {} &
        \mixprob \cdot \OPT_\buyerdist + (1 - \mixprob) \cdot \frac{1}{e - 1} \cdot  \OPT_\buyerdist
        \\
        \overset{(b)}{\geq} {} &
        \left(\frac{e - 1}{e} + \frac{1}{e + e^2 \cdot \frac{m}{n}}\right)\cdot \OPT_\buyerdist
    \end{align*}
    where inequality~(a) holds due to \Cref{lem:mhr:multi-unit seller offer} and the social welfare optimality of the {\VCGAuction}, and inequality~(b) holds due to the lower bound of mixing probability $\mixprob$ derived above.

    Combining both cases, the welfare of the {\KSsolution} is at least $\multiUnitRatio(m/n)$ times the optimum, completing the proof.
\end{proof} 

\section{Welfare Analysis of the {\Nashsolution} and {\CSNashsolution}}
\label{apx:nash solution analysis}

\subsection{Proof of \texorpdfstring{\Cref{thm:Nash solution welfare approximation mhr}}{Theorem 5.1}}
\thmNashsolutionMHR*
\begin{proof} 
We analyze the welfare guarantee of the {\Nashsolution} in \Cref{example:multi-buyer:mhr hard instance}. The optimal social welfare  can be computed as 
\begin{align*}
    \OPT_\buyerdist &= \expect{\val_{(1:n)}} = 
    \expect{\val_{(1:n)}|\val_{(1:n)} < \constantH}\cdot \prob{\val_{(1:n)}<\constantH}
    +
    \constantH \cdot \prob{\val_{(1:n)} = \constantH}
    \\
    &=
    \expect{\val_{(1:n)}|\val_{(1:n)} < \constantH} \cdot 
    \left(1-\frac{1}{{n}^{1/4}}\right)^n
    +
    \constantH\cdot 
    \left(1 - 
    \left(1-\frac{1}{{n}^{1/4}}\right)^n\right)
    =
    (1 - o(1))\cdot \constantH.
\end{align*}
Next we analyze the social welfare of the {\Nashsolution}, denoted by $\mech$. Let random variable $\winnerval$ as the value of the winner (i.e., the buyer receiving the item) in the {\Nashsolution} if there exists one such winner. If no such winner exists, we set $\winnerval = 0$.
Define auxiliary value threshold $\threshold \triangleq \ln\ln n \leq \constantH$ 
and the following two probabilities:
\begin{align*}
    \allocLow\triangleq \prob{\text{winner exists and $\winnerval \leq \threshold$}}
    \;\;
    \mbox{and}
    \;\;
    \allocHigh\triangleq \prob{\text{winner exists and $\winnerval > \threshold$}}.
\end{align*}
Invoking \Cref{prop:revenue equivalence}, the expected revenue can be expressed as 
\begin{align*}
    \sellerexanteutil[\buyerdist]{\mech} &= \expect{\virtualval(\winnerval)\cdot \indicator{\text{winner exists}}} 
    \\
    &=
    \allocLow\cdot \expect{\virtualval(\winnerval)\condition \text{winner exists and $\winnerval \leq \threshold$}}
    +
    \allocHigh\cdot \expect{\virtualval(\winnerval)\condition \text{winner exists and $\winnerval > \threshold$}}
    \\
    &
    \overset{(a)}{\leq} 
    \allocLow\cdot \virtualval(\threshold)
    +
    \allocHigh\cdot \virtualval(\constantH)
    \overset{(b)}{=}
    \allocLow\cdot \left(\ln\ln n - \frac{1}{2\ln\ln n}\right)
    +
    \allocHigh\cdot \frac{1}{2}\sqrt{\ln n}
    =
    \left(\allocLow\cdot o(1) + \allocHigh\right) \cdot \constantH,
\end{align*}
where inequality~(a) holds since virtual value function $\virtualval(\cdot)$ is increasing, and equality~(b) holds since virtual values $\virtualval(\threshold) = \threshold - \frac{1}{2\threshold} = \ln\ln n - \frac{1}{2\ln\ln n}$ and $\virtualval(\constantH) = \constantH = \frac{1}{2}\sqrt{\ln n}$ by construction.

Similarly, invoking \Cref{prop:revenue equivalence},  buyers' surplus can be expressed as 
    \begin{align*}
    &\buyerexanteutil[\buyerdist]{\mech} 
    \\
    ={}& \expect{1/\buyerhazardrate(\winnerval)} 
    \\
     ={}&
    \allocLow\cdot \expect{1/\buyerhazardrate(\winnerval)\condition \text{winner exists and $\winnerval \leq \threshold$}}
    +
    \allocHigh\cdot \expect{1/\buyerhazardrate(\winnerval)\condition \text{winner exists and $\winnerval > \threshold$}}
    \\
    \overset{(a)}{\leq}{} &  
    \allocLow\cdot \expect[\val\sim\buyerdist]{\buyerhazardrate(\val)\condition\val \leq \threshold}
    +
    \allocHigh\cdot \buyerhazardrate(\threshold)
    \\
    \overset{(b)}{=}{}&
    \allocLow\cdot (1-o(1))\frac{\sqrt{\pi}}{2}
    +
    \allocHigh\cdot \frac{1}{2\ln\ln n}
    \\
    \overset{(c)}{=} {}&
    \left(\allocLow\cdot (1 + o(1)) + \allocHigh\cdot o(1)\right) \cdot \expect[\val\sim \buyerdist]{\val},
\end{align*}
where inequality~(a) holds since buyers' hazard rate function $\buyerhazardrate$ is decreasing, and the buyer's interim allocation rule is weakly increasing, equality~(b) holds since hazard rate function $\buyerhazardrate(\val) = 2\val$, and inequality~(c) holds since $\expect[\val\sim\buyerdist]{\val} = (1 - o(1))\cdot \frac{\sqrt{\pi}}{2}$ by construction.

Next we provide a lower bound on the optimal Nash social welfare, which helps us pin down the value range of $\allocLow$ and $\allocHigh$. Consider mechanism $\mech\primed$ that implements the Myerson auction with probability $1/(n + 1)$, and gives the item to a uniformly random buyer for free with probability $n/(n + 1)$ (i.e., each buyer $i\in[n]$ receives the item for free with probability $1/(n+1)$). The Nash social welfare of mechanism $\mech\primed$ can be computed as
\begin{align*}
    \sellerexanteutil[\buyerdist]{\mech\primed}
    \cdot 
    \prod\nolimits_{i\in[n]}\buyerexanteutil[i,\buyerdist]{\mech\primed}
    &=
    \frac{1}{n+1}\cdot (1-o(1))\cdot \constantH \cdot 
     \prod\nolimits_{i\in[n]}\frac{1}{n + 1}\cdot \expect[\val\sim\buyerdist]{\val}
     \\
    &=
    \frac{1}{(n+1)^{(n + 1)}} \cdot (1 - o(1))\cdot \constantH\cdot \left(\expect[\val\sim\buyerdist]{\val}\right)^n.
\end{align*}
Since the {\Nashsolution} maximizes the Nash social welfare,
the following inequality holds:
\begin{align*}
    \sellerexanteutil[\buyerdist]{\mech}
    \cdot 
    \prod\nolimits_{i\in[n]}\buyerexanteutil[i,\buyerdist]{\mech}
    \geq 
    \sellerexanteutil[\buyerdist]{\mech\primed}
    \cdot 
    \prod\nolimits_{i\in[n]}\buyerexanteutil[i,\buyerdist]{\mech\primed}.
\end{align*}
Combining the revenue and buyers' surplus upper bounds above, it implies
\begin{align*}
&\left(\allocLow\cdot o(1) + \allocHigh\right) \cdot \constantH
\cdot 
\left(
\frac{1}{n}\cdot \left(\allocLow\cdot (1 + o(1)) + \allocHigh\cdot o(1)\right) \cdot \expect[\val\sim \buyerdist]{\val}\right)^n
\\
&\qquad\qquad\qquad\qquad\qquad\qquad\qquad\geq
    \frac{1}{(n+1)^{(n + 1)}} \cdot (1 - o(1))\cdot \constantH\cdot \left(\expect[\val\sim\buyerdist]{\val}\right)^n.
\end{align*}
Invoking $\allocLow + \allocHigh\leq 1$, we obtain the following inequality on $\allocHigh$:
\begin{align*}
    \allocHigh\cdot (1-\allocHigh)^n\geq \frac{(1-o(1))\cdot n^n}{(n+1)^{(n+1)}}
    \;\;
    \Longrightarrow
    \;\;
    \allocHigh\leq 1 - \frac{(1-o(1))\cdot n}{(n+1)^{\frac{n+1}{n}}}.
\end{align*}
Therefore, the social welfare of the {\Nashsolution} $\mech$ can be upper bounded as follows:
\begin{align*}
    \GFT[\buyerdist]{\mech} 
    &\overset{(a)}{=}
    \sellerexanteutil[\buyerdist]{\mech}
    +
    \buyerexanteutil[\buyerdist]{\mech}
    \overset{(b)}{\leq}
    {\left(\allocLow\cdot o(1) + \allocHigh\right) \cdot \constantH+\left(\allocLow\cdot (1 + o(1)) + \allocHigh\cdot o(1)\right) \cdot \expect[\val\sim \buyerdist]{\val}}
    \\
    &\overset{(c)}{\leq} 
    \left(1 - \frac{(1-o(1))\cdot n}{(n+1)^{\frac{n+1}{n}}}\right)
    \cdot (1+o(1))\cdot \constantH
    =o(1)\cdot \constantH = o(1)\cdot \OPT_\buyerdist,
\end{align*}
where equality~(a) holds by the welfare definition, equality~(b) holds due to the seller revenue and buyers surplus upper bounds derived above, and inequality~(c) holds due to the upper bound of $\allocHigh$ derived above and the fact that $\constantH\gg \expect[\val\sim \buyerdist]{\val}$. This finishes the analysis of \Cref{example:multi-buyer:mhr hard instance} and concludes the proof of \Cref{thm:Nash solution welfare approximation mhr}.
\end{proof}

\subsection{Proof of \texorpdfstring{\Cref{thm:Nash solution welfare approximation reg anti-mhr}}{Theorem 5.2}}
\thmNashsolutionRegAntiMHR*
\begin{proof}
Given any valuation distribution $\buyerdist$ that is both regular and anti-MHR, we show that the social welfare of the {\Nashsolution} (denoted by $\mech$) is at least $\frac{n-1}{n}$ of the optimal social welfare. 
By the definition of the social welfare, it suffices to prove that 
\begin{align*}
    \buyerexanteutil[\buyerdist]{\mech} + \sellerexanteutil[\buyerdist]{\mech}
    \geq 
    \frac{n-1}{n}\cdot \left(\buyerbenchmarkDist+\sellerbenchmarkDist\right).
\end{align*}
To prove this inequality, we first establish a lower bound on the optimal Nash social welfare, which helps us to pin down a relation between $\buyerexanteutil[\buyerdist]{\mech}$, $\sellerexanteutil[\buyerdist]{\mech}$ and $\buyerbenchmarkDist$, $\sellerbenchmarkDist$. Consider the second-price auction, denoted by $\mech\primed$. Since distribution $\buyerdist$ is regular and anti-MHR, invoking \Cref{lem:SPA:revenue approx regular,lem:SPA:buyer surplus antiMHR}, we obtain
\begin{align*}
    \buyerexanteutil[\buyerdist]{\mech\primed} = \buyerbenchmarkDist
    \;\;
    \mbox{and}
    \;\;
    \sellerexanteutil[\buyerdist]{\mech\primed} \geq \frac{n - 1}{n}\cdot \sellerbenchmarkDist.
\end{align*}
Thus, the optimal Nash social welfare (achieve by {\Nashsolution}) must satisfy
\begin{align*}
    \sellerexanteutil[\buyerdist]{\mech}
    \cdot 
    \left(\frac{\buyerexanteutil[\buyerdist]{\mech}}{n}\right)^n
    &\geq 
    \sellerexanteutil[\buyerdist]{\mech\primed}
    \cdot 
    \left(\frac{\buyerexanteutil[\buyerdist]{\mech\primed}}{n}\right)^n
    \geq 
    \frac{n-1}{n}\cdot \sellerbenchmarkDist\cdot 
    \left(\frac{\buyerbenchmarkDist}{n}\right)^n.
\end{align*}
Introduce auxiliary notations $\varx,\varyy,\varA,\varB$ defined as
\begin{align*}
    \varx\triangleq \sellerexanteutil[\buyerdist]{\mech}~,
    \;\;
    \varyy\triangleq \frac{\buyerexanteutil[\buyerdist]{\mech}}{n}~,
    \;\;
    \varA\triangleq \sellerbenchmarkDist~,
    \;\;
    \mbox{and}
    \;\;
    \varB\triangleq \frac{\buyerbenchmarkDist}{n}~.
\end{align*}
We can formulate the the problem of identifying the worst-case welfare approximation ratio as the following minimization program:
\begin{align*}
\min_{\substack{\varA,\varB>0\\ 0<\varx\le \varA,\;0<\varyy\le \varB}}
\frac{\varx+n\varyy}{\varA+n\varB}
\quad\text{subject to}\quad
\varx \varyy^{n} \;\ge\; \frac{n-1}{n}\,\varA \varB^{n}.
\end{align*}
In the remaining analysis, we show the optimal objective value of this minimization program is $(n-1)/n$, which suffices to prove \Cref{thm:Nash solution welfare approximation reg anti-mhr}.

\xhdr{Step 1: Reduction to boundary solutions in $\varx,\varyy$.}
Fix $\varA,\varB>0$. Since the objective $\varx+n\varyy$ is strictly increasing in both
$\varx$ and $\varyy$, the constraint must bind at optimality:
\begin{align*}
\varx \varyy^{n} = \frac{n-1}{n}\,\varA \varB^{n}.
\end{align*}
Moreover, under the upper bounds $\varx\le\varA$ and $\varyy\le\varB$, the optimizer must lie
on one of the two boundary faces:
\begin{align*}
(\varx,\varyy)=\Big(\tfrac{n-1}{n}\varA,\;\varB\Big)
\quad\text{or}\quad
(\varx,\varyy)=\Big(\varA,\;\varB(\tfrac{n-1}{n})^{1/n}\Big).
\end{align*}
Thus, after minimizing over $\varx,\varyy$, the objective reduces to
\begin{align*}
\min\!\left\{
\frac{\frac{n-1}{n}\varA+n\varB}{\varA+n\varB},
\;
\frac{\varA+n\varB(\frac{n-1}{n})^{1/n}}{\varA+n\varB}
\right\}.
\end{align*}

\xhdr{Step 2: Scale invariance.}
For any $t>0$, replacing $(\varA,\varB)$ by $(t\varA,t\varB)$ leaves both expressions unchanged.
Hence the problem depends only on the ratio $r:=\varA/\varB$.
Without loss of generality, set $\varB=1$ and $\varA=r>0$.
Define
\begin{align*}
\phi_1(r)=\frac{\frac{n-1}{n}r+n}{r+n},
\qquad
\phi_2(r)=\frac{r+n(\frac{n-1}{n})^{1/n}}{r+n},
\end{align*}
and let $\phi(r)=\min\{\phi_1(r),\phi_2(r)\}$.

\xhdr{Step 3: Worst-case over $r$.}
We examine the boundary behavior:
\begin{align*}
\lim_{r\to\infty}\phi_1(r)=\frac{n-1}{n},
\qquad
\lim_{r\to 0^+}\phi_2(r)=\left(\frac{n-1}{n}\right)^{1/n}.
\end{align*}
Since for every $n\ge1$, we have
\begin{align*}
\left(\frac{n-1}{n}\right)^{1/n} \geq \frac{n-1}{n}.
\end{align*}
Therefore,
\begin{align*}
\inf_{r>0}\phi(r)=\frac{n-1}{n}.
\end{align*}

\xhdr{Step 4: Conclusion.}
Combining the above steps, we conclude that
\begin{align*}
\inf_{\substack{\varA,\varB>0\\ 0<\varx\le \varA,\;0<\varyy\le \varB}}
\frac{\varx+n\varyy}{\varA+n\varB}
=
\frac{n-1}{n}.
\end{align*}
The infimum is not attained, but is approached asymptotically as
$\varA/\varB\to\infty$ with
$(\varx,\varyy)=(\tfrac{n-1}{n}\varA,\varB)$.
\end{proof}

\subsection{Proof of \texorpdfstring{\Cref{thm:cross-side Nash solution welfare approximation}}{Theorem 5.3}}
\thmCSNashsolution*

\begin{proof}
We now prove all the statements sequentially.

\xhdr{General distributions.} 
The analysis is similar to the proof in \citet{BFM-25} for the single-seller single-buyer bilateral trade model.

Given any valuation distribution $\buyerdist$, we show that the social welfare of the {\CSNashsolution} (denoted by $\mech$) is at least $\frac{1}{2}$ of the optimal social welfare. It suffices to prove that 
\begin{align*}
    \sellerexanteutil[\buyerdist]{\mech} \geq \frac{1}{2}\cdot \sellerbenchmarkDist
    \;\;
    \mbox{and}
    \;\;
    \buyerexanteutil[\buyerdist]{\mech} \geq \frac{1}{2}\cdot \buyerbenchmarkDist.
\end{align*}
We prove above inequalities by contradiction.
Suppose that the seller revenue satisfies $\frac{\sellerexanteutil[\buyerdist]{\mech}}{\sellerbenchmarkDist} < \frac{1}{2}$. (The case for the buyers' surplus follows by a symmetric argument.) Define $\revratio \triangleq \frac{\sellerexanteutil[\buyerdist]{\mech}}{\sellerbenchmark} < \frac{1}{2}$.

Consider a new mechanism $\mech\primed$ that runs the Myerson auction $\Mye$ that maximizes the seller revenue with probability 
$\frac{1}{2} - \frac{\revratio}{2 - 2\revratio}$, 
and runs $\mech$ with the remaining probability 
$\frac{1}{2} + \frac{\revratio}{2 - 2\revratio}$.  
Since $\revratio < \frac{1}{2}$, both probabilities lie in $[0,1]$, so $\mech\primed$ is well-defined.
By construction, the seller revenue and buyers' surplus under $\mech\primed$ satisfy
\begin{align*}
    \sellerexanteutil[\buyerdist]{\mech\primed} 
    & = \left(\frac{1}{2} - \frac{\revratio}{2 - 2\revratio}\right) \cdot \sellerbenchmarkDist 
      + \left(\frac{1}{2} + \frac{\revratio}{2 - 2\revratio}\right) \cdot \sellerexanteutil[\buyerdist]{\mech} \\
    & \overset{(a)}{=} 
      \left(\frac{1}{2} - \frac{\revratio}{2 - 2\revratio}\right) \cdot \frac{1}{\revratio} \cdot \sellerexanteutil[\buyerdist]{\mech} 
      + \left(\frac{1}{2} + \frac{\revratio}{2 - 2\revratio}\right) \cdot \sellerexanteutil[\buyerdist]{\mech}, \\
    \buyerexanteutil[\buyerdist]{\mech\primed} 
    & \geq \left(\frac{1}{2} + \frac{\revratio}{2 - 2\revratio}\right) \cdot \buyerexanteutil[\buyerdist]{\mech},
\end{align*}
where equality~(a) follows from the definition $\revratio = \frac{\sellerexanteutil[\buyerdist]{\mech}}{\sellerbenchmarkDist}$.
Hence, the Nash social welfare of $\mech\primed$ is therefore at least
\begin{align*}
    &\left[
        \left(\frac{1}{2} - \frac{\revratio}{2 - 2\revratio}\right) \cdot \frac{1}{\revratio}
        + \left(\frac{1}{2} + \frac{\revratio}{2 - 2\revratio}\right)
    \right]
    \cdot \sellerexanteutil[\buyerdist]{\mech}
    \cdot 
    \left(\frac{1}{2} + \frac{\revratio}{2 - 2\revratio}\right) \cdot \buyerexanteutil[\buyerdist]{\mech} \\
    &
    \qquad\qquad\qquad\qquad\qquad\qquad\qquad
    = \frac{1}{4(1 - \revratio)\revratio} \cdot \sellerexanteutil[\buyerdist]{\mech} \cdot \buyerexanteutil[\buyerdist]{\mech} 
    > \sellerexanteutil[\buyerdist]{\mech} \cdot \buyerexanteutil[\buyerdist]{\mech},
\end{align*}
where the equality follows by direct algebra and the strict inequality holds because $\revratio < \frac{1}{2}$ implies $\frac{1}{4(1 - \revratio)\revratio} > 1$.

Thus, $\mech\primed$ achieves strictly higher Nash social welfare than $\mech$, contradicting the assumption that $\mech$ is a {\Nashsolution}. This completes the proof of the welfare approximation of $1/2$ for general valuation distributions.

\xhdr{Regular and anti-MHR distributions.}
The analysis is similar to the one used for \Cref{thm:Nash solution welfare approximation reg anti-mhr}. 
Given any valuation distribution $\buyerdist$ that is both regular and anti-MHR, we show that the social welfare of the {\CSNashsolution} (denoted by $\mech$) is at least $\frac{n-1}{n}$ of the optimal social welfare. 
By the definition of the social welfare, it suffices to prove that 
\begin{align*}
    \buyerexanteutil[\buyerdist]{\mech} + \sellerexanteutil[\buyerdist]{\mech}
    \geq 
    \frac{n-1}{n}\cdot \left(\buyerbenchmarkDist+\sellerbenchmarkDist\right).
\end{align*}
To prove this inequality, we first establish a lower bound on the optimal cross-side Nash social welfare, which helps us to pin down a relation between $\buyerexanteutil[\buyerdist]{\mech}$, $\sellerexanteutil[\buyerdist]{\mech}$ and $\buyerbenchmarkDist$, $\sellerbenchmarkDist$. Consider the second-price auction, denoted by $\mech\primed$. Since distribution $\buyerdist$ is regular and anti-MHR, invoking \Cref{lem:SPA:revenue approx regular,lem:SPA:buyer surplus antiMHR}, we obtain
\begin{align*}
    \buyerexanteutil[\buyerdist]{\mech\primed} = \buyerbenchmarkDist
    \;\;
    \mbox{and}
    \;\;
    \sellerexanteutil[\buyerdist]{\mech\primed} \geq \frac{n - 1}{n}\cdot \sellerbenchmarkDist.
\end{align*}
Thus, the optimal cross-side Nash social welfare (achieve by {\CSNashsolution}) must satisfy
\begin{align*}
    \sellerexanteutil[\buyerdist]{\mech}
    \cdot 
    \buyerexanteutil[\buyerdist]{\mech}
    &\geq 
    \sellerexanteutil[\buyerdist]{\mech\primed}
    \cdot 
    \buyerexanteutil[\buyerdist]{\mech\primed}
    \geq 
    \frac{n-1}{n}\cdot \sellerbenchmarkDist\cdot 
    \buyerbenchmarkDist.
\end{align*}
Introduce auxiliary notations $\varx,\varyy,\varA,\varB$ defined as
\begin{align*}
    \varx\triangleq \sellerexanteutil[\buyerdist]{\mech}~,
    \;\;
    \varyy\triangleq \buyerexanteutil[\buyerdist]{\mech}~,
    \;\;
    \varA\triangleq \sellerbenchmarkDist~,
    \;\;
    \mbox{and}
    \;\;
    \varB\triangleq \buyerbenchmarkDist~.
\end{align*}
We can formulate the the problem of identifying the worst-case welfare approximation ratio as the following minimization program:
\begin{align*}
\min_{\substack{\varA,\varB>0\\ 0<\varx\le \varA,\;0<\varyy\le \varB}}
\frac{\varx+\varyy}{\varA+\varB}
\quad\text{subject to}\quad
\varx \varyy \;\ge\; \frac{n-1}{n}\,\varA \varB.
\end{align*}
Due to the symmetric in $(\varx, \varA)$ and $(\varyy,\varB)$, it can be verified that the optimal objective value of this minimization program is $(n-1)/n$, which is attended by $\varx = \frac{n-1}{n}\cdot\varA$ and $\varyy=\varB$.

\xhdr{Hard instance with a MHR distribution.} The analysis below is similar to the one for \Cref{thm:Nash solution welfare approximation mhr}.
We analyze the welfare guarantee of the {\CSNashsolution} in \Cref{example:multi-buyer:mhr hard instance}. The optimal social welfare  can be computed as 
\begin{align*}
    \OPT_\buyerdist &= \expect{\val_{(1:n)}} = 
    \expect{\val_{(1:n)}|\val_{(1:n)} < \constantH}\cdot \prob{\val_{(1:n)}<\constantH}
    +
    \constantH \cdot \prob{\val_{(1:n)} = \constantH}
    \\
    &=
    \expect{\val_{(1:n)}|\val_{(1:n)} < \constantH} \cdot 
    \left(1-\frac{1}{{n}^{1/4}}\right)^n
    +
    \constantH\cdot 
    \left(1 - 
    \left(1-\frac{1}{{n}^{1/4}}\right)^n\right)
    =
    (1 - o(1))\cdot \constantH.
\end{align*}
Next we analyze the social welfare of the {\CSNashsolution}, denoted by $\mech$. Let random variable $\winnerval$ as the value of the winner (i.e., the buyer receiving the item) in the {\CSNashsolution} if there exists one such winner. If no such winner exists, we set $\winnerval = 0$.
Define auxiliary value threshold $\threshold \triangleq \ln\ln n \leq \constantH$ 
and the following two probabilities:
\begin{align*}
    \allocLow\triangleq \prob{\text{winner exists and $\winnerval \leq \threshold$}}
    \;\;
    \mbox{and}
    \;\;
    \allocHigh\triangleq \prob{\text{winner exists and $\winnerval > \threshold$}}.
\end{align*}
Invoking \Cref{prop:revenue equivalence}, the expected revenue can be expressed as 
\begin{align*}
    \sellerexanteutil[\buyerdist]{\mech} &= \expect{\virtualval(\winnerval)\cdot \indicator{\text{winner exists}}} 
    \\
    &=
    \allocLow\cdot \expect{\virtualval(\winnerval)\condition \text{winner exists and $\winnerval \leq \threshold$}}
    +
    \allocHigh\cdot \expect{\virtualval(\winnerval)\condition \text{winner exists and $\winnerval > \threshold$}}
    \\
    &
    \overset{(a)}{\leq} 
    \allocLow\cdot \virtualval(\threshold)
    +
    \allocHigh\cdot \virtualval(\constantH)
    \overset{(b)}{=}
    \allocLow\cdot \left(\ln\ln n - \frac{1}{2\ln\ln n}\right)
    +
    \allocHigh\cdot \frac{1}{2}\sqrt{\ln n}
    =
    \left(\allocLow\cdot o(1) + \allocHigh\right) \cdot \constantH,
\end{align*}
where inequality~(a) holds since virtual value function $\virtualval(\cdot)$ is increasing, and equality~(b) holds since virtual values $\virtualval(\threshold) = \threshold - \frac{1}{2\threshold} = \ln\ln n - \frac{1}{2\ln\ln n}$ and $\virtualval(\constantH) = \constantH = \frac{1}{2}\sqrt{\ln n}$ by construction.

Similarly, invoking \Cref{prop:revenue equivalence},  buyers' surplus can be expressed as 
    \begin{align*}
    &\buyerexanteutil[\buyerdist]{\mech} 
    \\
    ={}& \expect{1/\buyerhazardrate(\winnerval)} 
    \\
     ={}&
    \allocLow\cdot \expect{1/\buyerhazardrate(\winnerval)\condition \text{winner exists and $\winnerval \leq \threshold$}}
    +
    \allocHigh\cdot \expect{1/\buyerhazardrate(\winnerval)\condition \text{winner exists and $\winnerval > \threshold$}}
    \\
    \overset{(a)}{\leq}{} &  
    \allocLow\cdot \expect[\val\sim\buyerdist]{\buyerhazardrate(\val)\condition\val \leq \threshold}
    +
    \allocHigh\cdot \buyerhazardrate(\threshold)
    \\
    \overset{(b)}{=}{}&
    \allocLow\cdot (1-o(1))\frac{\sqrt{\pi}}{2}
    +
    \allocHigh\cdot \frac{1}{2\ln\ln n}
    \\
    \overset{(c)}{=} {}&
    \left(\allocLow\cdot (1 + o(1)) + \allocHigh\cdot o(1)\right) \cdot \expect[\val\sim \buyerdist]{\val},
\end{align*}
where inequality~(a) holds since buyers' hazard rate function $\buyerhazardrate$ is decreasing, and the buyer's interim allocation rule is weakly increasing, equality~(b) holds since hazard rate function $\buyerhazardrate(\val) = 2\val$, and inequality~(c) holds since $\expect[\val\sim\buyerdist]{\val} = (1 - o(1))\cdot \frac{\sqrt{\pi}}{2}$ by construction.

Next we provide a lower bound on the optimal cross-side Nash social welfare, which helps us pin down the value range of $\allocLow$ and $\allocHigh$. Consider mechanism $\mech\primed$ that implements the Myerson auction with probability $1/2$, and gives the item to a uniformly random buyer for free with probability $1/2$ (i.e., each buyer $i\in[n]$ receives the item for free with probability $1/2n$). The cross-side Nash social welfare of mechanism $\mech\primed$ can be computed as
\begin{align*}
    \sellerexanteutil[\buyerdist]{\mech\primed}
    \cdot 
    \buyerexanteutil[\buyerdist]{\mech\primed}
    &=
    \frac{1}{2}\cdot (1-o(1))\cdot \constantH \cdot 
    \frac{1}{2}\cdot \expect[\val\sim\buyerdist]{\val}
     \\
    &=
    \frac{1}{4} \cdot (1 - o(1))\cdot \constantH\cdot \left(\expect[\val\sim\buyerdist]{\val}\right)^n.
\end{align*}
Since the {\CSNashsolution} maximizes the cross-side Nash social welfare,
the following inequality holds:
\begin{align*}
    \sellerexanteutil[\buyerdist]{\mech}
    \cdot 
    \buyerexanteutil[\buyerdist]{\mech}
    \geq 
    \sellerexanteutil[\buyerdist]{\mech\primed}
    \cdot 
    \buyerexanteutil[\buyerdist]{\mech\primed}.
\end{align*}
Combining the revenue and buyers' surplus upper bounds above, it implies
\begin{align*}
&\left(\allocLow\cdot o(1) + \allocHigh\right) \cdot \constantH
\cdot 
\left(
\left(\allocLow\cdot (1 + o(1)) + \allocHigh\cdot o(1)\right) \cdot \expect[\val\sim \buyerdist]{\val}\right)
\geq
    \frac{1}{4} \cdot (1 - o(1))\cdot \constantH\cdot \left(\expect[\val\sim\buyerdist]{\val}\right)^n.
\end{align*}
Invoking $\allocLow + \allocHigh\leq 1$, we obtain the following inequality on $\allocHigh$:
\begin{align*}
    \allocHigh\cdot (1-\allocHigh)\geq \frac{1-o(1)}{4}
    \;\;
    \Longrightarrow
    \;\;
    \allocHigh\leq \frac{1}{2} + o(1).
\end{align*}
Therefore, the social welfare of the {\Nashsolution} $\mech$ can be upper bounded as follows:
\begin{align*}
    \GFT[\buyerdist]{\mech} 
    &\overset{(a)}{=}
    \sellerexanteutil[\buyerdist]{\mech}
    +
    \buyerexanteutil[\buyerdist]{\mech}
    \overset{(b)}{\leq}
    {\left(\allocLow\cdot o(1) + \allocHigh\right) \cdot \constantH+\left(\allocLow\cdot (1 + o(1)) + \allocHigh\cdot o(1)\right) \cdot \expect[\val\sim \buyerdist]{\val}}
    \\
    &\overset{(c)}{\leq} 
    \left(\frac{1}{2} + o(1)\right)
    \cdot (1+o(1))\cdot \constantH
    = \left(\frac{1}{2} + o(1)\right)\cdot \OPT_\buyerdist,
\end{align*}
where equality~(a) holds by the welfare definition, equality~(b) holds due to the seller revenue and buyers surplus upper bounds derived above, and inequality~(c) holds due to the upper bound of $\allocHigh$ derived above and the fact that $\constantH\gg \expect[\val\sim \buyerdist]{\val}$. This finishes the analysis of \Cref{example:multi-buyer:mhr hard instance} and concludes the proof of \Cref{thm:cross-side Nash solution welfare approximation}.
\end{proof}

\subsection{Improved Welfare Approximation for Bilateral Trade Model}

In this section, we consider the single-seller single-buyer bilateral trade model. The seller has a private cost $\cost\sim\sellerdist$ and the buyer has a private valuation $\val\sim\buyerdist$. We say the seller cost distribution $\sellerdist$ is MHR if $\sellerdist(\cost)/\sellerpdf(\cost)$ is weakly decreasing, and the buyer valuation distribution $\buyerdist$ is MHR if $(1-\buyerdist(\val))/\buyerdist(\cost)$ is weakly decreasing. We define {\SellerOffer} ({\SOM}) and {\BuyerOffer} ({\BOM}) as the seller-optimal mechanism and the buyer-optimal mechanism, respectively.

\begin{theorem}
\label{thm:Nash solution welfare approximation single-buyer mhr}
    Consider a single-seller single-buyer bilateral trade market. Suppose  both traders' distributions are MHR. Then the {\Nashsolution} achieves at least $1/(e-1)$ fraction of $\sellerbenchmark+\buyerbenchmark$.\footnote{We remark that $\sellerbenchmark+\buyerbenchmark$ is an upper bound on the standard second best GFT benchmark \citep[cf.][]{BCGZ-18,DMSW-22}.}
\end{theorem}
\begin{proof}
    We denote by $\mech$ as the {\Nashsolution}.
    Our argument studies two cases separately.

    \xhdr{Case 1.} Suppose the buyer's ex-ante utility $\buyerexanteutil{\mech}$ in {\Nashsolution} $\mech$ is strictly smaller than her ex-ante utility $\buyerexanteutil{\SOM}$ in the {\SellerOffer}, i.e., $\buyerexanteutil{\mech} < \buyerexanteutil{\SOM}$. In this case, we claim that the {\SellerOffer} has strictly higher Nash social welfare than mechanism $\mech$, since $\buyerexanteutil{\mech} < \buyerexanteutil{\SOM}$ (case assumption) and $\sellerexanteutil{\mech} \leq \sellerexanteutil{\SOM}$ (guarantee of the {\SellerOffer}). This is a contradiction, and thus this case never happens.

    \xhdr{Case 2.} Suppose the buyer's ex-ante utility $\buyerexanteutil{\mech}$ in {\Nashsolution} $\mech$ is weakly higher than her ex-ante utility $\buyerexanteutil{\SOM}$ in the {\SellerOffer}, i.e., $\buyerexanteutil{\mech} \geq \buyerexanteutil{\SOM}$. Define probability $\mixprob$ as 
    \begin{align*}
        \mixprob\triangleq \frac{\buyerexanteutil{\mech} - \buyerexanteutil{\SOM}}{\buyerexanteutil{\BOM} - \buyerexanteutil{\SOM}}.
    \end{align*}
    Since $\buyerexanteutil{\mech} \geq \buyerexanteutil{\SOM}$ (case assumption) and $\buyerexanteutil{\mech} \leq \buyerexanteutil{\BOM}$ (guarantee of the {\BuyerOffer}), probability $\mixprob$ is well-defined (i.e., $\mixprob \in [0, 1]$). Now consider the mechanism $\mech\primed$ that runs the {\BuyerOffer} with probability $\mixprob$ and the {\SellerOffer} with probability $1 - \mixprob$. By construction, the buyer's ex-ante utility is identical in both mechanisms $\mech$ and $\mech\primed$, i.e., 
    \begin{align*}
        \buyerexanteutil{\mech} = \buyerexanteutil{\mech\primed}.
    \end{align*}
    Moreover, the seller's ex-ante utility is weakly higher in {\Nashsolution} $\mech$ than the mechanism $\mech\primed$. (Otherwise, the latter mechanism $\mech\primed$ has a strictly higher Nash social welfare than the former mechanism $\mech$, a contradiction.) Thus, 
    \begin{align*}
        \GFT{\mech} &= \sellerexanteutil{\mech} + \buyerexanteutil{\mech} \geq 
        \sellerexanteutil{\mech\primed} + \buyerexanteutil{\mech\primed} \overset{(a)}{=}
        \mixprob\cdot \GFT{\BOM} + (1 - \mixprob) \cdot \GFT{\SOM}
        \\
        & \overset{(b)}{\geq}
        \frac{1}{e - 1}\cdot \OPTFB \geq 
        \frac{1}{e - 1}\cdot \OPTSB,
    \end{align*}  
    where equality~(a) holds due to the construction of the mechanism $\mech\primed$, and inequality~(b) holds since both players are MHR and \citet[Theorem~3.1]{Fei-22}. This finishes the proof of the statement as desired.
\end{proof}

\section{Revenue Approximation of Price-Posting Mechanisms in a Single Buyer Market}
\label{sec:revenue approximation with single buyer}
As we have already shown, for the problem of maximizing seller revenue under a buyer surplus constraint, in a single-buyer market, a price-posting mechanism is optimal when the value distribution is regular, and in general, a randomization between two price-posting mechanisms is optimal. 
In this section, we study how much revenue a price-posting mechanism can obtain in general. 
We start by defining the revenue approximation ratio. 
\begin{definition}[Revenue approximation ratio]
    The revenue approximation ratio of a class of mechanisms $\mechFam\primed \subset \mechFam$ is defined as 
    \[\inf_{\dist, \buyerutil} \frac{\sup_{\mech\primed \in \mechFam\primed, \buyerexanteutil[\dist]{\mech\primed} \geq \buyerutil} \sellerexanteutil[\dist]{\mech}}{\sup_{\mech \in \mechFam, \buyerexanteutil[\dist]{\mech} \geq \buyerutil} \sellerexanteutil[\dist]{\mech}}.\]
\end{definition}

We also introduce the concept of quasi-regular value definitions from \cite{FJ-24}, as follows: 
\begin{definition}[Quasi-regularity, from \cite{FJ-24}]
    A value distribution with the revenue curve $\revCurve$ is \emph{quasi-regular}, if for any $\quant < \quant' < 1$, 
    \[\frac{\revCurve(\quant')}{1 - \quant'} \geq \frac{\revCurve(\quant)}{1 - \quant}. \]
\end{definition}
Intuitively speaking, a value distribution is quasi-regular, if for any $\quant < \quant'$, $(\quant', \revCurve(\quant'))$ is not below the line between $(\quant, \revCurve(\quant))$ and $(1, \revCurve(1) = 0)$. A direct implication is that the family of regular distributions is a strict subset of the family of quasi-regular distributions. 

We have the following result. 
\begin{theorem} \label{thm:revenue approximation for single buyer}
    In a single-item seller, single-buyer market, the revenue approximation ratio of price-posting mechanisms is 
    \begin{itemize}
        \item 0, for general value distributions; 
        \item between {\revApproxL} and {\revApproxU}, for quasi-regular value distributions; 
        \item 1, for regular value distributions. 
    \end{itemize}
\end{theorem}

In particular, we notice that the last statement above is a corollary of \Cref{prop:Pareto frontier:single buyer general}. The rest of this section is devoted to proving \Cref{thm:revenue approximation for single buyer}. For simplicity, we slightly abuse notation, and let $\buyerexanteutil[\dist]{\quant}$ represent the buyers' surplus of the mechanism of posting a single price $\valFunc(\quant)$. 

\subsection{General Value Distributions}

For general value distributions, we proceed by showing that for any $\smallEps > 0$, there exists $\dist, \buyerutil$ such that 
\[\frac{\sup_{\mech\primed \in \mechFam\primed, \buyerexanteutil[\dist]{\mech\primed} \geq \buyerutil} \sellerexanteutil[\dist]{\mech}}{\sup_{\mech \in \mechFam, \buyerexanteutil[\dist]{\mech} \geq \buyerutil} \sellerexanteutil[\dist]{\mech}} \leq \smallEps.\]
In the above, $\mechFam\primed$ stands for the family of price-posting mechanisms.

Now, let $\largeN > \exp(\nicefrac{2}{\smallEps})$, and consider the following $\dist$ such that 
\[ \valFunc(\quant) = 
\begin{cases}
    \largeN^2, \quad & 0 \leq \quant \leq \nicefrac{1}{2\largeN}; \\
    \nicefrac{\largeN}{2}, \quad & \nicefrac{1}{2\largeN} < \quant \leq \nicefrac{1}{\largeN}; \\
    \nicefrac{1}{2\quant}, \quad & \nicefrac{1}{\largeN} < \quant \leq \nicefrac{1}{2}; \\
    0, \quad & \nicefrac{1}{2} < \quant \leq 1. 
\end{cases}
\]
Meanwhile, let $\buyerutil = \nicefrac{\largeN}{2}$. \footnote{Although the value distribution $\dist$ we present here is not differentiable within the interior of its support, it can be smoothed via standard techniques to satisfy this condition without causing loss to our result.}

For this problem instance, notice that there exists a $\dualVOPT$ such that $\lagCurve[\dualVOPT](\nicefrac{1}{2\largeN}) = \lagCurve[\dualVOPT](\nicefrac{1}{2}) = \largeN / 2$, and $\quantItvOPT(\dualVOPT) = [\nicefrac{1}{2\largeN}, \nicefrac{1}{2}]$. 
Now one can compute that $\buyerexanteutil[\dist]{\nicefrac{1}{2\largeN}} = 0$, and $\buyerexanteutil[\dist]{\nicefrac{1}{2}} = \nicefrac{(\largeN + \ln \largeN - \ln 2)}{2}$. 
Thus, $\buyerexanteutil[\dist]{\leftS(\dualVOPT)} \leq \buyerutil \leq \buyerexanteutil[\dist]{\rightS(\dualVOPT)}$. 
According to \Cref{coro:proofs in Pareto frontier:single buyer full characterization}, the optimal mechanism for this problem instance is a randomization of setting price $\valFunc(\nicefrac{1}{2\largeN}) = \largeN^2$ with probability $\interpol$ and $\valFunc(\nicefrac{1}{2}) = 1$ with probability $1 - \interpol$, satisfying
\begin{align*}
    \interpol \cdot 0 + (1 - \interpol)\cdot \frac{\largeN + \ln \largeN - \ln 2}{2} = \buyerutil = \frac{\largeN}{2} \Longleftrightarrow \interpol = \frac{\ln \largeN - \ln 2}{\largeN + \ln \largeN - \ln 2}. 
\end{align*}
Therefore, the optimal seller revenue is 
\begin{align*}
    \interpol \cdot \frac{\largeN}{2} + (1 - \interpol)\cdot \frac{1}{2} = \frac{\largeN \ln \largeN + \largeN - \largeN \ln 2}{2(\largeN + \ln \largeN - \ln 2)}. 
\end{align*}

Meanwhile, since $\buyerexanteutil[\dist]{\nicefrac{1}{\largeN}} = \nicefrac{\largeN}{2} = \buyerutil$, and $\buyerexanteutil[\dist]{\quant} < \nicefrac{\largeN}{2}$ for any $\quant < \nicefrac{1}{\largeN}$, any single-price-posting mechanism satisfying the constraint should post a price with quantile no less than $\nicefrac{1}{\largeN}$, therefore rendering a revenue of at most $\nicefrac{1}{2}$. 
Consequently, the revenue approximation ratio of price-posting mechanisms in this problem instance is upper-bounded by 
\[\frac{\largeN + \ln \largeN - \ln 2}{\largeN \ln \largeN + \largeN - \largeN \ln 2} < \frac{2}{\ln \largeN} < \smallEps. \]
This finishes the proof.

\subsection{Quasi-Regular Value Distributions}

\subsubsection{{\revApproxL} Revenue Approximation}

For any $\dist, \buyerutil$, we let $\quantDag \triangleq \sup\{\quant \suchthat \buyerexanteutil[\dist]{\quant} \leq \buyerutil\}$. Since the $\buyerexanteutil[\dist]{\cdot}$ is a continuous function, $\buyerexanteutil[\dist]{\quantDag} = \buyerutil$. Further, let $\dualVOPT = \sup\{\dualV \suchthat \leftS(\dualV) \leq \quantDag\}$. According to the proof of \Cref{prop:Pareto frontier:single buyer general}, $\buyerexanteutil[\dist]{\leftS(\dualVOPT)} \leq \buyerutil \leq \buyerexanteutil[\dist]{\rightS(\dualVOPT)}$ holds. 
If $\leftS(\dualVOPT) = \rightS(\dualVOPT) = \quantDag$, then posting price $\valFunc(\quantDag)$ is already optimal. Otherwise, write for brevity that $\leftS \triangleq \leftS(\dualVOPT)$ and $\rightS \triangleq \rightS(\dualVOPT)$. 
We prove that the better of posting price $\valFunc(\quantDag)$ and posting price $\valFunc(\rightS)$ gives a {\revApproxL} revenue approximation for the optimal mechanism. 
Notice that since $\rightS \geq \quantDag$, the mechanism of posting price $\valFunc(\rightS)$ satisfies the surplus constraint. 

Let $\interpol$ satisfy $\interpol \cdot \buyerexanteutil[\dist]{\leftS} + (1 - \interpol) \cdot \buyerexanteutil[\dist]{\rightS} = \buyerutil = \buyerexanteutil[\dist]{\quantDag}$. Thus, according to \Cref{coro:proofs in Pareto frontier:single buyer full characterization}, the optimal mechanism would be posting price $\valFunc(\leftS)$ with probability $\interpol$ and price $\valFunc(\rightS)$ with probability $1 - \interpol$, and thus the optimal seller revenue would be $\interpol \cdot \leftS \valFunc(\leftS) + (1 - \interpol) \cdot \rightS \valFunc(\rightS)$. 

\xhdr{Step 1: Bounding the optimal seller revenue with two prices, and the revenue approximation ratio.}
By the definition of quasi-regularity, we know that for any $\leftS \leq \quant \leq \quantDag$, $\valFunc(\quant) \geq (\nicefrac{\leftS \valFunc(\leftS)}{1 - \leftS}) \cdot \nicefrac{(1 - \quant)}{\quant}$. As a result, 
\begin{align*}
    \buyerexanteutil[\dist]{\quantDag} - \buyerexanteutil[\dist]{\leftS} &= \leftS(\valFunc(\leftS) - \valFunc(\quantDag)) + \int_{\leftS}^{\quantDag} (\valFunc(\quant) - \valFunc(\quantDag)) \dd \quant \\
    &\geq \leftS(\valFunc(\leftS) - \valFunc(\quantDag)) + \int_{\leftS}^{\quantDag} \inParentheses{\frac{\leftS \valFunc(\leftS)}{1 - \leftS} \cdot \frac{1 - \quant}{\quant} - \valFunc(\quantDag)} \dd \quant \\
    &= \leftS(\valFunc(\leftS) - \valFunc(\quantDag)) + \frac{\leftS \valFunc(\leftS)}{1 - \leftS}\cdot (\ln \quantDag - \ln \leftS + \leftS - \quantDag) - (\quantDag - \leftS)\cdot \valFunc(\quantDag) \\
    &= \leftS \valFunc(\leftS) - \quantDag \valFunc(\quantDag) + \frac{\leftS \valFunc(\leftS)}{1 - \leftS}\cdot (\ln \quantDag - \ln \leftS + \leftS - \quantDag). 
\end{align*}

Meanwhile, $\buyerexanteutil[\dist]{\rightS} - \buyerexanteutil[\dist]{\quantDag} \leq \rightS\cdot (\valFunc(\quantDag) - \valFunc(\rightS))$. Since $\interpol\cdot (\buyerexanteutil[\dist]{\quantDag} - \buyerexanteutil[\dist]{\leftS}) = (1 - \interpol)\cdot (\buyerexanteutil[\dist]{\rightS} - \buyerexanteutil[\dist]{\quantDag})$, putting the above together induces that
\[\interpol \leq \frac{\rightS\cdot (\valFunc(\quantDag) - \valFunc(\rightS))}{\leftS \valFunc(\leftS) - \quantDag \valFunc(\quantDag) + \frac{\leftS \valFunc(\leftS)}{1 - \leftS}\cdot (\ln \quantDag - \ln \leftS + \leftS - \quantDag) + \rightS\cdot (\valFunc(\quantDag) - \valFunc(\rightS))}. \]
Since $\lagCurve[\dualVOPT](\leftS) = \lagCurve[\dualVOPT](\rightS)$, we have $\leftS \valFunc(\leftS) \geq \rightS \valFunc(\rightS)$. We then have
\begin{align}
    &\mathrel{\phantom{=}} \interpol \cdot \leftS \valFunc(\leftS) + (1 - \interpol) \cdot \rightS \valFunc(\rightS) \notag \\
    &\leq \frac{\rightS\cdot (\valFunc(\quantDag) - \valFunc(\rightS))}{\leftS \valFunc(\leftS) - \quantDag \valFunc(\quantDag) + \frac{\leftS \valFunc(\leftS)}{1 - \leftS}\cdot (\ln \quantDag - \ln \leftS + \leftS - \quantDag) + \rightS\cdot (\valFunc(\quantDag) - \valFunc(\rightS))} \cdot (\leftS \valFunc(\leftS) - \rightS \valFunc(\rightS)) + \rightS \valFunc(\rightS) \notag \\
    &= \frac{\rightS\cdot (\valFunc(\quantDag) - \valFunc(\rightS))}{\leftS \valFunc(\leftS) - \rightS \valFunc(\rightS) + \frac{\leftS \valFunc(\leftS)}{1 - \leftS}\cdot (\ln \quantDag - \ln \leftS + \leftS - \quantDag) + (\rightS - \quantDag)\cdot \valFunc(\quantDag)} \cdot (\leftS \valFunc(\leftS) - \rightS \valFunc(\rightS)) + \rightS \valFunc(\rightS). \label{eq:revenue approximation:quasi regularity:good:1}
\end{align}

Observe that this is an increasing function of $\valFunc(\leftS)$ when all other parameters are fixed. By the definition of quasi-regularity, we have $\leftS \valFunc(\leftS) \leq \nicefrac{(1 - \leftS)}{(1 - \quantDag)} \cdot \quantDag \valFunc(\quantDag)$. The first term of \eqref{eq:revenue approximation:quasi regularity:good:1} is upper bounded by 
\begin{align*}
    &\mathrel{\phantom{=}} \frac{\rightS\cdot (\valFunc(\quantDag) - \valFunc(\rightS))}{\frac{1 - \leftS}{1 - \quantDag} \quantDag \valFunc(\quantDag) - \rightS \valFunc(\rightS) + \frac{\quantDag \valFunc(\quantDag)}{1 - \quantDag}\cdot (\ln \quantDag - \ln \leftS + \leftS - \quantDag) + (\rightS - \quantDag)\cdot \valFunc(\quantDag)} \cdot \inParentheses{\frac{1 - \leftS}{1 - \quantDag} \quantDag \valFunc(\quantDag) - \rightS \valFunc(\rightS)} \\
    &= \frac{\rightS\cdot (\valFunc(\quantDag) - \valFunc(\rightS))\inBrackets{(1 - \leftS)\cdot \quantDag\valFunc(\quantDag) - (1 - \quantDag)\cdot \rightS \valFunc(\rightS)}}{(\ln \quantDag - \ln \leftS)\quantDag \valFunc(\quantDag) + \rightS(1 - \quantDag)(\valFunc(\quantDag) - \valFunc(\rightS))}. 
\end{align*}

Considering the approximation of $\max(\quantDag\valFunc(\quantDag), \rightS\valFunc(\rightS))$ on the above formula plus $\rightS\valFunc(\rightS)$, let $\valRatio \triangleq \nicefrac{\valFunc(\quantDag)}{\valFunc(\rightS)}$, the revenue approximation ratio lower bound then becomes
\begin{align*}
    &\mathrel{\phantom{=}} \left. \max\inParentheses{\nicefrac{\quantDag}{\rightS} \cdot \valRatio, 1} \middle/ \inParentheses{\frac{(\valRatio - 1)\inBrackets{(1 - \leftS) \cdot \nicefrac{\quantDag}{\rightS} \cdot \valRatio - (1 - \quantDag)}}{(\ln \quantDag - \ln \leftS) \cdot \nicefrac{\quantDag}{\rightS} \cdot \valRatio + (1 - \quantDag)(\valRatio - 1)} + 1}\right. \\
    &= \frac{\inBrackets{(\ln \quantDag - \ln \leftS) \cdot \nicefrac{\quantDag}{\rightS} \cdot \valRatio + (1 - \quantDag)(\valRatio - 1)} \max\inParentheses{\nicefrac{\quantDag}{\rightS} \cdot \valRatio, 1}}{\inBrackets{\ln \quantDag - \ln \leftS + (1 - \leftS)(\valRatio - 1)}\cdot \nicefrac{\quantDag}{\rightS} \cdot \valRatio}. 
\end{align*}

\xhdr{Step 2: Solving a program to minimize the lower bound with the rest variables.}
Now, when $\rightS \geq \quantDag \cdot \valRatio$, $\max(\nicefrac{\quantDag}{\rightS} \cdot \valRatio, 1) = 1$, and the above formula reaches the minimum when $\rightS = \quantDag \cdot \valRatio$. On the other hand, when $\rightS \leq \quantDag \cdot \valRatio$, $\max(\nicefrac{\quantDag}{\rightS} \cdot \valRatio, 1) = \nicefrac{\quantDag}{\rightS} \cdot \valRatio$, and the minimum is also obtained when $\rightS = \quantDag \cdot \valRatio$. In all, when $\rightS = \quantDag \cdot \valRatio$, the revenue approximation ratio is lower bounded by 
\begin{align*}
    \frac{\ln \quantDag - \ln \leftS + (1 - \quantDag)(\valRatio - 1)}{\ln \quantDag - \ln \leftS + (1 - \leftS)(\valRatio - 1)}. 
\end{align*}
Further, fix $\leftS$ and $\quantDag$, the above is minimized when $\valRatio \geq 1$ reaches its maximum, or $\rightS = 1$ and $\valRatio = \nicefrac{1}{\quantDag}$. The lower bound is then simplified to  
\begin{align*}
    \frac{\quantDag(\ln \quantDag - \ln \leftS) + (1 - \quantDag)^2}{\quantDag(\ln \quantDag - \ln \leftS) + (1 - \leftS)(1 - \quantDag)}. 
\end{align*}

Finally, for $0 \leq \leftS \leq \quantDag \leq 1$, the global minimum of the above formula is approximately {\revApproxL}, reached when $\leftS \approx 0.238405$ and $\quantDag \approx 0.528482$. 

However, we mention that this lower bound we derive might be loose, because in the definition, we suppose that $\rightS = \rightS(\dualVOPT)$, but this may not be the case under the final instance we derive. This is a main reason for the gap between our approximability result and the following inapproximability result.

\subsubsection{A {\revApproxU} Revenue Inapproximability}

Consider $\dist, \buyerutil$ with parameters $0 \leq \leftS \leq \quantDag \leq 1$ to be decided satisfying: 
\[
\valFunc(\quant) = 
\begin{cases}
    \frac{\quantDag}{1 - \quantDag} \cdot \frac{1 - \leftS}{\leftS}, \quad & 0 \leq \quant \leq \leftS; \\
    \frac{\quantDag}{1 - \quantDag} \cdot \frac{1 - \quant}{\quant}, \quad & \leftS < \quant \leq \quantDag; \\
    \frac{\quantDag}{\quant}, \quad & \quantDag < \quant \leq 1. 
\end{cases}
\]

Further, $\buyerutil = \buyerexanteutil[\dist]{\quantDag}$. \footnote{Similarly, the value distribution $\dist$ we present here can be smoothed without causing loss to our result.}
Notice that in this problem instance, since the value function is equal-revenue on $[\quantDag, 1]$, the optimal seller revenue for price-posting is $\quantDag$, by posting price $\valFunc(\quantDag) = 1$. Meanwhile, if there exists some $\dualV$ such that $\quantItvOPT(\dualV) = [\leftS, 1]$, then the optimal mechanism is a randomization of posting price $\valFunc(\leftS)$ with probability $\interpol$ and price $\valFunc(1)$ with probability $1 - \interpol$, satisfying that $\buyerutil = \interpol \cdot \buyerexanteutil[\dist]{\leftS} + (1 - \interpol) \cdot \buyerexanteutil[\dist]{1}$. We compute that 
\begin{align*}
    \buyerexanteutil[\dist]{\leftS} &= 0; \\
    \buyerexanteutil[\dist]{\quantDag} &= \frac{\quantDag}{1 - \quantDag}\cdot (\ln \quantDag - \ln \leftS); \\
    \buyerexanteutil[\dist]{1} &= \frac{\quantDag}{1 - \quantDag}\cdot (\quantDag \ln \quantDag - \ln \leftS). 
\end{align*}

Therefore, $\interpol = \nicefrac{((\quantDag - 1) \ln \quantDag)}{(\quantDag \ln \quantDag - \ln \leftS)}$, and the optimal seller revenue is
\begin{align*}
    &\mathrel{\phantom{=}} \interpol\cdot \leftS\valFunc(\leftS) + (1 - \interpol)\cdot \valFunc(1) \\
    &= \frac{(\quantDag - 1)\ln \quantDag}{\quantDag \ln \quantDag - \ln \leftS} \cdot \frac{\quantDag}{1 - \quantDag}\cdot (1 - \leftS) + \frac{\ln \quantDag - \ln \leftS}{\quantDag \ln \quantDag - \ln \leftS}\cdot \quantDag \\
    &= \frac{\leftS \ln \quantDag - \ln \leftS}{\quantDag \ln \quantDag - \ln \leftS} \cdot \quantDag. 
\end{align*}
And the revenue approximation ratio is therefore
\[\frac{\quantDag \ln \quantDag - \ln \leftS}{\leftS \ln \quantDag - \ln \leftS}. \]

When $0 \leq \leftS \leq \quantDag \leq 1$, the above formula reaches its minimum of about {\revApproxU} when $\leftS \approx 0.325268$ and $\quantDag \approx 0.589198$. We are left to show the existence of $\dualVOPT$ such that $\quantItvOPT(\dualVOPT) = [\leftS, 1]$. To start with, to let $\lagCurve[\dualVOPT](\leftS) = \lagCurve[\dualVOPT](1)$, we should have
\[\frac{\quantDag}{1 - \quantDag} \cdot (1 - \leftS) = \dualVOPT\cdot \inBrackets{\frac{\quantDag}{1 - \quantDag}\cdot (\quantDag \ln \quantDag - \ln \leftS) + \quantDag} + (1 - \dualVOPT)\cdot \quantDag, \]
or that 
\[\dualVOPT = \frac{\quantDag - \leftS}{\quantDag \ln \quantDag - \ln \leftS}. \]
Further, according to the first-order condition of the optimality of $\leftS$, we also have
\[\left. \partial \frac{\quantDag \ln \quantDag - \ln \leftS}{\leftS \ln \quantDag - \ln \leftS} \middle/ \partial \leftS \right. = 0 \;\Longrightarrow\; \frac{\quantDag - \leftS}{\quantDag \ln \quantDag - \ln \leftS} = \leftS.\]
Thus $\dualVOPT = \leftS$. To prove that $\quantItvOPT(\dualVOPT) = [\leftS, 1]$, we observe that $\lagCurve[\dualVOPT](\cdot)$ is increasing on $[\quantDag, 1]$ since $\valFunc(\cdot)$ is equal-revenue on $[\quantDag, 1]$. It suffices to prove that $\lagCurve[\dualVOPT](\cdot)$ is decreasing on $[\leftS, \quantDag]$. To see this, when $\quant \in [\leftS, \quantDag]$, we obtain that
\begin{align*}
    \lagCurve[\dualVOPT]'(\quant) &= \dualVOPT\cdot \frac{\quantDag}{1 - \quantDag} \cdot \frac{1 - \quant}{\quant} - (1 - \dualVOPT)\cdot \frac{\quantDag}{1 - \quantDag} \\
    &= \dualVOPT\cdot \frac{\quantDag}{1 - \quantDag} \cdot \inParentheses{\frac{1}{\quant} - \frac{1}{\dualVOPT}} \\
    &= \dualVOPT\cdot \frac{\quantDag}{1 - \quantDag} \cdot \inParentheses{\frac{1}{\quant} - \frac{1}{\leftS}} \leq 0.
\end{align*}

Consequently, $\quantItvOPT(\dualVOPT) = [\leftS, 1]$ holds. We conclude that any single price-posting mechanism can at most obtain a {\revApproxU} revenue approximation in this problem instance. 
 
\end{document}